\documentclass[11pt,onecolumn,draftclsnofoot]{IEEEtran}

\renewcommand{\baselinestretch}{1.5}
\usepackage{graphicx}
\usepackage{citesort}
\usepackage[cmex10]{amsmath}
\usepackage{amssymb}
\begin{document}
\title{On the 3-Receiver Broadcast Channel with Degraded Message Sets and Confidential Messages}

\author{Li-Chia~Choo, and~Kai-Kit~Wong,~\IEEEmembership{Senior~Member,~IEEE}
\thanks{The authors are with the Department of Electronic and Electrical Engineering, University College London, Torrington Place, London WC1E 7JE, United Kingdom (email: \{l.choo, kwong\}@ee.ucl.ac.uk).}
\thanks{The material in this paper will be presented in part at the International Conference on Wireless Communications and Signal Processing 2009, Nov. 13-15, Nanjing, China, 2009.}}

\maketitle

\begin{abstract}
In this paper, bounds to the rate-equivocation region for the general 3-receiver broadcast channel (BC) with degraded message sets, are presented for confidential messages to be kept secret from one of the receivers. This model is more general than the 2-receiver BCs with confidential messages with an external wiretapper, and the recently studied 3-receiver degraded BCs with confidential messages, since in the model studied in this paper, the conditions on the receivers are general and the wiretapper receives the common message. Wyner's code partitioning combined with double-binning is used to show the achievable rate tuples. Error probability analysis and equivocation calculation are also provided. The secure coding scheme is sufficient to provide security for the 3-receiver BC with 2 or 3 degraded message sets, for the scenarios: (i) 3 degraded message sets, where the first confidential message is sent to receivers 1 and 2 and the second confidential message is sent to receiver 1, (ii) 2 degraded message sets, where one confidential message is sent to receiver 1, and (iii) 2 degraded message sets, where one confidential message is sent to receivers 1 and 2. The proof for the outer bound is shown for the cases where receiver 1 is more capable than the wiretap receiver 3, for the first two scenarios. Under the condition that both receivers 1 and 2 are less noisy than the wiretap receiver 3, the inner and outer bounds coincide, giving the rate-equivocation region for (iii). In addition, a new outer bound for the general 3-receiver BC with 3 degraded messages is obtained.
\end{abstract}


\section{Introduction}
\IEEEPARstart{W}{ireless} communications channels today are vulnerable to eavesdropping or wiretapping due to the open nature of the channel, making the characterization of transmission rates for secure and reliable communication for the physical layer an important issue. In the wireless broadcast medium, the model of the broadcast channel (BC) with confidential messages, which was studied by Csisz\'ar and K$\ddot{\mathrm{o}}$rner \cite{csiszar_78}, is used to study simultaneously secure and reliable communication. The model in \cite{csiszar_78} is a generalization of the characterization of the wiretap channel by Wyner \cite{wyner_75}. In \cite{csiszar_78}, a common message is sent to 2 receivers, while a confidential message is sent to one of the receivers and kept secret from the other. The secrecy level is determined by the equivocation rate, which is the entropy rate of the confidential message conditioned on the channel output at the eavesdropper or wiretapper. The secrecy capacity region is defined as the set of transmission rates where the legitimate receiver decodes its confidential message while keeping the message secret from the wiretapper.

In more recent studies on the BC with confidential messages, Liu {\em et al.}~\cite{liu_08} studied the scenario where there are 2 receivers and private messages are sent to each one and kept secret from the unintended receiver, while Xu {\em et al.}~\cite{xu_08_submit} looked at the same model in \cite{liu_08} but with a common message to both receivers. Then, Bagherikaram {\em et al.}~\cite{bagherikaram_08_submit} addressed the scenario where there are 2 receivers and one wiretapper, with confidential messages sent to the receivers. There have been recent studies where more than 2 receivers were considered. The authors in \cite{choo_08_submit} and Ekrem and Ulukus in \cite{ekrem_08_submit} independently studied the $K$-receiver BC with an external wiretapper. In \cite{choo_08_submit}, the $K$-receiver BC with confidential messages sent to each receiver was studied, while in \cite{ekrem_08_submit}, the same scenario was studied with the addition that each receiver also received a common message. Both used the degraded BC. In another work, an achievable inner bound for the $K$-receiver BC with a common message sent to all receivers and a confidential message sent to each of the receivers to be kept secret from an external wiretapper was derived by Kobayashi {\em et al.}~in \cite{kobayashi_09} for general conditions on the receivers' and wiretapper's channels. Finally, Chia and El Gamal in \cite{chia_09} derived an achievable inner bound for the 3-receiver BC with a common message sent to all receivers and a private message sent to 2 of the receivers to be kept secret from the third.

Recently in \cite{nair_07}--\cite{nair_09}, Nair and El Gamal introduced the channel model of the 3-receiver BC with degraded message sets. In the general form of this model, a common message $W_0$ is sent to all of the receivers, denoted by the set $\mathbb{R}_{\mathrm{all}}$, and the private messages, $W_i,W_{i-1},\dots,W_1$, are sent to subsets of receivers $\mathbb{R}_i \subset \mathbb{R}_{i-1} \subset \dots \subset \mathbb{R}_1 \subset \mathbb{R}_{\mathrm{all}}$.  This model best describes a multimedia broadcasting system, in which the common message $W_0$ may represent the lowest quality transmission, and $W_1$ the next higher quality transmission, and so on. In \cite{nair_07}--\cite{nair_09}, three types of 3-receiver BCs with degraded message sets are studied:
\begin{enumerate}
\item 3-receiver BC with 3 degraded message sets where $W_0$ is sent to all three receivers, $W_1$ is sent to receivers 1 and 2, and a second private message $W_2$ is sent to receiver 1;
\item 3-receiver BC with 2 degraded message sets (Type 1) where the common message $W_0$ is sent to all three receivers and a private message $W_1$ is sent to the first receiver;
\item 3-receiver BC with 2 degraded message sets (Type 2) where the common message $W_0$ is sent to all three receivers and a private message $W_1$ sent to receivers 1 and 2.
\end{enumerate}

While preparing this paper for submission, the authors became aware that Nair and El Gamal in \cite{nair_09} used a different coding scheme for their achievability proof compared to their earlier work \cite{nair_07}, with detailed proofs in \cite{nair_07_ext}. The added ingredient is rate splitting. However, a coding scheme with and without rate splitting is shown to give the same rate region in \cite{nair_09}. Based on this, in this paper, we shall not use rate splitting but base our achievability proof on the one in \cite{nair_07,nair_07_ext}.

The objective of this paper is to study this model of the 3-receiver BC with degraded message sets of \cite{nair_07},  \cite{nair_07_ext} with {\em secrecy} constraints. In particular, we characterize the transmission rates for the three types of 3-receiver BCs with degraded message sets from the model mentioned above where receiver 3 is a wiretapper. We note that the insights which this model of the 3-receiver BC with degraded message sets might bring are due to it being a more general model than the 2- or 3-receiver degraded BC with secrecy constraints. We also note that Chia and El Gamal in \cite{chia_09} have also studied the 3-receiver BC with 2 degraded message sets (Type 2) with receiver 3 being a wiretapper, but using a different coding scheme.

For the 3-receiver BC with 3 degraded message sets and 2 degraded message sets (Type 1) without secrecy constraints, the inner capacity bound in \cite{nair_07}, \cite{nair_07_ext} is achievable by superposition coding, Marton's achievability technique \cite{marton_79} and indirect decoding, where the receivers decoding the common message only do so via satellite codewords instead of cloud centers. For the general 3-receiver BC with degraded message sets, an outer bound to the capacity region was given in \cite{nair_07,nair_07_ext} only for the general 3-receiver BC with 2 degraded message sets (Types 1 and 2). For the 3-receiver BC with 2 degraded message sets (Type 2), the inner and outer bounds coincide under the condition that first and second receivers are less noisy than the third receiver.

In our earlier work \cite{choo_09_submit}, we had studied the 3-receiver BC with 2 degraded message sets (Type 1), with the third receiver regarded
as a wiretapper from which the private message is to be kept secret.  In this paper, we consider the more general model of the 3-receiver BC with 3
degraded message sets where the third receiver is a wiretapper from which the private messages $W_1$, $W_2$ are to be kept secret. As the wiretapper
in this case also decodes the common message, the 3-receiver BC with 3 degraded message sets with the third receiver a wiretapper describes a more
general scenario than three types of scenarios: the 2-receiver BCs with an external wiretapper of \cite{bagherikaram_08_submit}, the 2-receiver BC
with 3 degraded message sets and an external wiretapper, and the 3-receiver degraded BCs with an external wiretapper by the virtue of the general
conditions on the receivers.

In our secure coding scheme, we shall use a combination of the code partitioning of Wyner \cite{wyner_75} and double-binning of Liu {\em et
al.}~\cite{liu_08} to show the achievability of an inner bound to the rate-equivocation region for the 3-receiver BC with 3 degraded message sets.
Error probability analysis and equivocation calculation for the private messages are provided. The proposed secure coding scheme is shown to be
sufficient for providing security for both the 3-receiver BC with 3 degraded message sets and the 3-receiver BC with 2 degraded message sets (Type
1). We obtain outer bounds to the rate-equivocation region for the 3-receiver BC with 3 degraded message sets for the case where receiver 1 is more
capable than the wiretap receiver 3, a weaker condition than the condition that receiver 3 is a degraded version of receiver 1 or the condition that
receiver 1 is less noisy than the wiretap receiver 3 \cite{korner_77}. By removing the security constraints, we further obtain an outer bound to the
capacity region for the \textit{general} 3-receiver BC with 3 degraded message sets, which is not found in \cite{nair_07}-- \cite{nair_09}. This is because the
condition that receiver 1 is more capable than receiver 3 applies only to the case where we have secrecy constraints. Then, we show that the outer
bounds to the rate-equivocation region for the 3-receiver BC with 3 degraded message sets reduce to the outer bounds to the rate-equivocation region
for the 3-receiver BC with 2 degraded message sets (Type 1), if receiver 1 is more capable than the wiretap receiver 3. Finally, we show that, under
the condition that the first and second receivers are less noisy than the third receiver, respectively (still a more general condition than
degradedness \cite{korner_77}), the inner and outer bounds to the rate-equivocation region for the 3-receiver BC with 3 degraded message sets reduce
to the region for the 3-receiver BC with 2 degraded message sets (Type 2). This rate-equivocation region we obtain is furthermore a special case of
the variant of the 3-receiver BC with 2 degraded message sets (Type 2) studied in \cite{chia_09} with a different coding scheme.

This paper is organized as follows. In Section \ref{3rx_bc}, we describe the model for the 3-receiver BC with degraded message sets. In Section
\ref{secrecy_bounds}, we state our main results, the bounds to the rate-equivocation region. In Section \ref{inner}, we show achievability of the
inner bound to the rate-equivocation region using our secure coding scheme for the 3-receiver BC with 3 degraded message sets and the 3-receiver BC
with 2 degraded message sets (Type 1) and show error probability analysis and equivocation calculation for the private messages. We show that the
coding scheme provides security for both types of channel. In Section \ref{outer}, we show the proof of the outer bounds for the three types of the
3-receiver BC with degraded message sets. Lastly, we give conclusions in Section \ref{concl}.

\section{The 3-Receiver BC with Degraded Message Sets}\label{3rx_bc}
In this paper, we use the uppercase letter to denote a random variable (e.g., $X$) and the lowercase letter for its realization (e.g., $x$). The alphabet set of $X$ is denoted by ${\cal X}$ so that $x\in{\cal X}$. We denote a sequence of $n$ random variables by ${\bf X}=(X_1,\dots,X_n)$ with its realization ${\bf x}=(x_1,\dots,x_n)\in{\cal X}^n$ if $x_i\in{\cal X}$ for $i=1,2,\dots,n$. Furthermore, we define the subsequences of ${\bf X}$ as ${\bf X}^i\triangleq(X_1,X_2,\dots,X_i)$ and $\tilde{\bf X}^i\triangleq(X_i,\dots,X_n)$.

The discrete memoryless BC with 3 receivers has an input random sequence, ${\bf X}$, and 3 output random sequences at the receivers, denoted respectively by ${\bf Y}_1,{\bf Y}_2$ and ${\bf Y}_3$, all of length $n$, with ${\bf x}\in\mathcal{X}^n$, ${\bf y}_1\in\mathcal{Y}_1^n$, ${\bf y}_2\in\mathcal{Y}_2^n$, and ${\bf y}_3\in\mathcal{Y}_3^n$. The conditional distribution for $n$ uses of the channel is given by
\begin{equation}
p(\mathbf{y}_1,\mathbf{y}_2,\mathbf{y}_3|\mathbf{x})=\prod_{i=1}^n p(y_{1i},y_{2i},y_{3i}|x_i).
\end{equation}

A $(2^{nR_0}, 2^{nR_1}, 2^{nR_2}, n)$-code for the 3-receiver BC with 3 degraded message sets, as depicted in Figure \ref{F:Fig3dm}, consists of the following parameters:
\begin{align*}
\mathcal{W}_0&=\left\{1,\dots,2^{nR_0}\right\},\mbox{(common message set)}\\
\mathcal{W}_1&=\left\{1,\dots,2^{nR_1}\right\},\mbox{(private message set)},\\
\mathcal{W}_2&=\left\{1,\dots,2^{nR_2}\right\},\mbox{(private message set)},\\
f&:\mathcal{W}_0 \times \mathcal{W}_1 \times \mathcal{W}_2 \mapsto \mathcal{X}^n,\mbox{(encoding function)},\\
g_1&:\mathcal{Y}_1^n \mapsto \mathcal{W}_0 \times \mathcal{W}_1 \times \mathcal{W}_2,\mbox{(decoding function 1)},\\
g_2&:\mathcal{Y}_2^n \mapsto \mathcal{W}_0 \times \mathcal{W}_1,\mbox{(decoding function 2)},\\
g_3&:\mathcal{Y}_3^n \mapsto \mathcal{W}_0,\mbox{(decoding function 3)}.
\end{align*}
In particular, we have $g_1 (\mathbf{Y}_1) = (\hat W_0^{(1)}, \hat W_1^{(1)}, \hat W_2 )$, $g_2 (\mathbf{Y}_2) = (\hat W_0^{(2)}, \hat W_1^{(2)})$, and $g_3 (\mathbf{Y}_3)=\hat W_0^{(3)}$, where the notation ``$\hat{(\cdot)}$'' highlights that the decoded messages are estimates, with the error probability
\begin{equation}
P_e^{(n)}=\Pr\left\{ (\hat W_0^{(1)}, \hat W_0^{(2)}, \hat W_0^{(3)}, \hat W_1^{(1)}, \hat W_1^{(2)}, \hat W_2) \neq (W_0,W_0,W_0,W_1,W_1,W_2)   \right\}.
\end{equation}
In this setup, $Y_3$ is the wiretapper, and the secrecy level of the messages sent are as follows:
\begin{enumerate}
\item For $W_1$ sent to users 1 and 2,  the secrecy level is defined by the equivocation rate $\frac{1}{n}H(W_1 | \mathbf{Y}_3)$;
\item For $W_2$ sent to user 1,  the secrecy level is defined by the equivocation rate $\frac{1}{n}H(W_2 | \mathbf{Y}_3)$;
\item The combined message $(W_1,W_2)$ sent to user 1 has secrecy level defined by the equivocation rate $\frac{1}{n}H(W_1, W_2 | \mathbf{Y}_3)$.
\end{enumerate}

\begin{figure}
\centering
\includegraphics[scale=0.65]{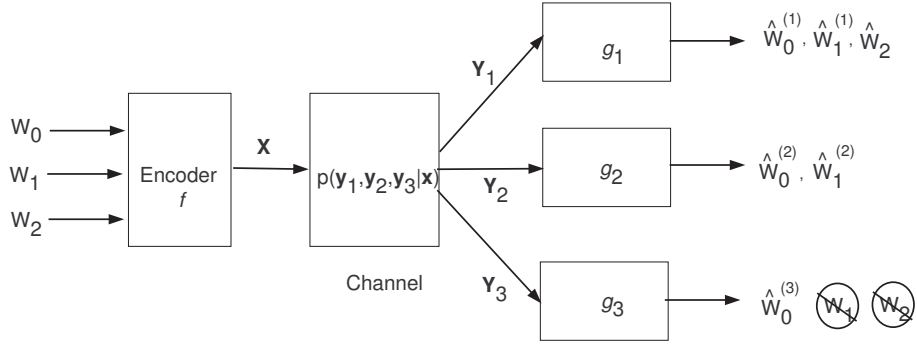}
\caption{The 3-receiver BC with 3 degraded message sets and confidential messages.}\label{F:Fig3dm}
\end{figure}

In addition, a $(2^{nR_0}, 2^{nR_1},n)$-code for the 3-receiver BC with 2 degraded message sets (Type 1), as shown in Figure \ref{F:Fig2dm1},
consists of the following parameters:
\begin{align*}
\mathcal{W}_0&=\left\{1,\dots,2^{nR_0}\right\},\mbox{(common message set)}\\
\mathcal{W}_1&=\left\{1,\dots,2^{nR_1}\right\},\mbox{(private message set)},\\
f&:\mathcal{W}_0 \times \mathcal{W}_1 \mapsto \mathcal{X}^n,\mbox{(encoding function)},\\
g_1&:\mathcal{Y}_1^n \mapsto \mathcal{W}_0 \times \mathcal{W}_1,\mbox{(decoding function 1)},\\
g_2&:\mathcal{Y}_2^n \mapsto \mathcal{W}_0,\mbox{(decoding function 2)},\\
g_3&:\mathcal{Y}_3^n \mapsto \mathcal{W}_0,\mbox{(decoding function 3)}.
\end{align*}
We have $g_1 (\mathbf{Y}_1) = (\hat W_0^{(1)}, \hat W_1^{(1)})$, $g_2 (\mathbf{Y}_2) = \hat W_0^{(2)}$, and $g_3 (\mathbf{Y}_3)=\hat W_0^{(3)}$, with the error probability
\begin{equation}
P_e^{(n)}=\Pr\left\{ (\hat W_0^{(1)}, \hat W_0^{(2)}, \hat W_0^{(3)}, \hat W_1^{(1)}) \neq (W_0,W_0,W_0,W_1,)   \right\}.
\end{equation}
With $Y_3$ the wiretapper, and the secrecy level of the message sent is $\frac{1}{n}H(W_1 | \mathbf{Y}_3)$.

\begin{figure}
\centering
\includegraphics[scale=0.65]{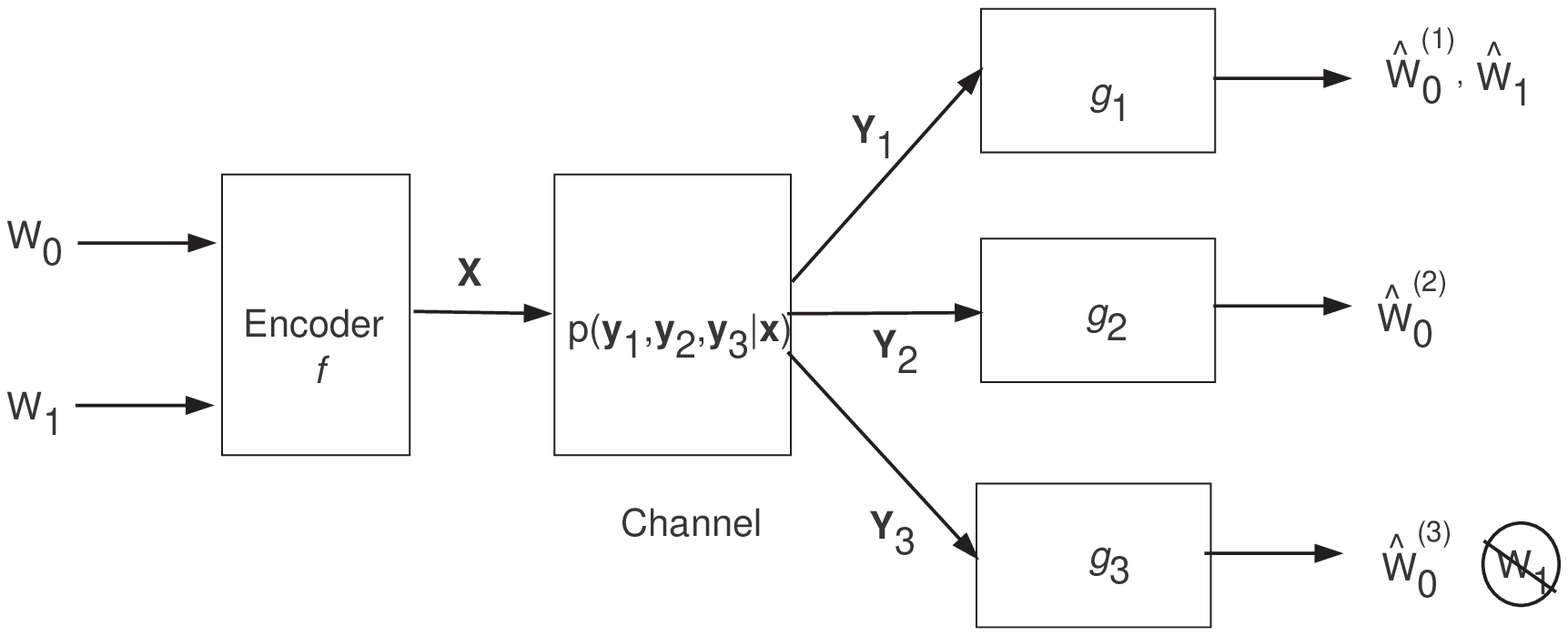}
\caption{The 3-receiver BC with 2 degraded message sets (Type 1) and confidential message.}\label{F:Fig2dm1}
\end{figure}

Finally, a $(2^{nR_0}, 2^{nR_1},n)$-code for the 3-receiver BC with 2 degraded message sets (Type 2), as shown in Figure \ref{F:Fig2dm2}, consists of
the parameters:
\begin{align*}
\mathcal{W}_0&=\left\{1,\dots,2^{nR_0}\right\},\mbox{(common message set)}\\
\mathcal{W}_1&=\left\{1,\dots,2^{nR_1}\right\},\mbox{(private message set)},\\
f&:\mathcal{W}_0 \times \mathcal{W}_1 \mapsto \mathcal{X}^n,\mbox{(encoding function)},\\
g_1&:\mathcal{Y}_1^n \mapsto \mathcal{W}_0 \times \mathcal{W}_1,\mbox{(decoding function 1)},\\
g_2&:\mathcal{Y}_2^n \mapsto \mathcal{W}_0 \times \mathcal{W}_1,\mbox{(decoding function 2)},\\
g_3&:\mathcal{Y}_3^n \mapsto \mathcal{W}_0,\mbox{(decoding function 3)}.
\end{align*}
We have $g_1 (\mathbf{Y}_1) = (\hat W_0^{(1)}, \hat W_1^{(1)})$, $g_2 (\mathbf{Y}_2) = (\hat W_0^{(2)}, \hat W_1^{(2)})$, and $g_3 (\mathbf{Y}_3)=\hat W_0^{(3)}$, and error probability
\begin{equation}
P_e^{(n)}=\Pr\left\{ (\hat W_0^{(1)}, \hat W_0^{(2)}, \hat W_0^{(3)}, \hat W_1^{(1)}, \hat W_1^{(2)}) \neq (W_0,W_0,W_0,W_1,W_1)   \right\}.
\end{equation}
The secrecy level of the message $W_1$ sent to users 1 and 2 is defined by the equivocation rate $\frac{1}{n}H(W_1 | \mathbf{Y}_3)$.

\begin{figure}
\centering
\includegraphics[scale=0.65]{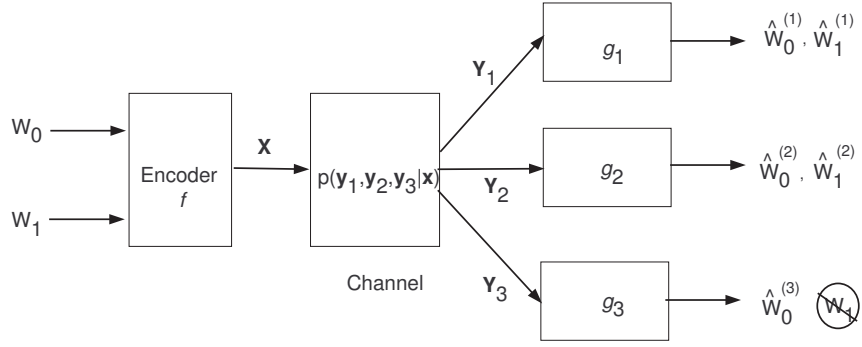}
\caption{The 3-receiver BC with 2 degraded message sets (Type 2) and confidential message.}\label{F:Fig2dm2}
\end{figure}

\section{Bounds to the Rate-Equivocation Region}\label{secrecy_bounds}
\subsection{The 3-Receiver BC with 3 Degraded Message Sets}
For the 3-receiver BC with 3 degraded message sets, the rate tuple $(R_0,R_1,R_{1e},R_2,R_{2e})$ is said to be achievable if for any
$\eta,\epsilon_1,\tilde{\epsilon}_1,\epsilon_2,\epsilon_{1,2}>0$, there exists a sequence of $(2^{nR_0},2^{nR_1},2^{nR_2},n)$-codes for which $P_e^{(n)}\le\eta$ and the
equivocation rates $R_{1e}$ and $R_{2e}$ satisfy
\begin{subequations}\label{sec_cond}
\begin{align}
\frac{1}{n}H(W_1|\mathbf{Y}_3)& \ge R_{1e}-\epsilon_1,~\mbox{or}~\frac{1}{n}H(W_1|\mathbf{Y}_3) \ge R_{1e}-\tilde{\epsilon}_1,\label{sec_cond1} \\
\frac{1}{n}H(W_2|\mathbf{Y}_3)& \ge R_{2e}-\epsilon_2, \label{sec_cond2} \\
\frac{1}{n}H(W_1,W_2|\mathbf{Y}_3) & \ge R_{1e} + R_{2e}-\epsilon_{1,2}. \label{sec_cond12}
\end{align}
\end{subequations}
The two conditions on $W_1$ arise because the equivocation rate depends on which destination $W_1$ is sent to, as can be seen below in \eqref{inner_eq1}.  Recall from the model of the 3-receiver BC with 3 degraded message sets that $W_1$ is sent to both $Y_1$ and $Y_2$.  The first equivocation rate in \eqref{sec_cond1} corresponds to $W_1$ being sent to receiver $Y_2$ and the second equivocation rate in \eqref{sec_cond1} corresponds to $W_1$ being sent to receiver $Y_1$. The rate-equivocation region for the 3-receiver BC with 3 degraded message sets is the closure of the set of all rate-tuples such that $(R_0,R_1,R_{1e},R_2,R_{2e})$ is achievable. Our analysis does not include the case of perfect secrecy (i.e., the rate region with $R_{1e}=R_1$ and $R_{2e}=R_2$).  The following theorems summarize the main results of this paper.

\newtheorem{theorem}{Theorem}
\newtheorem{corollary}{Corollary}

\begin{theorem}\label{thm1}
An inner bound to the rate-equivocation region for the 3-receiver BC with 3 degraded message sets is the closure of all rate-tuples $(R_0,R_1,R_{1e},R_2,R_{2e})$ satisfying
\begin{subequations}
\begin{align}
R_{1e} \le & R_1,\\
R_{2e} \le & R_2,\\
R_0\le & I(U_3;Y_3),\\
R_{1e}\le & \min\left\{I(U_2;Y_2|U_1)-R_1', I(X;Y_1 | U_3)-R_1'-R_2' \right\}, \label{inner_eq1} \\
R_{2e} \le & I(X;Y_1 | U_2)-R_2' , \label{inner_eq2} \\
R_{1e} + R_{2e} \le & I(X;Y_1|U_1) -R_1'-R_2', \label{inner_eq12}\\
R_0 + R_1 \le & \min \left\{ I(U_2;Y_2),I(U_3;Y_3) + I(U_2 ; Y_2|U_1) - I(U_2;U_3|U_1) \right\} \\
2R_0 + R_1 \le & I(U_3;Y_3) + I(U_2 ; Y_2) - I(U_2;U_3|U_1) ,\\
R_0 + R_2 \le & I(U_3;Y_3) + I(X;Y_1|U_2,U_3) ,\\
\nonumber R_0 + R_1 + R_2 \le & \min \left\{ I(U_3;Y_3) + I(X;Y_1|U_3), I(X;Y_1),  \right.\\
& \left. I(U_3;Y_3) + I(U_2;Y_2|U_1) -I(U_2;U_3|U_1)+I(X;Y_1|U_2,U_3) \right\}, \\
2R_0 + R_1 + R_2 \le & I(U_3;Y_3)+I(U_2;Y_2)-I(U_2;U_3|U_1)+I(X;Y_1|U_2,U_3), \\
R_0 + 2R_1 + R_2 \le & I(U_3;Y_3)+I(U_2;Y_2|U_1)-I(U_2;U_3|U_1)+I(X;Y_1|U_3), \\
2R_0 + 2R_1 + R_2 \le & I(U_3;Y_3)+I(U_2;Y_2)-I(U_2;U_3|U_1)+I(X;Y_1|U_3),
\end{align}
\end{subequations}
in which $R_1'\triangleq I(U_2;Y_3|U_1)$ and $R_2' \triangleq I(X;Y_3|U_2)$ are defined over the the probability density function (p.d.f.)
\begin{equation}\label{pdf}
p(u_1,u_2,u_3,x)=p(u_1)p(u_2|u_1)p(x,u_3|u_2)=p(u_1)p(u_3|u_1)p(x,u_2|u_3)=p(u_1)p(u_2,u_3|u_1)p(x|u_2,u_3),
\end{equation}
which is induced by the coding scheme.  In addition, we require that the condition
\begin{align}\label{thm1_cond}
    I(X;Y_3|U_2) \le I(X;Y_1|U_2,U_3)
\end{align}
is met. From the p.d.f.~\eqref{pdf}, the auxiliary random variables $U_1$, $U_2$ and $U_3$ satisfy the Markov chain conditions
\begin{subequations}\label{mkv}
\begin{align}
& U_1\to U_2\to (U_3,X) \to (Y_1,Y_2,Y_3), \label{mkv1} \\
& U_1\to U_3\to (U_2,X) \to (Y_1,Y_2,Y_3),  \label{mkv2} \\
& U_1 \to (U_2,U_3) \to X \to (Y_1,Y_2,Y_3). \label{mkv3}
\end{align}
\end{subequations}

\end{theorem}
\begin{proof}
The proof of achievability is based on that for the 3-receiver BC with 3 degraded message sets in \cite{nair_07}, \cite{nair_07_ext} which uses Marton's achievability scheme \cite{marton_79} combined with superposition coding and is given in Section \ref{ach} with the equivocation calculation (bounds for $R_{e1}, R_{e2}$) to be presented in Section \ref{eqv_calc}.
\end{proof}

Since our achievability scheme is based upon that of \cite{nair_07}, \cite{nair_07_ext}, it is natural that the inner bound is the same as that of \cite{nair_07}, \cite{nair_07_ext}, but with the addition of the equivocation rates.  In fact it will be the same as \cite{nair_09}, with the addition of the equivocation rates.  As a check, setting $Y_1 = Y_3$ in \eqref{inner_eq1}--\eqref{inner_eq12}, $R_{1e} \le 0$, $R_{2e} \le 0$ and $R_{1e} + R_{2e} \le 0$, so no secrecy rate is possible.  Thus the equivocation rates \eqref{inner_eq1}--\eqref{inner_eq12} are achievable.


\begin{theorem}\label{thm2}
An outer bound to the rate-equivocation region for the 3-receiver BC with 3 degraded message sets, where $Y_1$ is more capable than $Y_3$, is the closure of all rate-tuples $(R_0,R_1,R_{1e},R_2,R_{2e})$ that satisfies
\begin{subequations}
\begin{align}
R_{1e} \le & R_1, \\
R_{2e} \le & R_2, \\
R_0\le & \min \left\{ I(U_1;Y_1),I(U_3;Y_3)-I(U_3;Y_1|U_1) \right\}, \\
R_{1e} \le & \min \left\{ I(U_2;Y_2|U_1) - I(U_2;Y_3|U_1), I(X;Y_1|U_3)-I(X;Y_3|U_1) \right\}, \\
R_{2e} \le & I(X;Y_1|U_2)-I(X;Y_3|U_2), \\
R_{1e} + R_{2e} \le & I(X;Y_1|U_1)-I(X;Y_3|U_1), \\
\nonumber R_0 + R_1 \leq & \min \left\{ I(U_2;Y_1), I(U_2;Y_2), I(U_1;Y_1)+I(U_2;Y_2|U_1),\right.\\
& \left. I(U_3;Y_3)+I(U_2;Y_1|U_1) , I(U_3;Y_3)+I(U_2;Y_2|U_1)\right\}, \\
R_0 + R_2 \leq & \min \left\{ I(U_1;Y_1)+I(X;Y_1|U_2,U_3), I(U_3;Y_3)+I(X;Y_1|U_2,U_3)  \right\}, \\
\nonumber R_0 + R_1 + R_2 \leq & \min \left\{ I(X;Y_1), I(U_3;Y_3)+I(X;Y_1|U_3), \right. \\
\nonumber & I(U_1;Y_1)+I(U_2;Y_2|U_1)+I(X;Y_1|U_2,U_3), \\
& I(U_3;Y_3)+I(U_2;Y_2|U_1)+I(X;Y_1|U_2,U_3),\left.I(U_2;Y_2)+I(X;Y_1|U_2,U_3) \right\}.
\end{align}
\end{subequations}
\end{theorem}

\begin{proof}
The proof for this outer bound is given in Section \ref{outer_gen}.
\end{proof}

We see that the equivocation rates for $(R_{1e},R_{2e})$ in the inner and outer bounds in Theorems \ref{thm1} and \ref{thm2} match. Note that the equivocation rate for $R_{1e}$ received at $Y_1$ is reduced by $\Delta_1 = I(U_2;Y_3|U_1) + I(X;Y_3|U_2) = I(X;Y_3|U_1)$. In $\Delta_1$, the first term is needed to protect the codewords generated by Marton's achievability scheme, and the second term protects codewords generated by superposition coding. While it is only required to protect the codewords generated by Marton's achievability scheme for the general 2-receiver BC in \cite{bagherikaram_08_submit}, our secure scheme (to be presented in Section IV) does this, as well as protects the additional codewords generated by superposition coding. Hence, our secure scheme results in a loss for $R_{1e}$ (compared to $R_1$) that may be larger than expected.

It is also noted that by removing the secrecy constraints from the outer bound to the rate-equivocation region for the 3-receiver BC with 3 degraded message sets, we can obtain a new outer bound to the capacity region of the \textit{general} 3-receiver BC with 3 degraded message sets without secrecy.  We see this by setting $R_{1e}=0$ and $R_{2e}=0$ in Theorem \ref{thm2} above. Since the restriction that receiver $Y_1$ is more capable than receiver $Y_3$ is only applicable when deriving $R_{1e}$ and $R_{2e}$ as will be shown in Section \ref{outer_gen}, removing the secrecy constraints will give us the outer bound to the capacity region of the general 3-receiver BC with 3 degraded message sets.

\begin{theorem}\label{thm3}
An outer bound to the capacity region for the general 3-receiver BC with 3 degraded message sets is the closure of all rate-tuples $(R_0,R_1,R_2)$ satisfying
\begin{subequations}
\begin{align}
R_0\le & \min \left\{ I(U_1;Y_1),I(U_3;Y_3)-I(U_3;Y_1|U_1) \right\}, \\
\nonumber R_0 + R_1 \leq & \min \left\{ I(U_2;Y_1), I(U_2;Y_2), I(U_1;Y_1)+I(U_2;Y_2|U_1),\right.\\
& \left. I(U_3;Y_3)+I(U_2;Y_1|U_1) , I(U_3;Y_3)+I(U_2;Y_2|U_1)\right\}, \\
R_0 + R_2 \leq & \min \left\{ I(U_1;Y_1)+I(X;Y_1|U_2,U_3), I(U_3;Y_3)+I(X;Y_1|U_2,U_3)  \right\}, \\
\nonumber R_0 + R_1 + R_2 \leq & \min \left\{ I(X;Y_1), I(U_3;Y_3)+I(X;Y_1|U_3),I(U_1;Y_1)+I(U_2;Y_2|U_1)+I(X;Y_1|U_2,U_3), \right.\\
& I(U_3;Y_3)+I(U_2;Y_2|U_1)+I(X;Y_1|U_2,U_3),\left. I(U_2;Y_2)+I(X;Y_1|U_2,U_3)\right\}.
\end{align}
\end{subequations}
\end{theorem}

\begin{proof}
As described above.
\end{proof}


\subsection{The 3-Receiver BC with 2 Degraded Message Sets}
The 3-receiver BC with 3 degraded message sets with secrecy constraints can be specialized to 2 classes of a 3-receiver BC with 2 degraded message sets with secrecy constraints:
\begin{enumerate}
\item Type 1: A 3-receiver BC where $(W_0,W_1)$ is sent to receiver $Y_1$ and $W_0$ is sent to receivers $Y_2$ and $Y_3$, where $W_1$ is to be kept secret from receiver $Y_3$;
\item Type 2: A 3-receiver BC where $(W_0,W_1)$ is sent to receivers $Y_1$ and $Y_2$ and $W_0$ is sent to receiver $Y_3$, where $W_1$ is to be kept secret from receiver $Y_3$.
\end{enumerate}
We note that the inner and outer bounds do not match for the first case, but match for the second case under the condition that both receivers $Y_1$ and $Y_2$ are less noisy than receiver $Y_3$.

We have studied the Type 1 channel in \cite{choo_09_submit}. In this paper, we shall briefly review the achievability scheme for secrecy constraints to see the differences from the 3 degraded message sets case, and show that the outer bound for the 3 degraded message sets case can be reduced to the outer bound for this Type 1 channel.

For the Type 2 channel, we shall show that the bounds on the rate-equivocation region can be specialized from the 3 degraded message sets case. We also note that the Type 2 channel is a special case of the inner bound to the rate-equivocation region for a 3-receiver BC with 2 degraded message sets studied in Chia and El Gamal \cite{chia_09} using a different coding scheme. In \cite{chia_09}, the message reception and secrecy conditions are the same as the Type 2 channel. Thus, both our bounds and that of \cite{chia_09} will reduce to the Type 2 channel. Also, our outer bounds will reduce to the Type 2 channel under the conditions that both receivers $Y_1$ and $Y_2$ are less noisy than receiver $Y_3$.

We state the inner and outer bounds to the rate-equivocation region for the Type 1 channel in Corollaries \ref{col1} and \ref{col2}, and the rate-equivocation region for the Type 2 channel in Corollary \ref{col3}.

\begin{corollary}\label{col1}
An inner bound to the rate-equivocation region for the 3-receiver BC with 2 degraded message sets (Type 1) is the closure of all rate-tuples $(R_0,R_1,R_{1e})$ satisfying
\begin{subequations}
\begin{align}
R_{1e} \le & R_1 \\
R_0\le & \min \left\{I(U_2;Y_2),I(U_3;Y_3)\right\}\\
R_{1e} \le & \min \big\{I(X;Y_1|U_1)-\Delta_2,I(X;Y_1|U_2)+I(X;Y_1|U_3)-I(X;Y_3|U_2)-\Delta_2\big\} , \label{eqv_r1} \\
2R_0\le & I(U_2;Y_2)+I(U_3;Y_3)-I(U_2;U_3|U_1)\\
R_0+R_1\le & \min \big\{I(X;Y_1),I(U_2;Y_2)+I(X;Y_1|U_2),I(U_3;Y_3)+I(X;Y_1|U_3)\big\},\\
2R_0+R_1\le & I(U_2;Y_2)+I(U_3;Y_3)-I(U_2;U_3|U_1)+I(X;Y_1|U_2,U_3),\\
2R_0+2R_1\le & I(U_2;Y_2)+I(X;Y_1|U_2)+I(U_3;Y_3)+I(X;Y_1|U_3)-I(U_2;U_3|U_1),
\end{align}
\end{subequations}
under the same Markov chain conditions \eqref{mkv} for the auxiliary random variables, where $\Delta_2 \triangleq I(U_2;Y_3|U_1)+I(X;Y_3|U_2)$, and the conditions
\begin{equation}\label{cor2_cond}
\left\{\begin{aligned}
I(X;Y_3|U_2) &\le I(X;Y_1|U_2,U_3),\\
I(X;Y_3|U_2) &\le I(X;Y_1|U_2),
\end{aligned}\right.
\end{equation}
are satisfied.
\end{corollary}

\begin{proof}
See Section \ref{ach2} for the achievability proof, and \cite{choo_09_submit} for the equivocation calculation.
\end{proof}

We see that $\Delta_2$ in Corollary \ref{col2} may be expressed as
\begin{equation}
\Delta_2 \triangleq I(U_2;Y_3|U_1)+I(X;Y_3|U_2) = I(X;Y_3|U_1),
\end{equation}
which is $\ge I(X;Y_3|U_3)$.  Thus, as a check, when $Y_1=Y_3$ in \eqref{eqv_r1}, $R_{1e} \le 0$, so no secrecy rate is possible and therefore the equivocation rate \eqref{eqv_r1} is achievable. Also, when compared to the equivocation rates on $R_{1e}$ for the 3 degraded message sets channel in \eqref{inner_eq1}, a smaller rate is achievable for $W_1$ sent to $Y_1$. Then, by the virtue of sending $W_2$ to $Y_1$, the coding scheme of \cite{nair_07}, \cite{nair_07_ext} is able to give a higher equivocation rate for $W_1$ sent to $Y_1$. It appears that by sending more messages to receiver $Y_1$, then the achievable equivocation rates can be increased.

The lower achievable rate for $W_1$ sent to $Y_1$ for the 2 degraded message sets (Type 1) channel is due to the fact that the achievable coding scheme protects all the codewords generated by superposition coding. We note that the coding scheme of \cite{nair_07,nair_07_ext} generates codewords giving rise to the rates $R_1 \leq I(X;Y_1|U_2)+I(X;Y_1|U_3)$ and $R_1 \leq I(X;Y_1|U_1)$.  From the fact that when $Y_1=Y_3$ in \eqref{eqv_r1}, $R_{1e} \le 0$ for both choices of $R_{1e}$, so implying the equivocation rates \eqref{eqv_r1} are achievable, we see that our proposed secure scheme is able to protect all the codewords generated by superposition coding, but with a smaller achievable equivocation rate for $W_1$ sent to $Y_1$ compared to $R_{1e}$ (with $W_1$ sent to $Y_1$) for the 3 degraded message sets channel.


The outer bound for the Type 1 3-receiver 2 degraded message sets BC is stated as follows.

\begin{corollary}\label{col2}
An outer bound to the rate-equivocation region for the 3-receiver BC with 2 degraded message sets (Type 1), where $Y_1$ is more capable than $Y_3$, is the closure of all rate-tuples
$(R_0,R_1,R_{1e})$ satisfying
\begin{subequations}
\begin{align}
R_{1e} \le & R_1, \\
R_0 \le & \min\left\{I(U_1;Y_1),I(U_2;Y_2)-I(U_2;Y_1|U_1),I(U_3;Y_3)-I(U_3;Y_1|U_1)\right\}\\
R_{1e} \le & I(X;Y_1|U_1)-I(X;Y_3|U_1),\\
R_0+R_1\le & \min\big\{I(X;Y_1),I(U_2;Y_2)+I(X;Y_1|U_2),I(U_3;Y_3)+I(X;Y_1|U_3)\big\}.
\end{align}
\end{subequations}
\end{corollary}
\begin{proof}
See Section \ref{outer_spec1}.
\end{proof}


We state the rate-equivocation region for the Type 2 3-receiver 2 degraded message sets BC below.

\begin{corollary}\label{col3}
The secrecy capacity region for the 3-receiver BC with 2 degraded message sets (Type 2) for the case where $Y_1$ and $Y_2$ are both less noisy than $Y_3$ is the closure of all rate-tuples $(R_0,R_1,R_{1e})$ satisfying
\begin{subequations}
\begin{align}
R_{1e} \le & R_1,\\
R_0 \le & I(U;Y_3),\\
R_{1e} \leq & \min \left\{ I(X;Y_1|U)-I(X;Y_3|U), I(X;Y_2|U)-I(X;Y_3|U) \right\}, \\
R_0 + R_1 \leq & \min \left\{ I(X;Y_1), I(X;Y_2) \right\},
\end{align}
\end{subequations}
over the p.d.f. $p(u,x)=p(u)p(x|u)$.
\end{corollary}
\begin{proof}
In this channel class, the inner and outer bounds match. The proof of achievability follows by using code partitioning for security, as in
\cite{csiszar_78,wyner_75}, where it can be seen that the codeword $\mathbf{X}$ is protected by the partition of $I(X;Y_3|U)$. The rate-equivocation
region is achievable by setting $R_2=0$, $R_{2e}=0$, $U_2=X$, $U_3 = U_1 = U$ in Theorems \ref{thm1} and \ref{thm2} and using the conditions that
$Y_1$ and $Y_2$ are less noisy than $Y_3$. Therefore, we have the conditions $I(U;Y_3) \leq I(U;Y_1)$ and $I(U;Y_3) \leq I(U;Y_2)$. See Section
\ref{outer_spec2} for the converse proof.
\end{proof}

It is worth emphasizing here that this channel class is more general than the special case of the 3-receiver BC with 2 degraded message sets (Type 1)
under the condition that $Y_1$ is less noisy than $Y_3$ in \cite{choo_09_submit}, since $Y_2$ receives $W_1$ here but this is not the case in
\cite{choo_09_submit}.

\section{Inner Bound for the 3-Receiver BC with 3 Degraded Message Sets}\label{inner}
\subsection{Proof of Achievability for 3-Receiver BC with 3 Degraded Message Sets}\label{ach}
Our achievability proof for the 3-receiver BC with 3 degraded message sets is an alternative version of the one in [\ref{nair_eg}, Appendix III]. We use Wyner's code partitioning \cite{wyner_75} with the double-binning scheme of \cite{liu_08} to provide secrecy, together with the coding scheme for the 3-receiver BC with 3 degraded message sets in \cite{nair_07,nair_07_ext}.

The scheme of \cite{nair_07}, \cite{nair_07_ext} represents $W_0$ by $U_1$, then breaks $W_2$ into 2 parts. The first part is combined with $U_1$ by superposition coding to generate $U_3$.  The message $W_1$ is combined with $U_1$ by superposition coding to generate $U_2$. $U_2$ and $U_3$ are partitioned into bins and the product bin containing the joint typical pair (achievable by Marton's coding scheme) is combined with the second part of $W_2$ by superposition coding to obtain $X$.

At the receivers, $Y_1$ decodes $U_1$, $U_2$, $U_3$, and $X$ to recover the messages $W_0$, $W_1$ and $W_2$, while $Y_2$ decodes $U_1$ and $U_2$ to recover messages $W_0$ and $W_1$ and $Y_3$ decodes $U_1$ indirectly using $U_3$ to recover $W_0$. In our secure scheme, the codewords $\mathbf{U}_2$ and $\mathbf{X}$ are, respectively, protected from receiver $Y_3$ (i.e., the wiretapper) by a one-sided double-binning and code partitioning. This is depicted in Figure \ref{F:Fig1}.

\begin{figure}[!t]
\centering
\includegraphics[scale=0.65]{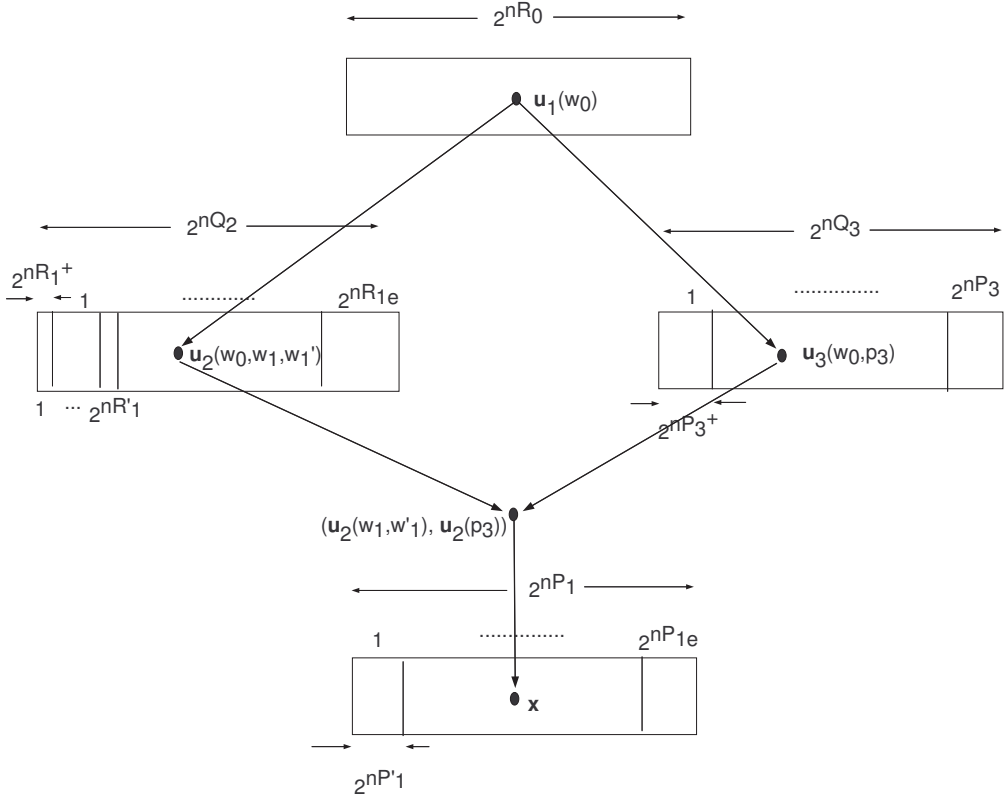}
\caption{Coding for 3-receiver BC with degraded message sets and confidential messages.}\label{F:Fig1}
\end{figure}

Suppose that we have the p.d.f.~in \eqref{pdf} which induces the Markov chain conditions $U_1\to U_2\to (U_3,X)$, $U_1\to U_3\to (U_2,X)$ and $U_1
\to (U_2,U_3) \to X$. The following describes the encoding and decoding processes.

\underline{Codebook generation:} Let $\tilde R_1 = R_{1e} + R_1' + R_1^{\dagger}$, $\tilde P_3 = P_3 + P_3^{\dagger}$, $P_1 = P_{1e} + P_1'$, and $R_{2e} = P_3 + P_{1e}$.  Define, for security,
\begin{equation}\label{security_rates}
P_1'\triangleq I(X;Y_3|U_2)-\delta_1,~\mbox{and}~
R_1'\triangleq I(U_2;Y_3|U_1)-\delta_1,
\end{equation}
where $\delta_1 > 0$ and is small for $n$ sufficiently large.

First of all, generate $2^{nR_0}$ sequences $\mathbf{U}_1(w_0)$, for $w_0\in\mathcal{W}_0$, randomly and uniformly from the set of typical
$\mathbf{U}_1$ sequences. For each $\mathbf{U}_1 (w_0)$, generate $2^{n Q_2}$ sequences $\mathbf{U}_2 (w_0, q_2)$ randomly and uniformly from the set
of conditionally typical $\mathbf{U}_2$ sequences, and also $2^{nQ_3}$ sequences $\mathbf{U}_3 (w_0, q_3)$ randomly and uniformly from the set of
conditionally typical $\mathbf{U}_3$ sequences. Next, randomly partition the sequences, $\mathbf{U}_2(w_0,q_2)$, into $2^{n \tilde R_1}$
equally-sized bins, and the sequences, $\mathbf{U}_3(w_0,q_3)$, into $2^{n \tilde P_3}$ equally-sized bins.  The $\mathbf{U}_2$ codewords undergo a
double partition: the first into $2^{nR_{1e}}$ bins, and the second further partitions them into $2^{nR_1'}$ bins, each of size $2^{nR_1^\dag}$. On
the other hand, the $\mathbf{U}_3$ codewords undergo a single partition into $2^{nP_3}$ bins, each of size $2^{nP_3^\dag}$.

Each product bin $(w_1,w_1',p_3)$ contains the joint typical pair $(\mathbf{U}_2(w_0,w_1,w_1',w_1^\dag),\mathbf{U}_3(w_0,p_3, p_3^\dag))$ for
$w_1\in\{1,\dots,2^{nR_{1e}}\}$, $w_1'\in\{1,\dots,2^{R_1'}\}$, $w_1^\dag\in\{1,\dots,2^{nR_1^\dag}\}$, $p_3\in\{1,\dots, 2^{nP_3}\}$, and
$p_3^\dag\in\{1,\dots,2^{nP_3^\dag}\}$ with high probability under the conditions \cite{el_gamal_81}
\begin{equation}\label{bin_rates}
\begin{aligned}
R_{1e} + R_1' + R_1^\dag&\le Q_2 \Rightarrow R_{1e} + R_1' \leq Q_2,\\
P_3+P_3^\dag&\le Q_3 \Rightarrow P_3 \leq Q_3,\\
R_1^\dag+P_3^\dag&>I(U_2;U_3|U_1),\\
R_{1e} + R_1'+P_3&\le Q_2+Q_3-I(U_2;U_3|U_1).
\end{aligned}
\end{equation}
Now let us rewrite the joint typical pair as $(\mathbf{u}_2 (w_0,w_1 , w_1'), \mathbf{u}_3 (w_0,p_3))$. For each such pair corresponding to the
product bin $(w_1,w_1',p_3)$, generate $2^{n P_1}$ sequences of codewords $\mathbf{X}(w_0,w_1,w_1',p_3,p_1,p_1')$, for
$p_1\in\{1,\dots,2^{nP_{1e}}\}$ and $p_1'\in\{1,\dots,2^{nP_1'}\}$, uniformly and randomly over the set of conditionally typical $\mathbf{X}$
sequences. The $2^{nP_1}$ codewords are partitioned into $2^{nP_{1e}}$ subcodes with $2^{nP_1'}$ codewords within the subcodes.

\underline{Encoding:} To send $(w_0,w_1,w_2)$, express $w_2$ by $(p_1,p_3)$ and send the codeword $\mathbf{x}(w_0,w_1,w_1',p_3,p_1,p_1')$.

\underline{Decoding:} Use $T_\epsilon^n(P_Z)$ to denote the set of jointly strong typical $n$-sequence with respect to the p.d.f.~$p(z)$. Without
loss of generality, assume that $(w_0, w_1, p_3, p_1) = (1,1,1,1)$ is sent and $w_1'$ and $p_1'$ can be arbitrary. The receivers decode as follows:
\begin{enumerate}
\item Receiver 1 uses joint typical decoding of $\{ \mathbf{u}_1, \mathbf{u}_2, \mathbf{u}_3, \mathbf{x}, \mathbf{y}_1 \}$ to find the indices $(w_0, w_1, p_3, p_1)$.
\item Receiver 2 uses indirect decoding of $\mathbf{u}_2$ \cite{nair_07} to find the index $w_0$. Once this is known, $\mathbf{u}_1$ is also found. Then, receiver 2 uses joint typical decoding of
$\{ \mathbf{u}_1, \mathbf{u}_2, \mathbf{y}_2 \}$ to find $w_1$.
\item Receiver 3 uses indirect decoding of $\mathbf{u}_3$ to find the index $w_0$.
\end{enumerate}

At receiver 1, the decoder seeks the indices $(w_0, w_1, p_3, p_1)$ so that
\begin{equation}
\left(\mathbf{u}_1(w_0),\mathbf{u}_2(w_0,w_1,w_1'),\mathbf{u}_3(p_3),\mathbf{x}(w_0,w_1,w_1',p_3,p_1,p_1'),\mathbf{y}_1\right) \in T^n_{\epsilon}(P_{U_1 U_2 U_3 X Y_1}).
\end{equation}
If there is none or more than one possible codeword, an error is declared.  The possible error events are as follows:

a) ${\tt E}_1:(w_0,w_1,w_1',p_3,p_1,p_1')=(1,1,w_1',1,1,p_1')$ but $\mathbf{u}_1$, $\mathbf{u}_2$, $\mathbf{u}_3$, $\mathbf{x}$ are not jointly
typical with $\mathbf{y}$. By the properties of strong typical sequences \cite{kramer_bk}, $\Pr\{{\tt E}_1\}\le\epsilon'$, where $\epsilon' \to 0$
for large $n$.

b) ${\tt E}_2: w_0 \neq 1$ and arbitrary $w_1$, $p_3$, $p_1$, with $\mathbf{u}_1$, $\mathbf{u}_2$, $\mathbf{u}_3$, $\mathbf{x}$ jointly typical with
$\mathbf{y}_1$. Then, we have
\begin{align}
\nonumber \Pr\{{\tt E}_2\} &\le \sum_{w_0 \neq 1\atop w_1,p_3,p_1,w_1',p_1'} \Pr \left\{ (\mathbf{U}_1 (w_0), \mathbf{U}_2 (w_0,w_1,w_1'), \mathbf{U}_3 (p_3), \mathbf{X}(w_0,w_1,w_1',p_3,p_1,p_1'), \mathbf{y}_1)\in T^n_{\epsilon}(P_{U_1 U_2 U_3 X Y_1}) \right\} \\
&\leq 2^{n(R_0 + R_{1e} + R_1' + P_{1e} + P_1' + P_3)} 2^{-n(I(U_1,U_2,U_3,X;Y_1)-2\delta)},
\end{align}
where $\delta \to 0$ as $\epsilon \to 0$ for $n$ sufficiently large.  For $\Pr\{{\tt E}_2\} \leq \epsilon'$, we require
\begin{equation}\label{pe1_2}
    R_0 + R_{1e} + R_1' + P_{1e} + P_1' + P_3 < I(U_1,U_2,U_3,X;Y_1) = I(X;Y_1)
\end{equation}
since $I(U_1,U_2,U_3;Y_1|X)=0$ by the Markov chain condition
\begin{equation}\label{mchain_err}
    U_1 \to (U_2,U_3) \to X \to Y_1.
\end{equation}

c) ${\tt E}_3: w_0 = 1, w_1 \neq 1$ and arbitrary $p_3$, $p_1$, with $\mathbf{u}_1$, $\mathbf{u}_2$, $\mathbf{u}_3$, $\mathbf{x}$ jointly typical
with $\mathbf{y}_1$. Then, we have
\begin{align}
\nonumber \Pr\{{\tt E}_3\} &\le \sum_{w_1 \neq 1\atop p_3,p_1,w_1',p_1'} \Pr \left\{ (\mathbf{u}_1 (1), \mathbf{U}_2 (1,w_1,w_1'), \mathbf{U}_3
(p_3),\mathbf{X}(1,w_1,w_1',p_3,p_1,p_1'),\mathbf{y}_1) \in T^n_{\epsilon}(P_{U_1 U_2 U_3 X Y_1}) \right\}\\
&\leq 2^{n(R_{1e} + R_1' + P_{1e} + P_1' + P_3)} 2^{-n(I(U_2,U_3,X;Y_1|U_1)-2\delta)}.
\end{align}
For $\Pr\{{\tt E}_3\} \leq \epsilon'$, we require
\begin{align}\label{pe1_3}
\nonumber R_{1e} + R_1' + P_{1e} + P_1' + P_3 & < I(U_2,U_3,X;Y_1|U_1) \\
& = I(X;Y_1|U_1) + I(U_2,U_3;Y_1|X,U_1) =  I(X;Y_1|U_1),
\end{align}
where the second line is due to $U_1 \to (U_2,U_3) \to X \to Y_1$.

d) ${\tt E}_4: w_0 = 1, w_1 = 1, p_3 \neq 1$ and arbitrary $p_1$, with $\mathbf{u}_1$, $\mathbf{u}_2$, $\mathbf{u}_3$, $\mathbf{x}$ jointly typical
with $\mathbf{y}_1$. Then, we have
\begin{align}
\nonumber \Pr\{{\tt E}_4\} &\le \sum_{p_3 \neq 1\atop p_1,w_1',p_1'} \Pr \left\{ (\mathbf{u}_1 (1), \mathbf{u}_2 (1,1,w_1'), \mathbf{U}_3 (p_3),
\mathbf{X}(1,1,w_1',p_3,p_1,p_1'),\mathbf{y}_1) \in T^n_{\epsilon}(P_{U_1 U_2 U_3 X Y_1}) \right\}\\
&\leq 2^{n(P_{1e} + P_1' + P_3)} 2^{-n(I(U_3,X;Y_1|U_1,U_2)-2\delta)}.
\end{align}
For $\Pr\{{\tt E}_4\} \leq \epsilon'$, we require
\begin{align}
    \nonumber P_{1e} + P_1' + P_3 & < I(U_3,X;Y_1|U_1,U_2) = I(X;Y_1|U_1,U_2) + I(U_3;Y_1|U_1,U_2,X) \\
    \nonumber & \stackrel{(a)}{=} I(X;Y_1|U_2) + I(U_3;Y_1|U_2,X)\\
    & \stackrel{(b)}{=}  I(X;Y_1|U_2), \label{pe1_4}
\end{align}
where the first term in (a) is due to $U_1 \to U_2 \to X \to Y_1$ and the second term is due to  $U_1 \to (U_2,U_3) \to X \to Y_1$, and (b) is due to $U_3 \to (U_2,X) \to Y_1$.

e) ${\tt E}_5: w_0 = 1, w_1 = 1, p_3 = 1, p_1 \neq 1$ with $\mathbf{u}_1$, $\mathbf{u}_2$, $\mathbf{u}_3$, $\mathbf{x}$ jointly typical with
$\mathbf{y}_1$.  Then, we have
\begin{align}
    \nonumber \Pr\{{\tt E}_5\} &\le \sum_{p_1 \neq 1\atop p_1'} \Pr \left\{ (\mathbf{u}_1 (1), \mathbf{u}_2 (1,1,w_1'), \mathbf{u}_3 (1), \mathbf{X}(1,1,w_1',1,p_1,p_1'),
    \mathbf{y}_1) \in T^n_{\epsilon}(P_{U_1 U_2 U_3 X Y_1}) \right\} \\
    &\leq 2^{n(P_{1e} + P_1')} 2^{-n(I(X;Y_1|U_1,U_2,U_3)-2\delta)}.
\end{align}
For $\Pr\{{\tt E}_5\} \leq \epsilon'$, we require
\begin{equation}\label{pe1_5}
P_{1e} + P_1' < I(X;Y_1|U_1,U_2,U_3) = I(X;Y_1|U_2,U_3)
\end{equation}
where the equality is due to $U_1 \to (U_2,U_3) \to X \to Y_1$.

e) ${\tt E}_6: w_0 = 1, w_1 \neq 1, p_3 = 1$ and $p_1$ arbitrary with $\mathbf{u}_1$, $\mathbf{u}_2$, $\mathbf{u}_3$, $\mathbf{x}$ jointly typical
with $\mathbf{y}_1$. Then, we have
\begin{align}
\nonumber \Pr\{{\tt E}_6\} &\le \sum_{w_1 \neq 1\atop p_1,w_1',p_1'} \Pr \left\{ (\mathbf{u}_1 (1), \mathbf{U}_2 (1,w_1,w_1'), \mathbf{u}_3 (1),
\mathbf{X}(1,w_1,w_1',1,p_1,p_1'),
\mathbf{y}_1) \in T^n_{\epsilon}(P_{U_1 U_2 U_3 X Y_1}) \right\} \\
&\leq 2^{n(R_{1e}+R_1'+P_{1e} + P_1')} 2^{-n(I(U_2,X;Y_1|U_1,U_3)-2\delta)}.
\end{align}
For $\Pr\{{\tt E}_6\} \leq \epsilon'$, we require
\begin{align}
    \nonumber R_{1e}+R_1'+P_{1e} + P_1' & < I(U_2,X;Y_1|U_1,U_3) = I(X;Y_1|U_1,U_3) + I(U_2;Y_1|U_1,U_3,X) \\
    \nonumber & \stackrel{(a)}{=} I(X;Y_1|U_3) + I(U_2;Y_1|U_3,X)\\
    & \stackrel{(b)}{=}  I(X;Y_1|U_3), \label{pe1_6}
\end{align}
where the first term of (a) is due to $U_1 \to U_3 \to X \to Y_1$ and the second term of (a) and (b) are due to $U_1 \to U_2 \to (U_3,X) \to Y_1$.  Consequently, under the conditions \eqref{pe1_2}, \eqref{pe1_3}, \eqref{pe1_4}, \eqref{pe1_5}, \eqref{pe1_6} listed above, the error probability at receiver 1 is less than $\sum_{i=1}^6 \Pr\{{\tt E}_i\}\leq 6\epsilon'$.

Now, assume that $(w_0,q_2)=(1,1)$ is sent to receiver 2.  At receiver 2, the decoder first finds $w_0$ by indirect decoding, then finds $w_1$ by
joint typical decoding.  The error events at receiver 2 may be divided into:

a) ${\tt E}'_1:(w_0,q_2)=(1,1)$ but $\mathbf{u}_2$ is not jointly typical with $\mathbf{y}_2$ (indirect decoding). In this case, by the properties of
strong typical sequences, we have $\Pr\{{\tt E}'_1\}\le\epsilon'$.

b) ${\tt E}'_2: w_0 \neq 1$, $q_2$ arbitrary and $\mathbf{u}_2$ is jointly typical with $\mathbf{y}_2$ (indirect decoding). This is the same as
receiver 2 trying to estimate $w_0$ such that $(\mathbf{u}_2 (w_0,q_2), \mathbf{y}_3) \in T^n_{\epsilon}(P_{U_2 Y_2})$ for any $q_2 \in \{ 1 , \dots,
2^{nQ_2} \}$. We have
\begin{equation}
\Pr\{{\tt E}'_2\} \le \sum_{w_0 \neq 1} \sum_{q_2} \Pr \{ (\mathbf{U}_2 (w_0,q_2), \mathbf{y}_2) \in T^n_{\epsilon}(P_{U_2 Y_2}) \}\leq
2^{n(R_0+Q_2)} 2^{-n(I(U_2;Y_2)-2\delta)}.
\end{equation}
Then, for $\Pr\{{\tt E}'_2\} \le \epsilon'$, we need
\begin{equation}\label{pe2_2}
R_0+Q_2 < I(U_2;Y_2).
\end{equation}

c) ${\tt E}'_3: w_0 = 1$, $q_2 \neq 1$, and $\mathbf{u}_1$, $\mathbf{u}_2$ are jointly typical with $\mathbf{y}_2$. Then, we have
\begin{equation}
\Pr\{{\tt E}'_3\} \le \sum_{q_2} \Pr \{ (\mathbf{u}_1 (1), \mathbf{U}_2 (1,q_2), \mathbf{y}_2) \in T^n_{\epsilon}(P_{U_1 U_2 Y_2}) \}\leq 2^{n Q_2}
2^{-n(I(U_2;Y_2|U_1)-2\delta)}.
\end{equation}
Then, for $\Pr\{{\tt E}'_2\} \le \epsilon'$, we need
\begin{equation}\label{pe2_3}
    Q_2 < I(U_2;Y_2|U_1).
\end{equation}
Thus, under the conditions \eqref{pe2_2} and \eqref{pe2_3}, the error probability at receiver 2 is less than $\sum_{i=1}^3 \Pr\{{\tt E}'_i\}\leq 3\epsilon'$.

At receiver 3, indirect decoding is used, so that the decoder estimates $w_0$ such that $(\mathbf{u}_3 (w_0,q_3), \mathbf{y}_3) \in
T^n_{\epsilon}(P_{U_3 Y_3})$ for any $q_3 \in \{ 1 , \dots, 2^{nQ_3} \}$. Assuming that $(w_0,q_3)=(1,1)$ is sent, we require
\begin{equation}\label{pe3}
R_0+Q_3<I(U_3;Y_3),
\end{equation}
for the error probability at receiver 3 to be small for $n$ sufficiently large.

In addition to the decoding conditions above, we require that
\begin{equation}\label{p1_part}
P_{1e}+P_1' > I(X;Y_1|U_2),
\end{equation}
which is a consequence of setting $P_1' = I(X;Y_1|U_2)-\delta_1$ as the partition size.

Combining \eqref{bin_rates}, \eqref{pe1_2}, \eqref{pe1_3}, \eqref{pe1_4}, \eqref{pe1_5}, \eqref{pe1_6}, \eqref{pe2_2}, \eqref{pe2_3}, \eqref{pe3} and
\eqref{p1_part} using Fourier-Motzkin elimination with $R_1 = R_{1e} + R_1'$, $R_2 = R_{2e} + R_2'$, $R_{2e} = P_{1e} + P_3$, we can obtain the inner
bound to the secrecy capacity region in Theorem \ref{thm1} as well as condition \eqref{thm1_cond}, which completes the proof.

\subsection{Proof of Achievability for 3-Receiver BC with 2 Degraded Message Sets (Type 1)}\label{ach2}
Here, we outline the proof of achievability for the Type 1 3-receiver BC with 2 degraded message sets and secrecy constraints. The coding scheme
largely follows that for the 3 degraded message sets case, but with the key difference being the assignment of the message $W_1$ using the auxiliary
codewords. Specifically, instead of encoding $W_1$ using the auxiliary codeword $\mathbf{U}_2$ and $W_2$ using $\mathbf{U}_3$ and $\mathbf{X}$ as in
the 3 degraded message sets case, here, $W_1$ is encoded using $\mathbf{U}_2$, $\mathbf{U}_3$ and $\mathbf{X}$. We can use the same code partitions
and sizes of the partitions for security as in the 3 degraded message sets case, even for this different coding scheme.

\underline{Codebook generation:}  Let us define $R_1\triangleq R_{1e} + R_1'$, $R_{1e}\triangleq P_{1e} + P_{2e} + P_3$, $R_1'\triangleq P_1'+P_2'$,
and
\begin{equation}\label{security_rates2}
\left\{\begin{aligned}
P_1'&\triangleq I(X;Y_3|U_2)-\delta_1, \\
P_2'&\triangleq I(U_2;Y_3|U_1)-\delta_1,
\end{aligned}\right.
\end{equation}
where $\delta_1 > 0$ and is small for $n$ sufficiently large.

The code generation follows the same way as in Section \ref{ach}, except that we randomly partition the sequences, $\mathbf{U}_2(w_0,q_2)$, into
$2^{n \tilde P_2}$ equally-sized bins, and $\mathbf{U}_3(w_0,q_3)$, into $2^{n \tilde P_3}$ equally-sized bins, where $\tilde P_2 = P_{2e} + P_2' +
P^{\dagger}_2$ and $\tilde P_3 = P_3 + P^{\dagger}_3$. The $\mathbf{U}_2$ codewords undergo a double partition while $\mathbf{U}_3$ undergo a single
partition. Then, for each product bin $(p_2,p_3)$ contains the joint typical pair $(\mathbf{u}_2(p_2,p_2',p_2^\dag),\mathbf{u}_3(p_3,p_3^\dag))$ for
$p_2\in\{1,\dots,2^{nP_{2e}}\}$, $p_2'\in\{1,\dots,2^{P_2'}\}$, $p_2^\dag\in\{1,\dots,2^{nP_2^\dag}\}$, $p_3\in\{1,\dots,2^{nP_3}\}$,
$p_3^\dag\in\{1,\dots,2^{nP_3^\dag}\}$ with high probability
\begin{equation}\label{bin_rates2}
\begin{aligned}
P_{2e}+P_2' &\le Q_2, \\
P_3 &\le Q_3,\\
P_{2e}+P_2' + P_3 &\le Q_2+Q_3-I(U_2;U_3|U_1).
\end{aligned}
\end{equation}

As before, for each joint typical pair $(\mathbf{u}_2 (p_2 , p_2'), \mathbf{u}_3 (p_3))$ corresponding to the product bin $(p_2,p_2',p_3)$, generate
$2^{n \tilde P_1}$ sequences of codewords $\mathbf{X}(w_0,p_1,p_1',p_2,p_2',p_3)$, where $\tilde P_1=P_{1e}+P_1'$, for
$p_1\in\{1,\dots,2^{nP_{1e}}\}$ and $p_1'\in\{1,\dots,2^{nP_1'}\}$, uniformly and randomly over the set of conditionally typical $\mathbf{X}$
sequences. The $2^{n\tilde P_1}$ codewords are partitioned into $2^{nP_{1e}}$ subcodes with $2^{nP_1'}$ codewords within the subcodes.

\underline{Encoding:} To send $(w_0,w_1)$, express $w_1$ by $(p_1,p_2,p_3)$ and send the codeword $\mathbf{x}(w_0,p_1,p_1',p_2,p_2',p_3)$.

\underline{Decoding:} Assume that $(w_0,p_1,p_2,p_3)=(1,1,1,1)$ is sent and $p_1'$, $p_2'$ can be arbitrary. The receivers decode the messages as
follows:
\begin{enumerate}
\item Receiver 1 uses joint typical decoding of $\{ \mathbf{u}_1, \mathbf{u}_2, \mathbf{u}_3, \mathbf{x}, \mathbf{y}_1 \}$ to find the indices $(w_0, p_1, p_2, p_3)$.
\item Receiver 2 uses indirect decoding of $\mathbf{u}_2$ to find the index $w_0$.
\item Receiver 3 uses indirect decoding of $\mathbf{u}_3$ to find the index $w_0$.
\end{enumerate}

At receiver 1, the decoder seeks the message so that
\begin{equation}
(\mathbf{u}_1(w_0),\mathbf{u}_2(p_2,p_2'),\mathbf{u}_3(p_3), \mathbf{x}(w_0,p_1,p_1',p_2,p_2',p_3),\mathbf{y}_1)\in T^n_{\epsilon}(P_{U_1 U_2 U_3 X
Y_1}).
\end{equation}
The error events at receiver 1 can be classified into:

a) ${\tt E}_1:(w_0,p_1,p_1',p_2,p_2',p_3)=(1,1,p_1',1,p_2',1)$ but $\mathbf{u}_1$, $\mathbf{u}_2$, $\mathbf{u}_3$, $\mathbf{x}$ are not jointly
typical with $\mathbf{y}_1$. In this case, we have $\Pr\{{\tt E}_1\}\le\epsilon\to 0$ for large $n$.

b) ${\tt E}_2:w_0\neq 1$, with arbitrary $p_1,p_2$ and $p_3$, but $\mathbf{u}_1,\mathbf{u}_2,\mathbf{u}_3$ and $\mathbf{x}$ are jointly typical with
$\mathbf{y}_1$. For $\Pr\{{\tt E}_2\}\le\epsilon\to 0$ with $n$ sufficiently large to be true, we then need
\begin{equation}\label{fm1}
R_0 + P_{1e}+ P_1' + P_{2e}+ P_2' + P_3 < I(X;Y_1).
\end{equation}

c) ${\tt E}_3 : w_0 = 1$, $p_2, p_3 \neq 1$, and $p_1$ arbitrary, but $\mathbf{u}_1$, $\mathbf{u}_2$, $\mathbf{u}_3$, $\mathbf{x}$ are jointly
typical with $\mathbf{y}_1$. For $\Pr\{{\tt E}_3\}\le\epsilon\to 0$ with $n$ sufficiently large to be true, we require
\begin{equation}\label{fm2}
P_{1e}+ P_1' + P_{2e}+ P_2' + P_3 < I(X;Y_1|U_1).
\end{equation}

d) ${\tt E}_4:w_0 = 1$, $p_2 = 1$, $p_3 \neq 1$, and $p_1$ arbitrary, but $\mathbf{u}_1$, $\mathbf{u}_2$, $\mathbf{u}_3$, $\mathbf{x}$ are jointly
typical with $\mathbf{y}_1$. Then, for $\Pr\{{\tt E}_4\}\le\epsilon\to 0$ with $n$ sufficiently large to be true, we need
\begin{align}\label{fm3}
P_{1e}+ P_1' + P_3 < I(X;Y_1|U_2).
\end{align}

e) ${\tt E}_5:w_0 = 1$, $p_2 \neq 1$, $p_3 =1$, and $p_1$ arbitrary, but $\mathbf{u}_1$, $\mathbf{u}_2$, $\mathbf{u}_3$, $\mathbf{x}$ are jointly
typical with $\mathbf{y}_1$. Then, for $\Pr\{{\tt E}_5\}\le\epsilon\to 0$ with $n$ sufficiently large to be true, we need
\begin{align}\label{fm4}
P_{1e}+ P_1' + P_{2e}+ P_2'  < I(X;Y_1|U_3).
\end{align}

f) ${\tt E}_6:w_0=1$, $p_2=1$, $p_3=1$ and $p_1\neq 1$, but $\mathbf{u}_1$, $\mathbf{u}_2$, $\mathbf{u}_3$, $\mathbf{x}$ are jointly typical with
$\mathbf{y}_1$. Then, for $\Pr\{{\tt E}_6\}\le\epsilon\to 0$ with $n$ sufficiently large to be true, we require
\begin{equation}\label{fm5}
P_{1e}+ P_1' <  I(X;Y_1|U_2,U_3).
\end{equation}
The error probability at receiver 1 is therefore less than $\sum_{i=1}^6 \Pr\{{\tt E}_i\}\leq 6\epsilon$.

At receivers 2 and 3, assuming that $(w_0,q_2)=(w_0,q_3)=(1,1)$ is sent, we require
\begin{equation}\label{fm6}
\left\{\begin{aligned}
R_0+Q_2 &<I(U_2;Y_2),\\
R_0+Q_3 &<I(U_3;Y_3),
\end{aligned}\right.
\end{equation}
for the error probabilities tending to 0 for $n$ sufficiently large.  We additionally have
\begin{equation}\label{p1_part2}
P_{1e}+P_1' > I(X;Y_1|U_2),
\end{equation}
which is a consequence of setting $P_1' = I(X;Y_1|U_2)-\delta_1$ as the partition size.

Combining \eqref{bin_rates2} and \eqref{fm1} to \eqref{fm6} and \eqref{p1_part2} by using Fourier-Motzkin elimination with $R_1 = R_{1e} + R_1'$,
$R_{1e} = P_{1e} + P_{2e} + P_3$, we can obtain the rate region in Theorem \ref{thm2} and the conditions \eqref{cor2_cond}.

\subsection{Equivocation Calculation for 3-Receiver BC with 3 Degraded Message Sets}\label{eqv_calc}
In this section, we show that the equivocation rate for the 3-receiver BC with 3 degraded message sets satisfies the security conditions in
\eqref{sec_cond}. That is, we shall derive the bounds for $H(W_1|\mathbf{Y}_3)$, $H(W_2|\mathbf{Y}_3)$ and $H(W_1,W_2|\mathbf{Y}_3)$. In the
analysis, we shall make use of the following relation very frequently
\begin{equation}\label{ent_exp}
H(U,V) = H(U) + H(V|U).
\end{equation}
For the message $W_1$, the equivocation can be bounded in two ways, which respectively correspond to whether $\mathbf{U}_2$ is the codeword sent to
$Y_2$ or $\mathbf{X}$ is the codeword sent to $Y_1$.  For the former case, we have
\begin{align}
\nonumber  H(W_1|\mathbf{Y}_3) & \geq H(W_1|\mathbf{Y}_3, \mathbf{U}_1) \\
\nonumber  & \stackrel{(a)}{=} H(W_1,\mathbf{Y}_3|\mathbf{U}_1) - H(\mathbf{Y}_3|\mathbf{U}_1) \\
\nonumber  & \stackrel{(b)}{=} H(W_1,\mathbf{U}_2,\mathbf{Y}_3|\mathbf{U}_1) - H(\mathbf{U}_2|W_1,\mathbf{U}_1,\mathbf{Y}_3)- H(\mathbf{Y}_3|\mathbf{U}_1) \\
\nonumber  & \geq H(\mathbf{U}_2|\mathbf{U}_1) + [H(\mathbf{Y}_3|\mathbf{U}_2,\mathbf{U}_1)- H(\mathbf{Y}_3|\mathbf{U}_1)] - H(\mathbf{U}_2|W_1,\mathbf{U}_1,\mathbf{Y}_3) \\
  & = H(\mathbf{U}_2|\mathbf{U}_1) - I(\mathbf{U}_2;\mathbf{Y}_3|\mathbf{U}_1) - H(\mathbf{U}_2|W_1,\mathbf{U}_1,\mathbf{Y}_3),  \label{eqv_1}
\end{align}
where (a) is by \eqref{ent_exp}, and (b) has first two terms by \eqref{ent_exp}. Now, we can bound each term in \eqref{eqv_1} separately. For the
first term, given $\mathbf{u}_1$, $\mathbf{U}_2$ has $2^{nI(U_2;Y_2|U_1)}$ codewords with equal probability. As such,
\begin{equation}\label{eqv1_t1}
H(\mathbf{U}_2|\mathbf{U}_1) = nI(U_2;Y_2|U_1)- n\delta_1',
\end{equation}
where $\delta_1' > 0$ and is small for $n$ sufficiently large.  The second term can be bounded by \cite{liu_08}
\begin{align}\label{eqv1_t2}
I(\mathbf{U}_2; \mathbf{Y}_3 | \mathbf{U}_1 ) \leq nI(U_2 ; Y_3 | U_1 ) + n\delta',
\end{align}
where $\delta'>0$ and is small for $n$ sufficiently large. For the third term, by Fano's inequality, we have
\begin{equation}\label{fano_eqv1}
    \frac{1}{n}H(\mathbf{U}_2|W_1,\mathbf{U}_1,\mathbf{Y}_3) \leq \frac{1}{n}(1+\lambda(w_1')\log R_1') \triangleq \epsilon'_{2,n},
\end{equation}
where $\epsilon'_{2,n} \to 0$ for $n$ sufficiently large.

To show that $\lambda(w_1') \leq 2 \kappa$ where $\kappa \to 0$ for $n$ sufficiently large so that \eqref{fano_eqv1} holds, consider decoding at the
wiretapper and the codebook with rate $R_1'$ to be decoded at the wiretapper with error probability $\lambda(w_1')$.  Let $W_1 = w_1$ and $W_0 = w_0$
be fixed.  We note that the wiretapper decodes $\mathbf{U}_2$ first as it will then use this knowledge to decode $\mathbf{X}$ later.  The wiretapper
decodes $\mathbf{U}_2$ given $W_1 = w_1$ and $\mathbf{U}_1 = \mathbf{u}_1$, by finding the index $w_1'$, so that
\begin{equation}
\left(\mathbf{u}_1(w_0),\mathbf{u}_2(w_0,w_1,w_1'),\mathbf{y}_3\right) \in T^n_{\epsilon}(P_{U_1 U_2 Y_3}).
\end{equation}
If there is none or more than one possible codeword, an error is declared. Now, define the event
\begin{equation}
    {\tt E}^{(Y_3)}_1 (w_1') \triangleq \{ \mathbf{u}_1(w_0), \mathbf{U}_2(w_0,w_1,w_1'), \mathbf{y}_3 \in  T^n_{\epsilon}(P_{U_1 U_2 Y_3}) \}.
\end{equation}
Then, assuming that $\mathbf{u}_2(w_0,w_1,1)$ is sent,
\begin{equation}
\lambda(w_1')\leq \Pr \left\{ \left( {\tt E}^{(Y_3)}_1 (1)\right)^c \right\} + \sum_{w_1'} \Pr \left\{ {\tt E}^{(Y_3)}_1 (1) \right\}\leq \kappa +
2^{nR_1'}2^{-n(I(U_2;Y_3|U_1)-2\delta)},
\end{equation}
where $\delta \to 0$ as $\epsilon \to 0$ for $n$ sufficiently large. Thus, since we have chosen $R_1' = I(U_2;Y_3|U_1) - \delta_1$ for the
double-binning partition, we get $\lambda(w_1') \leq 2 \kappa$ for $\delta_1 > 2\delta$ and \eqref{fano_eqv1} holds. Substituting
\eqref{eqv1_t1}--\eqref{fano_eqv1} into \eqref{eqv_1}, we have $H(W_1|\mathbf{Y}_3) \geq nR_{1e} - n\epsilon_1$, where $\epsilon_1 = \delta_1' +
\delta' + \epsilon'_{2,n}$, and hence the equivocation rate satisfies the first condition in \eqref{sec_cond1}.

For message $W_1$ sent using $\mathbf{X}$ to $Y_1$, we have
\begin{align}
\nonumber   H(W_1|\mathbf{Y}_3) & \geq  H(W_1,\mathbf{Y}_3|\mathbf{U}_1) - H(\mathbf{Y}_3|\mathbf{U}_1) \\
\nonumber & = H(W_1,\mathbf{X},\mathbf{Y}_3|\mathbf{U}_1) - H(\mathbf{X}|W_1,\mathbf{U}_1,\mathbf{Y}_3) - H(\mathbf{Y}_3|\mathbf{U}_1) \\
\nonumber & \geq H(\mathbf{X}|\mathbf{U}_1) + H(\mathbf{Y}_3|\mathbf{U}_1,\mathbf{X}) - H(\mathbf{Y}_3|\mathbf{U}_1) - H(\mathbf{X}|W_1,\mathbf{U}_1,\mathbf{Y}_3) \\
\nonumber & \geq H(\mathbf{X}|\mathbf{U}_1,\mathbf{U}_2,\mathbf{U}_3) + H(\mathbf{Y}_3|\mathbf{U}_1,\mathbf{U}_2,\mathbf{X}) - H(\mathbf{Y}_3|\mathbf{U}_1) - H(\mathbf{U}_2,\mathbf{X}|W_1,\mathbf{U}_1,\mathbf{Y}_3) \\
& = H(\mathbf{X}|\mathbf{U}_1,\mathbf{U}_2,\mathbf{U}_3) - I(\mathbf{U}_2,\mathbf{X};\mathbf{Y}_3|\mathbf{U}_1) - H(\mathbf{U}_2|W_1,\mathbf{U}_1,\mathbf{Y}_3) -  H(\mathbf{X}|W_1,\mathbf{U}_1,\mathbf{U}_2,\mathbf{Y}_3) \label{eqv_1a}.
\end{align}
For the first term in \eqref{eqv_1a}, given $\mathbf{u}_1,\mathbf{u}_2,\mathbf{u}_3$, $\mathbf{X}$ has $2^{nI(X;Y_1|U_2,U_3,U_1)}$ codewords with
equal probability. Then,
\begin{align}\label{ent_x}
\nonumber H(\mathbf{X}|\mathbf{U}_1,\mathbf{U}_2,\mathbf{U}_3) & = nI(X;Y_1|U_2,U_3,U_1) - n\delta_1' \\
& = nI(X;Y_1|U_2) - n\delta_1' \;\; \textrm{or} \;\; nI(X;Y_1|U_3) - n\delta_1'.
\end{align}
The last equalities are due to $I(X;Y_1|U_2,U_3,U_1) = I(X;Y_1|U_1) - I(U_2,U_3;Y_1|U_1)$ and
\begin{align}
I(U_2,U_3;Y_1|U_1) = I(U_2;Y_1|U_1) + I(U_3;Y_1|U_2,U_1) = I(U_2;Y_1|U_1), \\
I(U_2,U_3;Y_1|U_1) = I(U_3;Y_1|U_1) + I(U_2;Y_1|U_3,U_1) = I(U_3;Y_1|U_1),
\end{align}
where the above equalities are due to the Markov chain conditions \eqref{mkv}. Thus, for this case, we choose
\begin{equation}\label{eqv_1a_t1}
H(\mathbf{X}|\mathbf{U}_1,\mathbf{U}_2,\mathbf{U}_3) = nI(X;Y_1|U_3) - n\delta_1'.
\end{equation}
The second term in \eqref{eqv_1a} can be bounded as
\begin{equation}\label{eqv_1a_t2}
I(\mathbf{U}_2,\mathbf{X};\mathbf{Y}_3|\mathbf{U}_1)= I(\mathbf{U}_2;\mathbf{Y}_3|\mathbf{U}_1) +
I(\mathbf{X};\mathbf{Y}_3|\mathbf{U}_2,\mathbf{U}_1) \leq nI(U_2;Y_3|U_1) + nI(X;Y_3|U_2) + 2n\delta'.
\end{equation}
The third term in \eqref{eqv_1a} may be bounded using Fano's inequality as in \eqref{fano_eqv1}. The fourth term can also be bounded using Fano's
inequality, by which we have
\begin{equation}\label{fano_eqv1a}
    \frac{1}{n}H(\mathbf{X}|W_1,\mathbf{U}_1,\mathbf{U}_2,\mathbf{Y}_3) \leq \frac{1}{n}(1+\lambda(p_1')\log P_1') \triangleq \epsilon'_{1,n},
\end{equation}
where $\epsilon'_{1,n} \to 0$ for $n$ sufficiently large.  To show that $\lambda(p_1') \leq 2\kappa$ so that \eqref{fano_eqv1a} holds, assume that wiretapper $Y_3$ knows $\mathbf{U}_2= \mathbf{u}_2$, $\mathbf{U}_1 = \mathbf{u}_1$ and decodes $\mathbf{x}(w_0,w_1,w_1',p_3,p_1,p_1')$ by finding the index $p_1'$, so that
\begin{equation}
\left(\mathbf{u}_1(w_0),\mathbf{u}_2(w_0,w_1,w_1'),\mathbf{u}_3(p_3),\mathbf{x}(w_0,w_1,w_1',p_3,p_1,p_1'),\mathbf{y}_3\right) \in T^n_{\epsilon}(P_{U_1 U_2 U_3 X Y_3}).
\end{equation}
If there is none or more than one possible codeword, an error is declared.  Define the event
\begin{equation}
    {\tt E}^{(Y_3)}_2 (p_1') \triangleq \{ \mathbf{u}_1(w_0), \mathbf{u}_2(w_0,w_1,w_1'), \mathbf{U}_3(p_3), \mathbf{X}(w_0,w_1,w_1',p_3,p_1,p_1'), \mathbf{y}_3 \in  T^n_{\epsilon}(P_{U_1 U_2 U_3 X Y_3}) \},
\end{equation}
where $w_0, w_1,w_1'$ are known. Assuming that $\mathbf{x}(w_0,w_1,w_1',p_3,p_1,1)$ is sent, we then have
\begin{equation}
\lambda(p_1')\leq \Pr \left\{ \left( {\tt E}^{(Y_3)}_2 (1)\right)^c \right\} + \sum_{p_1'} \Pr \left\{ {\tt E}^{(Y_3)}_2 (1) \right\}\leq \kappa +
2^{nP_1'}2^{-n(I(X;Y_3|U_1, U_2)-2\delta)},
\end{equation}
where $\delta \to 0$ as $\epsilon \to 0$ for $n$ sufficiently large.  Since we have chosen $P_1' = I(X;Y_3|U_2) - \delta_1$, we obtain $\lambda(p_1')
\leq 2 \kappa$ for $\delta_1 > 2\delta$. Thus, \eqref{fano_eqv1a} holds and substituting \eqref{eqv_1a_t1}, \eqref{eqv_1a_t2}, \eqref{fano_eqv1},
\eqref{fano_eqv1a} into \eqref{eqv_1a}, we have $H(W_1|\mathbf{Y}_3) \geq nR_{1e} - n\tilde{\epsilon}_1$, where $n\tilde{\epsilon}_1 = \delta_1' +
2\delta' + \epsilon'_{1,n} + \epsilon'_{2,n}$ is small for $n$ sufficiently large, so the second condition in \eqref{sec_cond1} is satisfied.

For the message $W_2$, the equivocation can be bounded by
\begin{align}
\nonumber  H(W_2|\mathbf{Y}_3) & \geq H(W_2|\mathbf{Y}_3, \mathbf{U}_1,\mathbf{U}_2) \\
\nonumber  & = H(W_2,\mathbf{Y}_3|\mathbf{U}_1,\mathbf{U}_2) - H(\mathbf{Y}_3|\mathbf{U}_1,\mathbf{U}_2)  \\
\nonumber  & = H(W_2,\mathbf{X},\mathbf{Y}_3|\mathbf{U}_1,\mathbf{U}_2) - H(\mathbf{X}|W_2,\mathbf{U}_1,\mathbf{U}_2,\mathbf{Y}_3) - H(\mathbf{Y}_3|\mathbf{U}_1,\mathbf{U}_2) \\
\nonumber  & \geq H(\mathbf{X}|\mathbf{U}_1,\mathbf{U}_2) + [H(\mathbf{Y}_3|\mathbf{U}_1,\mathbf{U}_2,\mathbf{X}) -  H(\mathbf{Y}_3|\mathbf{U}_1,\mathbf{U}_2)] - H(\mathbf{X}|W_2,\mathbf{U}_1,\mathbf{U}_2,\mathbf{Y}_3) \\
  & \geq H(\mathbf{X}|\mathbf{U}_1,\mathbf{U}_2,\mathbf{U}_3)  - I(\mathbf{X};\mathbf{Y}_3|\mathbf{U}_1,\mathbf{U}_2) - H(\mathbf{X}|W_2,\mathbf{U}_1,\mathbf{U}_2,\mathbf{Y}_3).  \label{eqv_2}
\end{align}
For the first term in \eqref{eqv_2}, given $\mathbf{u}_1$, $\mathbf{u}_2$, $\mathbf{u}_3$, $\mathbf{X}$ has $2^{nI(X;Y_1|U_2,U_3,U_1)}$ codewords
with equal probability. Thus,
\begin{equation}\label{eqv2_t1}
    H(\mathbf{X}|\mathbf{U}_1, \mathbf{U}_2,\mathbf{U}_3) = nI(X;Y_1|U_2) - n\delta_1',
\end{equation}
as discussed in the obtaining of \eqref{ent_x}. The second term is bounded by
\begin{equation}\label{eqv2_t2}
I(\mathbf{X};\mathbf{Y}_3|\mathbf{U}_2,\mathbf{U}_3) \leq nI(X ; Y_3 | U_2,U_3 ) + n\delta',
\end{equation}
where $\delta'>0$ and is small for $n$ sufficiently large.  For the third term, by Fano's inequality, we have
\begin{equation}\label{fano_eqv2}
    \frac{1}{n}H(\mathbf{X}|W_2,\mathbf{U}_2,\mathbf{U}_3,\mathbf{Y}_3) \leq \frac{1}{n}(1+\lambda(p_1')\log P_1') \triangleq \epsilon'_{3,n},
\end{equation}
where $\epsilon'_{3,n} \to 0$ for $n$ sufficiently large. To show that $\lambda(p_1') \leq 2\kappa$ so that \eqref{fano_eqv2} holds, since the
wiretapper knows $W_2$, we can assume that wiretapper $Y_3$ knows $\mathbf{U}_3 = \mathbf{u}_3$, $\mathbf{U}_2 = \mathbf{u}_2$, $\mathbf{U}_1 =
\mathbf{u}_1$ and decodes $\mathbf{x}(w_0,w_1,w_1',p_3,p_1,p_1')$ by finding the index $p_1'$, such that
\begin{equation}
\left(\mathbf{u}_1(w_0),\mathbf{u}_2(w_0,w_1,w_1'),\mathbf{u}_3(p_3),\mathbf{x}(w_0,w_1,w_1',p_3,p_1,p_1'),\mathbf{y}_3\right) \in T^n_{\epsilon}(P_{U_1 U_2 U_3 X Y_3}).
\end{equation}
If there is none or more than one possible codeword, an error is declared.  Define the event
\begin{equation}
    {\tt E}^{(Y_3)}_2 (p_1') \triangleq \{ \mathbf{u}_1(w_0), \mathbf{u}_2(w_0,w_1,w_1'), \mathbf{u}_3(p_3), \mathbf{X}(w_0,w_1,w_1',p_3,p_1,p_1'), \mathbf{y}_3 \in  T^n_{\epsilon}(P_{U_1 U_2 U_3 X Y_3}) \},
\end{equation}
where $w_0, w_1,w_1',p_3$ are known.  Assuming that $\mathbf{x}(w_0,w_1,w_1',p_3,p_1,1)$ is sent, we then have
\begin{equation}
\lambda(p_1')\leq \Pr \left\{ \left( {\tt E}^{(Y_3)}_2 (1)\right)^c \right\} + \sum_{p_1'} \Pr \left\{ {\tt E}^{(Y_3)}_2 (1) \right\}\leq \kappa +
2^{nP_1'}2^{-n(I(X;Y_3|U_1, U_2,U_3)-2\delta)},
\end{equation}
where $\delta \to 0$ as $\epsilon \to 0$ for $n$ sufficiently large.  Since we have chosen
\begin{equation}
P_1' = I(X;Y_3|U_2) - \delta_1 = I(X;Y_3|U_1,U_2,U_3) - \delta_1,
\end{equation}
we obtain $\lambda(p_1') \leq 2 \kappa$ for $\delta_1 > 2\delta$ and \eqref{fano_eqv2} holds. Substituting \eqref{eqv2_t1}, \eqref{eqv2_t2} and
\eqref{fano_eqv2} into \eqref{eqv_2}, we have $H(W_2|\mathbf{Y}_3) \geq nR_{2e} - n\epsilon_2$, where $\epsilon_2 = \delta_1'+\delta' +
\epsilon'_{3,n}$, and the equivocation rate satisfies \eqref{sec_cond2}.

For the combined message $(W_1,W_2)$, we have
\begin{align}
\nonumber H(W_1,W_2|\mathbf{Y}_3) & \geq H(W_1,W_2|\mathbf{Y}_3,\mathbf{U}_1) \\
\nonumber  & = H(W_1,W_2,\mathbf{Y}_3|\mathbf{U}_1) - H(\mathbf{Y}_3|\mathbf{U}_1) \\
\nonumber  & = H(W_1,W_2,\mathbf{X},\mathbf{Y}_3|\mathbf{U}_1) - H(\mathbf{X}|W_1,W_2,\mathbf{U}_1,\mathbf{Y}_3) - H(\mathbf{Y}_3|\mathbf{U}_1)\\
\nonumber  & \geq H(\mathbf{X}|\mathbf{U}_1) + H(\mathbf{Y}_3|\mathbf{U}_1,\mathbf{X}) - H(\mathbf{Y}_3|\mathbf{U}_1) - H(\mathbf{U}_2,\mathbf{X}|W_1,W_2,\mathbf{U}_1,\mathbf{Y}_3) \\
\nonumber & \geq H(\mathbf{X}|\mathbf{U}_1) + [H(\mathbf{Y}_3|\mathbf{U}_1,\mathbf{U}_2,\mathbf{X}) - H(\mathbf{Y}_3|\mathbf{U}_1)] -
H(\mathbf{U}_2|W_1,W_2,\mathbf{U}_1,\mathbf{Y}_3)\\
&- H(\mathbf{X}|W_1,W_2,\mathbf{U}_1,\mathbf{U}_2,\mathbf{Y}_3). \label{eqv_3}
\end{align}
For the first term, we have
\begin{equation}
H(\mathbf{X}|\mathbf{U}_1) = nI(X;Y_1|U_1) - n\delta_1'.
\end{equation}
The second term can be bounded by
\begin{equation}
I(\mathbf{U}_2,\mathbf{X};\mathbf{Y}_3|\mathbf{U}_1) \leq nI(U_2;Y_3|U_1) + nI(X;Y_3|U_2) + 2n\delta'.
\end{equation}
The fourth and fifth terms are, respectively,
\begin{align}
\frac{1}{n}H(\mathbf{U}_2|W_1,W_2,\mathbf{U}_1,\mathbf{Y}_3) \leq \frac{1}{n}H(\mathbf{U}_2|W_1,\mathbf{U}_1,\mathbf{Y}_3) \leq \epsilon'_{2,n}, \\
\frac{1}{n}H(\mathbf{X}|W_1,W_2,\mathbf{U}_1,\mathbf{U}_2,\mathbf{Y}_3) \leq \frac{1}{n}H(\mathbf{X}|W_1,\mathbf{U}_1,\mathbf{U}_2,\mathbf{Y}_3) \leq \epsilon'_{1,n}.
\end{align}
Substituting the above into \eqref{eqv_3}, we get
\begin{equation}
H(W_1,W_2|\mathbf{Y}_3) \geq nR_{1e} + nR_{2e} - n\epsilon_{1,2},
\end{equation}
where $\epsilon_{1,2} = \delta_1' + 2\delta' + \epsilon'_{1,n} + \epsilon'_{2,n}$, thus satisfying \eqref{sec_cond12}. As a result, we see that the
security conditions in \eqref{sec_cond} are satisfied and we have shown that the rate-equivocation tuple $(R_0,R_1,R_{1e},R_2,R_{2e})$ is achievable.


\section{Outer Bounds for the 3-Receiver BC with Degraded Message Sets}\label{outer}
In the derivation of the outer bounds, we note that, for the original Markov chain conditions
\begin{subequations}\label{mkv_o}
\begin{align}
    & U_1\to U_2\to (U_3,X) \to (Y_1,Y_2,Y_3), \label{mkv1_o} \\
    & U_1\to U_3\to (U_2,X) \to (Y_1,Y_2,Y_3),  \label{mkv2_o} \\
    & U_1 \to (U_2,U_3) \to X \to (Y_1,Y_2,Y_3), \label{mkv3_o}
\end{align}
\end{subequations}
which arise from the code generation for the 3-receiver BC, there exists the set of conditions
\begin{subequations}\label{mkv_appx_o}
\begin{align}
    & U_1\to \tilde U_2 \to U_2 \to (U_3,X) \to (Y_1,Y_2,Y_3), \label{mkv1_appx_o} \\
    & U_1\to U_3\to (\tilde U_2,U_2,X) \to (Y_1,Y_2,Y_3),  \label{mkv2_appx_o} \\
    & U_1 \to (\tilde U_2,U_2,U_3) \to X \to (Y_1,Y_2,Y_3), \label{mkv3_appx_o}
\end{align}
\end{subequations}
which come about by inserting auxiliary random variable $\tilde U_2$ between $U_1$ and $U_2$ in the code generation, so that $U_1\to \tilde U_2 \to U_2$ is satisfied.  The code generation and decoding conditions are equivalent for the following:
\begin{enumerate}
	\item For the 3-receiver BC with 3 degraded message sets, let $\tilde U_2$ represent information about $W_0$, and set $\tilde U_2 = U_1$ for equivalent code generation and decoding conditions under \eqref{mkv_o} and \eqref{mkv_appx_o};
	\item For the 3-receiver BC with 3 degraded message sets, let $\tilde U_2$ represent information about $W_1$, and set $\tilde U_2 = U_2$ for equivalent code generation and decoding conditions under both \eqref{mkv_o} and \eqref{mkv_appx_o};
	\item For the 3-receiver BC with 2 degraded message sets (Type 1), let $\tilde U_2$ represent information about $W_1$, and set $\tilde U_2 = U_1$ for equivalent code generation and decoding conditions under \eqref{mkv_o} and \eqref{mkv_appx_o}.
\end{enumerate}
We will show that case (1) is true in the Appendix of this paper, while cases (2) and (3) are shown to be true in [\ref{nair_eg}, Appendix III].

So, to obtain the outer bound to the rate equivocation region for the 3-receiver BC with 3 degraded message sets, we first find the outer bound $\mathcal{R}'_O$ for the 3-receiver BC using conditions \eqref{mkv_appx_o}.  Then we set $\tilde U_2 = U_1$ (as in case (1)) to obtain the outer bound to the rate equivocation region for the 3-receiver BC with 3 degraded message sets $\mathcal{R}_O$ with original conditions \eqref{mkv_o}.

For the 3-receiver BC with 2 degraded message sets (Type 1), we use the same procedure.

\newtheorem{relation}{Relation}

\subsection{Proof for the 3-receiver BC with 3 degraded message sets}\label{outer_gen}
In this section we show the proof for the outer bound in Theorem \ref{thm2}.  We use a $(2^{nR_0},2^{nR_1},2^{nR_2},n)$-code with error probability $P_e^{(n)}$ with the code construction so that we have the Markov chain condition $(W_0 , W_1, W_2)$ $\to \mathbf{X}\to (\mathbf{Y}_1,\mathbf{Y}_2,\mathbf{Y}_3)$.  Then, the probability distribution on ${\cal W}_0 \times {\cal W}_1 \times {\cal W}_2 \times \mathcal{X}^n \times \mathcal{Y}_1^n \times \mathcal{Y}_2^n \times \mathcal{Y}_3^n$ is given by
\begin{equation}\label{code_construct_3}
p(w_0 )p(w_1 )p(w_2)p(\mathbf{x}|w_0 ,w_1,w_2 ) \prod_{i=1}^n p(y_{1i}, y_{2i},y_{3i} | x_i ).
\end{equation}

By Fano's inequality, we have
\begin{equation}
\left\{\begin{aligned}
H(W_0 | \mathbf{Y}_k) & \leq nR_0 P_e^{(n)} +1 \triangleq n\gamma_k,~k=1,2,3,\\
H(W_0 , W_1 | \mathbf{Y}_1) & \leq n(R_0 +R_1)P_e^{(n)} +1 \triangleq n\gamma_4 , \\
H(W_0 , W_1 | \mathbf{Y}_2) & \leq n(R_0 +R_1)P_e^{(n)} +1 \triangleq n\gamma_5 , \\
H(W_0 , W_2 | \mathbf{Y}_1) & \leq n(R_0 +R_2)P_e^{(n)} +1 \triangleq n\gamma_6 , \\
H(W_0 , W_1, W_2 | \mathbf{Y}_1) & \leq n(R_0 +R_1 + R_2)P_e^{(n)} +1 \triangleq n\gamma_7, \\
H(W_0 , W_2 | \mathbf{Y}_2) & \leq n(R_0 +R_2)P_e^{(n)} +1 \triangleq n\gamma_8,
\end{aligned}\right.
\end{equation}
where $\gamma_k \to 0$ if $P^{(n)}_e \to 0~\forall k$.  Now we want to define the auxiliary random variables
$U_{1,i}$, $\tilde U_{2,i}$, $U_{2,i}$,  $U_{3,i}$, satisfying the conditions
\begin{subequations}\label{mkv_oi_s}
\begin{align}
U_{1,i}\to \tilde U_{2,i} \to U_{2,i}\to (U_{3,i},X_i)\to (Y_{1,i}, Y_{2,i}, Y_{3,i}), \\
U_{1,i}\to U_{3,i}\to (\tilde U_{2,i},U_{2,i},X_i)\to (Y_{1,i}, Y_{2,i}, Y_{3,i}),  \\
U_{1,i}\to (\tilde U_{2,i},U_{2,i},U_{3,i}) \to X_i \to (Y_{1,i}, Y_{2,i}, Y_{3,i}),
\end{align}
\end{subequations}
for all $i$.  When we have derived the outer bounds for the rates for conditions \eqref{mkv_oi_s}, we can then set $\tilde U_{2,i} = U_{1,i}$ to obtain the rates for the original conditions
\begin{subequations}\label{mkv_oi_o}
\begin{align}
U_{1,i}\to U_{2,i}\to (U_{3,i},X_i)\to (Y_{1,i}, Y_{2,i}, Y_{3,i}), \label{mkv_oi_o_1}\\
U_{1,i}\to U_{3,i}\to (U_{2,i},X_i)\to (Y_{1,i}, Y_{2,i}, Y_{3,i}), \label{mkv_oi_o_2} \\
U_{1,i}\to (U_{2,i},U_{3,i}) \to X_i \to (Y_{1,i}, Y_{2,i}, Y_{3,i}). \label{mkv_oi_o_3}
\end{align}
\end{subequations}
Here, however, we will define the auxiliary random variables $U_{1,i}\triangleq(W_0,
\mathbf{Y}_1^{i-1})$, $\tilde U_{2,i}\triangleq(U_{1,i}, \tilde {\mathbf{Y}}_2^{i+1})$, $U_{2,i} = W_1$,  $U_{3,i}\triangleq(U_{1,i}, \tilde{\mathbf{Y}}_3^{i+1})$ which satisfy the conditions
\begin{subequations}\label{mkv_oi}
\begin{align}
U_{1,i}\to (\tilde U_{2,i},U_{2,i}) \to (U_{3,i},X_i)\to (Y_{1,i}, Y_{2,i}, Y_{3,i}), \label{mkv_oi_1}\\
U_{1,i}\to U_{3,i}\to (\tilde U_{2,i},U_{2,i},X_i)\to (Y_{1,i}, Y_{2,i}, Y_{3,i}), \label{mkv_oi_2} \\
U_{1,i}\to (\tilde U_{2,i},U_{2,i},U_{3,i}) \to X_i \to (Y_{1,i}, Y_{2,i}, Y_{3,i}), \label{mkv_oi_3}
\end{align}
\end{subequations}
for all $i$, which are weaker than and included in conditions \eqref{mkv_oi_s}.  By setting $\tilde U_{2,i} = U_{1,i}$ in \eqref{mkv_oi}, we still obtain the original conditions \eqref{mkv_oi_o}.  Thus we use \eqref{mkv_oi} in our subsequent derivation for the outer bound.

We first prove three relations which are a consequence of \eqref{mkv_oi}.

\begin{relation}\label{relation1}
$I(\tilde {\mathbf{Y}}_3^{i+1};Y_{k,i}|W_0,\mathbf{Y}_1^{i-1})=I(\tilde {\mathbf{Y}}_2^{i+1},\tilde {\mathbf{Y}}_3^{i+1};Y_{k,i}|W_0,\mathbf{Y}_1^{i-1})=I(\tilde {\mathbf{Y}}_2^{i+1};Y_{k,i}|W_0,\mathbf{Y}_1^{i-1}), k=1,2,3$.
\end{relation}
\begin{proof}
For any $Y_{k,i}$, $k=1,2,3$, we have
\begin{align}
	\nonumber I(\tilde {\mathbf{Y}}_3^{i+1};Y_{k,i}|W_0,\mathbf{Y}_1^{i-1}) & = I(\tilde {\mathbf{Y}}_2^{i+1},\tilde {\mathbf{Y}}_3^{i+1};Y_{k,i}|W_0,\mathbf{Y}_1^{i-1}) - I(\tilde {\mathbf{Y}}_2^{i+1}; Y_{k,i}|W_0,\mathbf{Y}_1^{i-1},\tilde {\mathbf{Y}}_3^{i+1}) \\
	\nonumber & \stackrel{(a)}{=}  I(\tilde {\mathbf{Y}}_2^{i+1},\tilde {\mathbf{Y}}_3^{i+1};Y_{k,i}|W_0,\mathbf{Y}_1^{i-1}) \\
	\nonumber & = I(\tilde {\mathbf{Y}}_2^{i+1};Y_{k,i}|W_0,\mathbf{Y}_1^{i-1}) + I(\tilde {\mathbf{Y}}_3^{i+1};Y_{k,i}|W_0,\mathbf{Y}_1^{i-1},\tilde {\mathbf{Y}}_2^{i+1}) \\
	 & \stackrel{(b)}{=} I(\tilde {\mathbf{Y}}_2^{i+1};Y_{k,i}|W_0,\mathbf{Y}_1^{i-1}),
\end{align}
where (a) is due to $I(\tilde {\mathbf{Y}}_2^{i+1}; Y_{k,i}|W_0,\mathbf{Y}_1^{i-1},\tilde {\mathbf{Y}}_3^{i+1})=I(\tilde{U}_{2,i};Y_{k,i}|U_{3,i}) = 0$ by \eqref{mkv_oi_1} and (b) is due to the fact that $I(\tilde {\mathbf{Y}}_3^{i+1};Y_{k,i}|W_0,\mathbf{Y}_1^{i-1},\tilde {\mathbf{Y}}_2^{i+1}) = I(U_{3,i};Y_{k,i}|\tilde{U}_{2,i})=0$ by \eqref{mkv_oi_2}.
\end{proof}

\begin{relation}\label{relation2}
$I(\tilde {\mathbf{Y}}_3^{i+1};Y_{k,i}|W_0,\mathbf{Y}_1^{i-1})=I(W_1,\tilde {\mathbf{Y}}_3^{i+1};Y_{k,i}|W_0,\mathbf{Y}_1^{i-1})=I(W_1;Y_{k,i}|W_0,\mathbf{Y}_1^{i-1}), k=1,2,3$.
\end{relation}
\begin{proof}
For any $Y_{k,i}$, $k=1,2,3$, we have
\begin{align}
	\nonumber I(\tilde {\mathbf{Y}}_3^{i+1};Y_{k,i}|W_0,\mathbf{Y}_1^{i-1}) & = I(W_1,\tilde {\mathbf{Y}}_3^{i+1};Y_{k,i}|W_0,\mathbf{Y}_1^{i-1}) - I(W_1; Y_{k,i}|W_0,\mathbf{Y}_1^{i-1},\tilde {\mathbf{Y}}_3^{i+1}) \\
	\nonumber & \stackrel{(a)}{=}  I(W_1,\tilde {\mathbf{Y}}_3^{i+1};Y_{k,i}|W_0,\mathbf{Y}_1^{i-1}) \\
	\nonumber & = I(W_1;Y_{k,i}|W_0,\mathbf{Y}_1^{i-1}) + I(\tilde {\mathbf{Y}}_3^{i+1};Y_{k,i}|W_0,W_1,\mathbf{Y}_1^{i-1}) \\
	 & \stackrel{(b)}{=} I(W_1;Y_{k,i}|W_0,\mathbf{Y}_1^{i-1}),
\end{align}
where (a) is due to $I(W_1; Y_{k,i}|W_0,\mathbf{Y}_1^{i-1},\tilde {\mathbf{Y}}_3^{i+1})=I(U_{2,i};Y_{k,i}|U_{3,i}) = 0$ by \eqref{mkv_oi_1} and (b) is due to the fact that $I(\tilde {\mathbf{Y}}_3^{i+1};Y_{k,i}|W_0,W_1,\mathbf{Y}_1^{i-1}) = I(U_{3,i};Y_{k,i}|U_{2,i},U_{1,i})=0$ by \eqref{mkv_oi_2}.
\end{proof}

\begin{relation}\label{relation3}
$I(\tilde {\mathbf{Y}}_2^{i+1};Y_{k,i}|W_0,W_1,\mathbf{Y}_1^{i-1})=I(\tilde {\mathbf{Y}}_3^{i+1};Y_{k,i}|W_0,W_1,\mathbf{Y}_1^{i-1}), k=1,2,3$.
\end{relation}
\begin{proof}
For any $Y_{k,i}$, $k=1,2,3$, we have
\begin{align}
	\nonumber I(\tilde {\mathbf{Y}}_2^{i+1};Y_{k,i}|W_0,W_1,\mathbf{Y}_1^{i-1}) & = I(\tilde {\mathbf{Y}}_2^{i+1},\tilde {\mathbf{Y}}_3^{i+1};Y_{k,i}|W_0,W_1,\mathbf{Y}_1^{i-1})-I(\tilde {\mathbf{Y}}_3^{i+1};Y_{k,i}|W_0,W_1,\mathbf{Y}_1^{i-1},\tilde {\mathbf{Y}}_2^{i+1}) \\
	\nonumber & \stackrel{(a)}{=} I(\tilde {\mathbf{Y}}_2^{i+1},\tilde {\mathbf{Y}}_3^{i+1};Y_{k,i}|W_0,W_1,\mathbf{Y}_1^{i-1}) \\
	\nonumber & = I(\tilde {\mathbf{Y}}_3^{i+1};Y_{k,i}|W_0,W_1,\mathbf{Y}_1^{i-1}) + I(\tilde {\mathbf{Y}}_2^{i+1};Y_{k,i}|W_0,W_1,\mathbf{Y}_1^{i-1},\tilde {\mathbf{Y}}_3^{i+1}) \\
	& \stackrel{(b)}{=} I(\tilde {\mathbf{Y}}_3^{i+1};Y_{k,i}|W_0,W_1,\mathbf{Y}_1^{i-1}),
\end{align}
where (a) is due to $I(\tilde {\mathbf{Y}}_3^{i+1};Y_{k,i}|W_0,W_1,\mathbf{Y}_1^{i-1},\tilde {\mathbf{Y}}_2^{i+1})=I(U_{3,i};Y_{k,i}|U_{2,i},\tilde U_{2,i})=0$ by \eqref{mkv_oi_2}; and (b) is by $I(\tilde {\mathbf{Y}}_2^{i+1};Y_{k,i}|W_0,W_1,\mathbf{Y}_1^{i-1},\tilde {\mathbf{Y}}_3^{i+1})=I(\tilde U_{2,i};Y_{k,i}|U_{2,i},U_{3,i})=H(Y_{k,i}|U_{2,i},U_{3,i})-H(Y_{k,i}|\tilde U_{2,i},U_{2,i},U_{3,i})=0$ by \eqref{mkv_oi_1}.
\end{proof}

We begin by proving the outer bounds to the equivocation rates.  For $R_{1e}$, we have two possible choices corresponding to whether $\mathbf{X}$ is sent to $Y_1$ or $\mathbf{U}_2$ is sent to $Y_2$.  For the first case, we have
\begin{align}
\nonumber & nR_{1e} \leq H(W_1 | \mathbf{Y}_3) + n\tilde{\epsilon}_1 \;\; \mbox{(by secrecy condition)}\\
\nonumber & = H(W_1 | \mathbf{Y}_3, W_0) + I(W_1 ; W_0 | \mathbf{Y}_3 ) + n\tilde{\epsilon}_1 \\
\nonumber & = H(W_1 | W_0 ) - I(W_1 ; \mathbf{Y}_3 | W_0) +  I(W_1 ; W_0 | \mathbf{Y}_3 ) + n\tilde{\epsilon}_1 \\
\nonumber & = I(W_1 ; \mathbf{Y}_1 | W_0) + H(W_1 | W_0 , \mathbf{Y}_1) - I(W_1 ; \mathbf{Y}_3 | W_0) +  I(W_1 ; W_0 | \mathbf{Y}_3 ) + n\tilde{\epsilon}_1 \\
\nonumber & = I(W_1 ; \mathbf{Y}_1 | W_0) - I(W_1 ; \mathbf{Y}_3 | W_0) + H(W_0 | \mathbf{Y}_3) - H(W_0 | W_1 , \mathbf{Y}_3) + H(W_1 | W_0 , \mathbf{Y}_1) + n\tilde{\epsilon}_1 \\
\nonumber & \leq  I(W_1 ; \mathbf{Y}_1 | W_0) - I(W_1 ; \mathbf{Y}_3 | W_0) + H(W_0 | \mathbf{Y}_3) + H(W_1 | W_0 , \mathbf{Y}_1) + n\tilde{\epsilon}_1 \\
 & \stackrel{(a)}{\le} I(W_1 ; \mathbf{Y}_1 | W_0) - I(W_1 ; \mathbf{Y}_3 | W_0) + n(\tilde{\epsilon}_1 + \gamma_3 + \gamma_4) \label{r1_cnv},
\end{align}
where (a) is by Fano's inequality. Expanding the first two terms of (a) by the chain rule, we obtain
\begin{subequations}
\begin{align}
I(W_1 ; \mathbf{Y}_1 | W_0)& = \sum_{i=1}^n I(W_1 ; Y_{1,i}| W_0 , \mathbf{Y}_1^{i-1}), \label{mi_cnv_exp1} \\
I(W_1 ; \mathbf{Y}_3 | W_0)& = \sum_{i=1}^n I(W_1 ; Y_{3,i}| W_0 , \tilde {\mathbf{Y}}_3^{i+1}). \label{mi_cnv_exp2}
\end{align}
\end{subequations}
Now we have
\begin{equation}\label{r1_initial}
nR_{1e} \leq \sum_{i=1}^n \left[
I(W_1;Y_{1,i}|W_0,\mathbf{Y}_1^{i-1})
-I(W_1;Y_{3,i} |W_0,\tilde{\mathbf{Y}}_3^{i+1})\right] +n(\tilde{\epsilon}_1+\gamma_3+\gamma_4).
\end{equation}
The terms under the summation can be bounded by
\begin{align}
\nonumber & I(W_1;Y_{1,i}|W_0,\mathbf{Y}_1^{i-1}) -I(W_1;Y_{3,i} |W_0,\tilde{\mathbf{Y}}_3^{i+1}) \\
 \nonumber & = I(W_1,\tilde{\mathbf{Y}}_3^{i+1};Y_{1,i}|W_0,\mathbf{Y}_1^{i-1}) -  I(W_1,\mathbf{Y}_1^{i-1};Y_{3,i}|W_0,\tilde{\mathbf{Y}}_3^{i+1}) - I(\tilde{\mathbf{Y}}_3^{i+1};Y_{1,i}|W_0,W_1,\mathbf{Y}_1^{i-1}) \\
 \nonumber & \;\;  + I(\mathbf{Y}_1^{i-1};Y_{3,i}|W_0,W_1,\tilde{\mathbf{Y}}_3^{i+1}) \\
 \nonumber & \stackrel{(a)}{=} I(W_1,\tilde{\mathbf{Y}}_3^{i+1};Y_{1,i}|W_0,\mathbf{Y}_1^{i-1}) - I(W_1,\mathbf{Y}_1^{i-1};Y_{3,i}|W_0,\tilde{\mathbf{Y}}_3^{i+1}) - I(\tilde{\mathbf{Y}}_3^{i+1};Y_{1,i}|W_0,W_1,\mathbf{Y}_1^{i-1}) \\
 \nonumber & \;\;  + I(\tilde{\mathbf{Y}}_3^{i+1};Y_{1,i}|W_0,W_1,\mathbf{Y}_1^{i-1}) \\
 \nonumber & = I(W_1,\tilde{\mathbf{Y}}_3^{i+1};Y_{1,i}|W_0,\mathbf{Y}_1^{i-1}) - I(W_1,\mathbf{Y}_1^{i-1};Y_{3,i}|W_0,\tilde{\mathbf{Y}}_3^{i+1}) \\
 \nonumber & = I(X_i,W_1,\tilde{\mathbf{Y}}_3^{i+1};Y_{1,i}|W_0,\mathbf{Y}_1^{i-1}) - I(X_i,W_1,\mathbf{Y}_1^{i-1};Y_{3,i}|W_0,\tilde{\mathbf{Y}}_3^{i+1}) \\
 \nonumber & \;\; - [I(X_i;Y_{1,i}|W_0,W_1,\mathbf{Y}_1^{i-1},\tilde{\mathbf{Y}}_3^{i+1}) - I(X_i;Y_{3,i}|W_0,W_1,\mathbf{Y}_1^{i-1},\tilde{\mathbf{Y}}_3^{i+1}) ] \\
 \nonumber & \stackrel{(b)}{\leq} I(X_i,W_1,\tilde{\mathbf{Y}}_3^{i+1};Y_{1,i}|W_0,\mathbf{Y}_1^{i-1}) - I(X_i,W_1,\mathbf{Y}_1^{i-1};Y_{3,i}|W_0,\tilde{\mathbf{Y}}_3^{i+1}) \\
 \nonumber & \stackrel{(c)}{=} I(X_i;Y_{1,i}|W_0,\mathbf{Y}_1^{i-1},\tilde{\mathbf{Y}}_3^{i+1}) - I(X_i;Y_{3,i}|W_0,\mathbf{Y}_1^{i-1},\tilde{\mathbf{Y}}_3^{i+1}) + I(\tilde{\mathbf{Y}}_3^{i+1};Y_{1,i}|W_0,\mathbf{Y}_1^{i-1}) \\
 \nonumber & \;\; - I(\mathbf{Y}_1^{i-1};Y_{3,i}|W_0,\tilde{\mathbf{Y}}_3^{i+1}) \\
 \nonumber & \stackrel{(d)}{=} I(X_i;Y_{1,i}|W_0,\mathbf{Y}_1^{i-1},\tilde{\mathbf{Y}}_3^{i+1}) - I(X_i;Y_{3,i}|W_0,\mathbf{Y}_1^{i-1},\tilde{\mathbf{Y}}_3^{i+1}) + I(\tilde{\mathbf{Y}}_3^{i+1};Y_{1,i}|W_0,\mathbf{Y}_1^{i-1}) \\
  \nonumber & \;\; - I(\tilde{\mathbf{Y}}_3^{i+1};Y_{1,i}|W_0,\mathbf{Y}_1^{i-1}) \\
  \nonumber & = I(X_i;Y_{1,i}|W_0,\mathbf{Y}_1^{i-1},\tilde{\mathbf{Y}}_3^{i+1}) - I(X_i;Y_{3,i}|W_0,\mathbf{Y}_1^{i-1},\tilde{\mathbf{Y}}_3^{i+1}) \\
  \nonumber & = I(X_i;Y_{1,i}|W_0,\mathbf{Y}_1^{i-1},\tilde{\mathbf{Y}}_3^{i+1}) - [I(X_i,\tilde{\mathbf{Y}}_3^{i+1};Y_{3,i}|W_0,\mathbf{Y}_1^{i-1})-I(\tilde{\mathbf{Y}}_3^{i+1};Y_{3,i}|W_0,\mathbf{Y}_1^{i-1})] \\
  \nonumber & \stackrel{(e)}{=} I(X_i;Y_{1,i}|W_0,\mathbf{Y}_1^{i-1},\tilde{\mathbf{Y}}_3^{i+1}) - [I(X_i;Y_{3,i}|W_0,\mathbf{Y}_1^{i-1})-I(\tilde{\mathbf{Y}}_2^{i+1};Y_{3,i}|W_0,\mathbf{Y}_1^{i-1})] \\
  & = I(X_i;Y_{1,i}|U_{3,i})-[I(X_i;Y_{3,i} |U_{1,i})-I(\tilde U_2 ; Y_{3,i} |U_{1,i})] \stackrel{(f)}{=} I(X_i;Y_{1,i}|U_{3,i})-I(X_i;Y_{3,i} |U_{1,i}), \label{r1_1_1}
\end{align}
where (a) has last term by [\ref{ck}, Lemma 7] so that
\[
\sum_{i=1}^n I(\tilde{\mathbf{Y}}_3^{i+1};Y_{1,i}|W_0,W_1,\mathbf{Y}_1^{i-1}) = \sum_{i=1}^n I(\mathbf{Y}_1^{i-1};Y_{3,i}|W_0,W_1,\tilde{\mathbf{Y}}_3^{i+1});
\]
(b) is due to $[I(X_i;Y_{1,i}|W_0,W_1,\mathbf{Y}_1^{i-1},\tilde{\mathbf{Y}}_3^{i+1}) - I(X_i;Y_{3,i}|W_0,W_1,\mathbf{Y}_1^{i-1},\tilde{\mathbf{Y}}_3^{i+1}) ] \geq 0$ by the fact that $Y_1$ is a more capable channel than $Y_3$ along with the fact that it may be verified using a functional dependency graph \cite{kramer_bk} that $(W_0,W_1,\mathbf{Y}_1^{i-1},\tilde{\mathbf{Y}}_3^{i+1}) \to X_i \to (Y_{1,i},Y_{3,i})$ forms a Markov chain, so the more capable channel condition is satisfied \cite{korner_77}; (c) is because $I(W_1;Y_{1,i}|W_0,\mathbf{Y}_1^{i-1},\tilde{\mathbf{Y}}_3^{i+1},X_i)=0$ by
\begin{align}\label{w1_indp}
\nonumber & I(W_1;Y_{1,i}|W_0,\mathbf{Y}_1^{i-1},\tilde{\mathbf{Y}}_3^{i+1},X_i) = H(Y_{1,i}|W_0,\mathbf{Y}_1^{i-1},\tilde{\mathbf{Y}}_3^{i+1},X_i) - H(Y_{1,i}|W_0,W_1,\mathbf{Y}_1^{i-1},\tilde{\mathbf{Y}}_3^{i+1},X_i) \\
 & =H(Y_{1,i}|W_0,\mathbf{Y}_1^{i-1},\tilde{\mathbf{Y}}_3^{i+1},X_i) - H(Y_{1,i}|W_0,\mathbf{Y}_1^{i-1},\tilde{\mathbf{Y}}_3^{i+1},X_i)=0,
\end{align}
where the second equality is obtained using the relation $W_1 \to X_i \to Y_i$ on the second term on the right-hand-side, and similarly $I(W_1;Y_{3,i}|W_0,\mathbf{Y}_1^{i-1},\tilde{\mathbf{Y}}_3^{i+1},X_i)=0$; (d) has last term by [\ref{ck}, Lemma 7] so that
\[
\sum_{i=1}^n I(\tilde{\mathbf{Y}}_3^{i+1};Y_{1,i}|W_0,\mathbf{Y}_1^{i-1}) = \sum_{i=1}^n I(\mathbf{Y}_1^{i-1};Y_{3,i}|W_0,\tilde{\mathbf{Y}}_3^{i+1});
\]
(e) has the first term in the square brackets by the fact that
$I(\tilde{\mathbf{Y}}_3^{i+1};Y_{3,i}|W_0,\mathbf{Y}_1^{i-1},X_i)=0$ since given $X_i$, $\tilde{\mathbf{Y}}_3^{i+1}$ is independent of $Y_{3,i}$ as may be seen using a functional dependency graph, and the second term in the square brackets $I(\tilde{\mathbf{Y}}_3^{i+1};Y_{3,i}|W_0,\mathbf{Y}_1^{i-1})=I(\tilde{\mathbf{Y}}_2^{i+1};Y_{3,i}|W_0,\mathbf{Y}_1^{i-1})$ by Relation \ref{relation1}; and (f) is by substituting $\tilde U_{2,i}= U_{1,i}$.

Then, we have
\begin{equation}\label{r1_final1}
nR_{1e} \leq \sum_{i=1}^n \left[ I(X_i;Y_{1,i}|U_{3,i})-I(X_i;Y_{3,i} |U_{1,i})\right] +n(\tilde{\epsilon}_1+\gamma_3+\gamma_4).
\end{equation}

Next consider the rate for $W_1$ sent to receiver $Y_2$.  We have, following \eqref{r1_cnv},
\begin{align}
\nonumber nR_{1e} & \le I(W_1 ; \mathbf{Y}_2 | W_0) - I(W_1 ; \mathbf{Y}_3 | W_0) + H(W_0 | \mathbf{Y}_3) + H(W_1 | W_0, \mathbf{Y}_2) + n\epsilon_1 \\
 & \stackrel{(a)}{\le} I(W_1 ; \mathbf{Y}_2 | W_0) - I(W_1 ; \mathbf{Y}_3 | W_0) + n(\epsilon_1 + \gamma_3 + \gamma_5) \label{r1_cnv2},
\end{align}
where (a) is by Fano's inequality.  For the first two terms in \eqref{r1_cnv2}, we have
\begin{align}
 \nonumber & I(W_1 ; \mathbf{Y}_2 | W_0) - I(W_1 ; \mathbf{Y}_3 | W_0) = \sum_{i=1}^n \left[ I(W_1;Y_{2,i}| W_0, \tilde{\mathbf{Y}}_2^{i+1}) - I(W_1;Y_{3,i}| W_0, \tilde{\mathbf{Y}}_3^{i+1}) \right] \\
 \nonumber & = \sum_{i=1}^n \left[ I(W_1,\mathbf{Y}_1^{i-1};Y_{2,i}| W_0, \tilde{\mathbf{Y}}_2^{i+1}) - I(W_1,\mathbf{Y}_1^{i-1};Y_{3,i}| W_0, \tilde{\mathbf{Y}}_3^{i+1}) - I(\mathbf{Y}_1^{i-1};Y_{2,i}|W_0,W_1,\tilde{\mathbf{Y}}_2^{i+1}) \right. \\
 \nonumber & \;\;\left. + I(\mathbf{Y}_1^{i-1};Y_{3,i}|W_0,W_1,\tilde{\mathbf{Y}}_3^{i+1}) \right] \\
 \nonumber & \stackrel{(a)}{=} \sum_{i=1}^n \left[ I(W_1,\mathbf{Y}_1^{i-1};Y_{2,i}| W_0, \tilde{\mathbf{Y}}_2^{i+1}) - I(W_1,\mathbf{Y}_1^{i-1};Y_{3,i}| W_0, \tilde{\mathbf{Y}}_3^{i+1}) - I(\tilde{\mathbf{Y}}_2^{i+1};Y_{1,i}|W_0,W_1,\mathbf{Y}_1^{i-1}) \right. \\
  & \;\;\left. + I(\tilde{\mathbf{Y}}_3^{i+1};Y_{1,i}|W_0,W_1,\mathbf{Y}_1^{i-1}) \right] \label{r1e_u2_1} \\
 \nonumber & \stackrel{(b)}{=} \sum_{i=1}^n \left[ I(W_1,\mathbf{Y}_1^{i-1};Y_{2,i}| W_0, \tilde{\mathbf{Y}}_2^{i+1}) - I(W_1,\mathbf{Y}_1^{i-1};Y_{3,i}| W_0, \tilde{\mathbf{Y}}_3^{i+1}) \right] \\
 \nonumber &  = \sum_{i=1}^n \left[ I(\mathbf{Y}_1^{i-1};Y_{2,i}| W_0, \tilde{\mathbf{Y}}_2^{i+1}) + I(W_1;Y_{2,i}| W_0,\mathbf{Y}_1^{i-1}, \tilde{\mathbf{Y}}_2^{i+1}) - I(X_i,W_1,\mathbf{Y}_1^{i-1};Y_{3,i}| W_0, \tilde{\mathbf{Y}}_3^{i+1}) \right. \\
 \nonumber & \;\; \left. + I(X_i;Y_{3,i}|W_0,W_1,\mathbf{Y}_1^{i-1},\tilde{\mathbf{Y}}_3^{i+1}) \right] \\
 \nonumber & \stackrel{(c)}{\leq} \sum_{i=1}^n \left[ I(\tilde{\mathbf{Y}}_2^{i+1};Y_{1,i}| W_0,\mathbf{Y}_1^{i-1}) + I(W_1;Y_{2,i}| W_0,\mathbf{Y}_1^{i-1}, \tilde{\mathbf{Y}}_2^{i+1}) - I(\mathbf{Y}_1^{i-1};Y_{3,i}| W_0, \tilde{\mathbf{Y}}_3^{i+1}) \right. \\
 & \;\; \left. - I(X_i;Y_{3,i}|W_0,\mathbf{Y}_1^{i-1},\tilde{\mathbf{Y}}_3^{i+1}) + I(X_i;Y_{3,i}|W_0,W_1,\mathbf{Y}_1^{i-1},\tilde{\mathbf{Y}}_3^{i+1}) \right] \\
 \nonumber & \stackrel{(d)}{=} \sum_{i=1}^n \left[ I(\tilde{\mathbf{Y}}_2^{i+1};Y_{1,i}| W_0,\mathbf{Y}_1^{i-1}) + I(W_1;Y_{2,i}| W_0,\mathbf{Y}_1^{i-1}, \tilde{\mathbf{Y}}_2^{i+1}) - I(\tilde{\mathbf{Y}}_3^{i+1};Y_{1,i}| W_0, \mathbf{Y}_1^{i-1})  \right. \\
 \nonumber & \;\;  \left. - I(X_i;Y_{3,i}|W_0,\mathbf{Y}_1^{i-1}) + I(\tilde{\mathbf{Y}}_3^{i+1};Y_{3,i}|W_0,\mathbf{Y}_1^{i-1}) + I(X_i;Y_{3,i}|W_0,\mathbf{Y}_1^{i-1}) - I(W_1,\tilde{\mathbf{Y}}_3^{i+1};Y_{3,i}|W_0,\mathbf{Y}_1^{i-1}) \right] \\
  \nonumber & \stackrel{(e)}{\leq} \sum_{i=1}^n \left[ I(W_1;Y_{2,i}| W_0,\mathbf{Y}_1^{i-1}, \tilde{\mathbf{Y}}_2^{i+1}) + I(\tilde{\mathbf{Y}}_2^{i+1};Y_{3,i}|W_0,\mathbf{Y}_1^{i-1}) - I(W_1;Y_{3,i}|W_0,\mathbf{Y}_1^{i-1}) \right] \\
  & = \sum_{i=1}^n \left[ I(U_{2,i};Y_{2,i}| \tilde U_{2,i}, U_{1,i}) + I(\tilde U_{2,i};Y_{3,i}|U_{1,i}) - I(U_{2,i};Y_{3,i}|U_{1,i}) \right] = \sum_{i=1}^n \left[ I(U_{2,i};Y_{2,i}| U_{1,i}) - I(U_{2,i};Y_{3,i}|U_{1,i}) \right]\label{r1_y2_term1}
\end{align}
where (a) has the last two terms by [\ref{ck}, Lemma 7] which gives
\begin{align}
\nonumber & \sum_{i=1}^n I(\mathbf{Y}_1^{i-1};Y_{2,i}|W_0,W_1,\tilde{\mathbf{Y}}_2^{i+1}) = \sum_{i=1}^n I(\tilde{\mathbf{Y}}_2^{i+1};Y_{1,i}|W_0,W_1,\mathbf{Y}_1^{i-1}),\\
\nonumber & \sum_{i=1}^n I(\mathbf{Y}_1^{i-1};Y_{3,i}|W_0,W_1,\tilde{\mathbf{Y}}_3^{i+1})=\sum_{i=1}^n I(\tilde{\mathbf{Y}}_3^{i+1};Y_{1,i}|W_0,W_1,\mathbf{Y}_1^{i-1});
\end{align}
(b) is because the last two terms in \eqref{r1e_u2_1} above are $I(\tilde{\mathbf{Y}}_2^{i+1};Y_{1,i}|W_0,W_1,\mathbf{Y}_1^{i-1}) = I(\tilde{\mathbf{Y}}_3^{i+1};Y_{1,i}|W_0,W_1,\mathbf{Y}_1^{i-1})=I(U_{3,i};Y_{1,i}|U_{2,i},U_{1,i})=0$ by Relation \ref{relation3} and \eqref{mkv_oi_2} and similarly $I(\tilde{\mathbf{Y}}_3^{i+1};Y_{1,i}|W_0,W_1,\mathbf{Y}_1^{i-1}) =0$; (c) is by [\ref{ck}, Lemma 7] which gives
\[
\sum_{i=1}^n I(\mathbf{Y}_1^{i-1};Y_{2,i}| W_0, \tilde{\mathbf{Y}}_2^{i+1}) = \sum_{i=1}^n I(\tilde{\mathbf{Y}}_2^{i+1};Y_{1,i}| W_0,\mathbf{Y}_1^{i-1});
\]
(d) is by [\ref{ck}, Lemma 7] from which
\[
\sum_{i=1}^n I(\mathbf{Y}_1^{i-1};Y_{3,i}| W_0, \tilde{\mathbf{Y}}_3^{i+1})=\sum_{i=1}^n I(\tilde{\mathbf{Y}}_3^{i+1};Y_{1,i}| W_0, \mathbf{Y}_1^{i-1});
\]
and (e) is by $I(\tilde{\mathbf{Y}}_2^{i+1};Y_{1,i}|W_0,\mathbf{Y}_1^{i-1})=I(\tilde{\mathbf{Y}}_3^{i+1};Y_{1,i}|W_0,\mathbf{Y}_1^{i-1})$ and $I(\tilde{\mathbf{Y}}_3^{i+1};Y_{3,i}|W_0,\mathbf{Y}_1^{i-1})=I(\tilde{\mathbf{Y}}_2^{i+1};Y_{3,i}|W_0,\mathbf{Y}_1^{i-1})$ by Relation \ref{relation1}.  Consequently we have
\begin{equation}\label{r1_final2}
	nR_{1e} \leq \sum_{i=1}^n \left[ I(U_{2,i}; Y_{2,i}| U_{1,i}) - I(U_{2,i};Y_{3,i}|U_{1,i}) \right] + n(\epsilon_1 + \gamma_3 + \gamma_5).
\end{equation}

For the equivocation rate $R_{2e}$, we consider $W_2$ sent to receiver $Y_1$ using codeword $\mathbf{X}$. Following the same procedure to obtain \eqref{r1_cnv}, we get
\begin{align}
\nonumber  nR_{2e} & \leq H(W_2|\mathbf{Y}_3) + n\epsilon_2 \\
\nonumber  & \leq I(W_2 ; \mathbf{Y}_1|W_0) - I(W_2;\mathbf{Y}_3|W_0) + H(W_0|\mathbf{Y}_3) + H(W_2|W_0, \mathbf{Y}_1) \\
& \leq I(W_2 ; \mathbf{Y}_1|W_0) - I(W_2;\mathbf{Y}_3|W_0) + n(\epsilon_2 + \gamma_3 + \gamma_6), \label{r2_initiala}
\end{align}
by Fano's inequality.  Expanding the first two terms of the inequality above by the chain rule and following the same procedure as for $R_{1e}$ in \eqref{mi_cnv_exp1}, \eqref{mi_cnv_exp2} to \eqref{r1_initial}, we obtain
\begin{equation}\label{r2_initial}
nR_{2e} \leq \sum_{i=1}^n \left[
I(W_2;Y_{1,i}|W_0,\mathbf{Y}_1^{i-1})
-I(W_2;Y_{3,i} |W_0,\tilde{\mathbf{Y}}_3^{i+1})\right] +n(\epsilon_2 + \gamma_3+\gamma_6).
\end{equation}
The terms under the summation can be bounded as
\begin{subequations}
\begin{align}
\nonumber & I(W_2;Y_{1,i}|W_0,\mathbf{Y}_1^{i-1})
-I(W_2;Y_{3,i} |W_0,\tilde{\mathbf{Y}}_3^{i+1}) \\
\nonumber & = I(\tilde{\mathbf{Y}}_3^{i+1},W_2;Y_{1,i}|W_0,\mathbf{Y}_1^{i-1}) -I(W_2,\mathbf{Y}_1^{i-1};Y_{3,i} |W_0,\tilde{\mathbf{Y}}_3^{i+1}) - I(\tilde{\mathbf{Y}}_3^{i+1};Y_{1,i}|W_0,W_2,\mathbf{Y}_1^{i-1}) \\
\nonumber & \;\; + I(\mathbf{Y}_1^{i-1};Y_{3,i}|W_0,W_2,\tilde{\mathbf{Y}}_3^{i+1}) \\
\nonumber & \stackrel{(a)}{=} I(\tilde{\mathbf{Y}}_3^{i+1},W_2;Y_{1,i}|W_0,\mathbf{Y}_1^{i-1}) -I(W_2,\mathbf{Y}_1^{i-1};Y_{3,i} |W_0,\tilde{\mathbf{Y}}_3^{i+1}) - I(\tilde{\mathbf{Y}}_3^{i+1};Y_{1,i}|W_0,W_2,\mathbf{Y}_1^{i-1}) \\
 \nonumber & \;\; + I(\tilde{\mathbf{Y}}_3^{i+1};Y_{1,i}|W_0,W_2,\mathbf{Y}_1^{i-1}) \\
 \nonumber & = I(\tilde{\mathbf{Y}}_3^{i+1},W_2;Y_{1,i}|W_0,\mathbf{Y}_1^{i-1}) -I(W_2,\mathbf{Y}_1^{i-1};Y_{3,i} |W_0,\tilde{\mathbf{Y}}_3^{i+1}) \\
 \nonumber & = I(\tilde{\mathbf{Y}}_3^{i+1};Y_{1,i}|W_0,\mathbf{Y}_1^{i-1}) + I(W_2;Y_{1,i}|W_0,\mathbf{Y}_1^{i-1},\tilde{\mathbf{Y}}_3^{i+1}) - I(\mathbf{Y}_1^{i-1};Y_{3,i} |W_0,\tilde{\mathbf{Y}}_3^{i+1}) \\
 \nonumber & \;\; - I(W_2;Y_{3,i}|W_0,\mathbf{Y}_1^{i-1},\tilde{\mathbf{Y}}_3^{i+1}) \\
 \nonumber & \stackrel{(b)}{=} I(\tilde{\mathbf{Y}}_3^{i+1};Y_{1,i}|W_0,\mathbf{Y}_1^{i-1}) + I(W_2;Y_{1,i}|W_0,\mathbf{Y}_1^{i-1},\tilde{\mathbf{Y}}_3^{i+1}) - I(\tilde{\mathbf{Y}}_3^{i+1};Y_{1,i}|W_0,\mathbf{Y}_1^{i-1}) \\
 \nonumber & \;\; - I(W_2;Y_{3,i}|W_0,\mathbf{Y}_1^{i-1},\tilde{\mathbf{Y}}_3^{i+1}) \\
  & = I(W_2;Y_{1,i}|W_0,\mathbf{Y}_1^{i-1},\tilde{\mathbf{Y}}_3^{i+1}) - I(W_2;Y_{3,i}|W_0,\mathbf{Y}_1^{i-1},\tilde{\mathbf{Y}}_3^{i+1}) \label{r2_inter_a} \\
 \nonumber & = I(X_i,W_2;Y_{1,i}|W_0,\mathbf{Y}_1^{i-1},\tilde{\mathbf{Y}}_3^{i+1}) - I(X_i,W_2;Y_{3,i}|W_0,\mathbf{Y}_1^{i-1},\tilde{\mathbf{Y}}_3^{i+1}) \\
 \nonumber & \;\; - [I(X_i;Y_{1,i}|W_0,W_2,\mathbf{Y}_1^{i-1},\tilde{\mathbf{Y}}_3^{i+1}) - I(X_i;Y_{3,i}|W_0,W_2,\mathbf{Y}_1^{i-1},\tilde{\mathbf{Y}}_3^{i+1})] \\
 \nonumber & \stackrel{(c)}{\leq} I(X_i,W_2;Y_{1,i}|W_0,\mathbf{Y}_1^{i-1},\tilde{\mathbf{Y}}_3^{i+1}) - I(X_i,W_2;Y_{3,i}|W_0,\mathbf{Y}_1^{i-1},\tilde{\mathbf{Y}}_3^{i+1}) \\
 \nonumber & = I(X_i,W_2,\tilde{\mathbf{Y}}_3^{i+1};Y_{1,i}|W_0,\mathbf{Y}_1^{i-1}) - I(X_i,W_2,\tilde{\mathbf{Y}}_3^{i+1};Y_{3,i}|W_0,\mathbf{Y}_1^{i-1}) \\
 \nonumber & \;\; - [I(\tilde{\mathbf{Y}}_3^{i+1};Y_{1,i}|W_0,\mathbf{Y}_1^{i-1}) - I(\tilde{\mathbf{Y}}_3^{i+1};Y_{3,i}|W_0,\mathbf{Y}_1^{i-1})] \\
 \nonumber & \stackrel{(d)}{\leq} I(X_i;Y_{1,i}|W_0,\mathbf{Y}_1^{i-1}) - I(X_i;Y_{3,i}|W_0,\mathbf{Y}_1^{i-1}) - [I(W_1;Y_{1,i}|W_0,\mathbf{Y}_1^{i-1}) - I(W_1;Y_{3,i}|W_0,\mathbf{Y}_1^{i-1})] \\
 & = I(X_i;Y_{1,i}|U_{2,i})-I(X_i;Y_{3,i}|U_{2,i}) \label{r2_inter}
\end{align}
\end{subequations}
where (a) is by [\ref{ck}, Lemma 7] so that
\[
\sum_{i=1}^n I(\tilde{\mathbf{Y}}_3^{i+1};Y_{1,i}|W_0,W_2,\mathbf{Y}_1^{i-1}) = \sum_{i=1}^n I(\mathbf{Y}_1^{i-1};Y_{3,i}|W_0,W_2,\tilde{\mathbf{Y}}_3^{i+1}),
\]
(b) is also by [\ref{ck}, Lemma 7] giving
\[
\sum_{i=1}^n I(\tilde{\mathbf{Y}}_3^{i+1};Y_{1,i}|W_0,\mathbf{Y}_1^{i-1}) = \sum_{i=1}^n I(\mathbf{Y}_1^{i-1};Y_{3,i}|W_0,\tilde{\mathbf{Y}}_3^{i+1}),
\]
(c) is because $[I(X_i;Y_{1,i}|W_0,W_2,\mathbf{Y}_1^{i-1},\tilde{\mathbf{Y}}_3^{i+1}) - I(X_i;Y_{3,i}|W_0,W_2,\mathbf{Y}_1^{i-1},\tilde{\mathbf{Y}}_3^{i+1})] \ge 0$ since $Y_1$ is a more capable channel than $Y_3$ and $(W_0,W_2,\mathbf{Y}_1^{i-1},\tilde{\mathbf{Y}}_3^{i+1}) \to X_i \to (Y_{1,i},Y_{3,i})$ forms a Markov chain so satisfying the more capable channel condition, and (d) is due to firstly,
\begin{align}
	\nonumber & I(W_2,\tilde{\mathbf{Y}}_3^{i+1};Y_{1,i}|W_0,\mathbf{Y}_1^{i-1},X_i) = H(Y_{1,i}|W_0,\mathbf{Y}_1^{i-1},X_i) - H(Y_{1,i}|W_0,W_2,\tilde{\mathbf{Y}}_3^{i+1},\mathbf{Y}_1^{i-1},X_i)\\
	& = H(Y_{1,i}|W_0,\mathbf{Y}_1^{i-1},X_i) - H(Y_{1,i}|W_0,\mathbf{Y}_1^{i-1},X_i) = 0,
\end{align}
which is true since, given $X_i$, $\tilde{\mathbf{Y}}_3^{i+1}$ is independent of $Y_{3,i}$ as can be verified using a functional dependency graph, and by $W_2 \to X_i \to Y_{1,i}$; secondly, $I(W_2,\tilde{\mathbf{Y}}_3^{i+1};Y_{3,i}|W_0,\mathbf{Y}_1^{i-1},X_i)=0$ since given $X_i$, $\tilde{\mathbf{Y}}_3^{i+1}$ is independent of $Y_{3,i}$ and $W_2 \to X_i \to Y_{3,i}$; and thirdly, $I(\tilde{\mathbf{Y}}_3^{i+1};Y_{k,i}|W_0,\mathbf{Y}_1^{i-1}) = I(W_1;Y_{k,i}|W_0,\mathbf{Y}_1^{i-1})$ for $k=1,3$ from Relation \ref{relation2}.  Then, we shall have
\begin{equation}\label{r2_final}
nR_{2e} \leq \sum_{i=1}^n \left[ I(X_i;Y_{1,i}|U_{2,i}) - I(X_i ; Y_{3,i}| U_{2,i})\right] +n(\epsilon_2 + \gamma_3+\gamma_6).
\end{equation}

For the rates $(R_{1e} + R_{2e})$, consider the combined message $(W_1,W_2)$ sent to receiver $Y_1$ using codeword $\mathbf{X}$. It can be shown that
\begin{align}
\nonumber & n(R_{1e} + R_{2e}) \leq H(W_1,W_2|\mathbf{Y}_3) + \epsilon_{1,2}   \\
\nonumber  & \stackrel{(a)}{\leq}  I(W_1,W_2;\mathbf{Y}_1|W_0) - I(W_1,W_2;\mathbf{Y}_3|W_0) + H(W_0|\mathbf{Y}_3) + H(W_1,W_2|W_0,\mathbf{Y}_1) + \epsilon_{1,2} \\
\nonumber  & \stackrel{(b)}{\leq} I(W_1,W_2;\mathbf{Y}_1|W_0) - I(W_1,W_2;\mathbf{Y}_3|W_0) + n(\epsilon_{1,2}+ \gamma_3 + \gamma_7) \\
& \stackrel{(c)}{\leq} \sum_{i=1}^n \left[ I(W_1,W_2;Y_{1,i}|W_0,\mathbf{Y}_1^{i-1},\tilde{\mathbf{Y}}_3^{i+1}) - I(W_1,W_2;Y_{3,i}|W_0,\mathbf{Y}_1^{i-1},\tilde{\mathbf{Y}}_3^{i+1})\right] + n(\epsilon_{1,2}+ \gamma_3 + \gamma_7) \label{r012_initial}
\end{align}
where (a) results in following the steps in \eqref{r1_cnv} using $(W_1,W_2)$ instead of $W_1$, (b) is by Fano's inequality, and (c) results in following the steps to obtain \eqref{r2_inter_a} using $(W_1,W_2)$ instead of $W_2$.  The terms under the summation can be bounded as
\begin{align}
	\nonumber & I(W_1,W_2;Y_{1,i}|W_0,\mathbf{Y}_1^{i-1},\tilde{\mathbf{Y}}_3^{i+1}) - I(W_1,W_2;Y_{3,i}|W_0,\mathbf{Y}_1^{i-1},\tilde{\mathbf{Y}}_3^{i+1}) \\
	\nonumber & = I(X_i,W_1,W_2;Y_{1,i}|W_0,\mathbf{Y}_1^{i-1},\tilde{\mathbf{Y}}_3^{i+1}) - I(X_i,W_1,W_2;Y_{3,i}|W_0,\mathbf{Y}_1^{i-1},\tilde{\mathbf{Y}}_3^{i+1}) \\
	\nonumber & \;\; - [I(X_i;Y_{1,i}|W_0,W_1,W_2,\mathbf{Y}_1^{i-1},\tilde{\mathbf{Y}}_3^{i+1}) - I(X_i;Y_{3,i}|W_0,W_1,W_2,\mathbf{Y}_1^{i-1},\tilde{\mathbf{Y}}_3^{i+1})] \\
	\nonumber & \stackrel{(a)}{\leq} I(X_i,W_1,W_2;Y_{1,i}|W_0,\mathbf{Y}_1^{i-1},\tilde{\mathbf{Y}}_3^{i+1}) - I(X_i,W_1,W_2;Y_{3,i}|W_0,\mathbf{Y}_1^{i-1},\tilde{\mathbf{Y}}_3^{i+1}) \\
	\nonumber & = I(X_i,W_1,W_2,\tilde{\mathbf{Y}}_3^{i+1};Y_{1,i}|W_0,\mathbf{Y}_1^{i-1}) - I(X_i,W_1,W_2,\tilde{\mathbf{Y}}_3^{i+1};Y_{3,i}|W_0,\mathbf{Y}_1^{i-1}) \\
	\nonumber & \;\; - [I(\tilde{\mathbf{Y}}_3^{i+1};Y_{1,i}|W_0,\mathbf{Y}_1^{i-1})-I(\tilde{\mathbf{Y}}_3^{i+1};Y_{3,i}|W_0,\mathbf{Y}_1^{i-1})] \\
	\nonumber & \stackrel{(b)}{=} I(X_i;Y_{1,i}|W_0,\mathbf{Y}_1^{i-1}) - I(X_i;Y_{3,i}|W_0,\mathbf{Y}_1^{i-1}) - [I(\tilde{\mathbf{Y}}_2^{i+1};Y_{1,i}|W_0,\mathbf{Y}_1^{i-1})-I(\tilde{\mathbf{Y}}_2^{i+1};Y_{3,i}|W_0,\mathbf{Y}_1^{i-1})] \\
	\nonumber & = I(X_i;Y_{1,i}|U_{1,i}) - I(X_i;Y_{3,i}|U_{1,i}) - [I(\tilde U_{2,i};Y_{1,i}|U_{1,i}) - I(\tilde U_{2,i};Y_{3,i}|U_{1,i})] \\
	& \stackrel{(c)}{=} I(X_i;Y_{1,i}|U_{1,i}) - I(X_i;Y_{3,i}|U_{1,i}), \label{r12e_inter}
\end{align}
where (a) is due to $Y_1$ being a more capable channel than $Y_1$ which gives $[I(X_i;Y_{1,i}|W_0,W_1,W_2,\mathbf{Y}_1^{i-1},\tilde{\mathbf{Y}}_3^{i+1}) - I(X_i;Y_{3,i}|W_0,W_1,W_2,\mathbf{Y}_1^{i-1},\tilde{\mathbf{Y}}_3^{i+1})] \geq 0$ for $(W_0,W_1,W_2,\mathbf{Y}_1^{i-1},\tilde{\mathbf{Y}}_3^{i+1}) \to X_i \to (Y_{1,i},Y_{3,i})$ as may be verified using a functional dependency graph; (b) is due to, first, that
\begin{align}
	\nonumber & I(W_1,W_2,\tilde{\mathbf{Y}}_3^{i+1};Y_{1,i}|W_0,\mathbf{Y}_1^{i-1},X_i) = H(Y_{1,i}|W_0,\mathbf{Y}_1^{i-1},X_i) - H(Y_{1,i}|W_0,W_1,W_2,\mathbf{Y}_1^{i-1},\tilde{\mathbf{Y}}_3^{i+1},X_i) \\
	& = H(Y_{1,i}|W_0,\mathbf{Y}_1^{i-1},X_i) - H(Y_{1,i}|W_0,\mathbf{Y}_1^{i-1},X_i) = 0
\end{align}
since in the second term in the second equality is obtained using the relation $(W_1,W_2) \to X_i \to Y_{1,i}$ and the fact that given $X_i$, $\tilde{\mathbf{Y}}_3^{i+1}$ is independent of $Y_{1,i}$, secondly, we can obtain $I(W_1,W_2;Y_{1,i}|W_0,\mathbf{Y}_1^{i-1},\tilde{\mathbf{Y}}_3^{i+1},X_i)=0$ in a similar way, and thirdly we have, by Relation \ref{relation1},  $I(\tilde{\mathbf{Y}}_3^{i+1};Y_{k,i}|W_0,\mathbf{Y}_1^{i-1})=I(\tilde{\mathbf{Y}}_2^{i+1};Y_{k,i}|W_0,\mathbf{Y}_1^{i-1})$, $k=1,3$; (c) is by substituting $\tilde U_{2,i} = U_{1,i}$.  Then, we have
\begin{align}
n(R_{1e} + R_{2e}) \leq \sum_{i=1}^n \left[I(X_i;Y_{1,i}|U_{1,i}) - I(X_i;Y_{3,i}|U_{1,i}), \right] + n(\epsilon_{1,2}+ \gamma_3 + \gamma_7). \label{r12e_final}
\end{align}

We now prove the rates for $R_0$, $R_0 + R_1$, $R_0 + R_2$ and $R_0 + R_1 + R_2$. For rate $R_0$, we have
\begin{align}
\nonumber  nR_0 & = H(W_0) = I(W_0;\mathbf{Y}_1) + H(W_0|\mathbf{Y}_1) \\
\nonumber  & \leq I(W_0;\mathbf{Y}_1) + n\gamma_1 \;\;\;\;\; \mbox{by Fano's inequality} \\
\nonumber  & = \sum_{i=1}^n I(W_0;Y_{1,i}|\mathbf{Y}_1^{i-1}) + n\gamma_1 \\
  & \leq \sum_{i=1}^n I(W_0,\mathbf{Y}_1^{i-1};Y_{1,i}) + n\gamma_1 \label{r0_final1_p} \\
& = \sum_{i=1}^n I(U_{1,i} ; Y_{1,i}) + n\gamma_1. \label{r0_final1}
\end{align}
We also have
\begin{align}
\nonumber nR_0 & = H(W_0) = I(W_0;\mathbf{Y}_3) + H(W_0|\mathbf{Y}_3) \\
\nonumber  & \leq I(W_0;\mathbf{Y}_3) + n\gamma_3 \;\;\;\;\; \mathrm{by \; Fano's \; inequality} \\
\nonumber  & = \sum_{i=1}^n I(W_0;Y_{3,i}|\tilde{\mathbf{Y}}_3^{i+1}) + n\gamma_3 \\
\nonumber & = \sum_{i=1}^n \left[ I(W_0,\mathbf{Y}_1^{i-1};Y_{3,i}|\tilde{\mathbf{Y}}_3^{i+1}) - I(\mathbf{Y}_1^{i-1};Y_{3,i}|W_0,\tilde{\mathbf{Y}}_3^{i+1}) \right] + n\gamma_3 \\
 & \stackrel{(a)}{\leq} \sum_{i=1}^n \left[ I(W_0,\mathbf{Y}_1^{i-1},\tilde{\mathbf{Y}}_3^{i+1};Y_{3,i}) - I(\tilde{\mathbf{Y}}_3^{i+1};Y_{1,i}|W_0,\mathbf{Y}_1^{i-1}) \right] + n\gamma_3 \label{r0_final2_p} \\
& = \sum_{i=1}^n \left[ I(U_{3,i};Y_{3,i}) - I(U_{3,i};Y_{1,i}|U_{1,i}) \right] + n\gamma_3 \label{r0_final2}
\end{align}
where (a) is by [\ref{ck},Lemma 7] from which $\sum_{i=1}^n I(\mathbf{Y}_1^{i-1};Y_{3,i}|W_0,\tilde{\mathbf{Y}}_3^{i+1}) =\sum_{i=1}^n I(\tilde{\mathbf{Y}}_3^{i+1};Y_{1,i}|W_0,\mathbf{Y}_1^{i-1})$.

For the rates $(R_0 + R_1)$, we consider the following cases when the messages are sent:
\begin{enumerate}
	\item Case 1: $W_1$ sent to $Y_1$, $W_0$ sent to $Y_1$ or $Y_3$;
	\item Case 2: $W_1$ sent to $Y_2$, $W_0$ sent to $Y_1$ or $Y_3$;
	\item Case 3: $W_0,W_1$ both sent to $Y_2$.
\end{enumerate}

For Case 1, we have
\begin{align}
\nonumber  n(R_0 + R_1) & = H(W_0,W_1) = H(W_0) + H(W_1|W_0) \\
\nonumber  & = H(W_0) + I(W_1;\mathbf{Y}_1|W_0) + H(W_1|W_0,\mathbf{Y}_1) \\
\nonumber  & \stackrel{(a)}{\leq} H(W_0) + \sum_{i=1}^n I(W_1;Y_{1,i}|W_0,\mathbf{Y}_1^{i-1}) + \gamma_4
\end{align}
where (a) is by expanding using the chain rule and using Fano's inequality.  Then, on combining with $H(W_0)$ using \eqref{r0_final1_p}, we can get
\begin{align}\label{r01_final1}
	\nonumber n(R_0 + R_1) & \leq \sum_{i=1}^n \left[ I(W_0,\mathbf{Y}_1^{i-1};Y_{1,i}) + I(W_1;Y_{1,i}|W_0,\mathbf{Y}_1^{i-1}) \right] + n(\gamma_1+\gamma_4) \\
	 & = \sum_{i=1}^n I(W_1;Y_{1,i}) + n(\gamma_1 + \gamma_4) = \sum_{i=1}^n I(U_{2,i};Y_{1,i}) + n(\gamma_1 + \gamma_4)
\end{align}
and, combining with $H(W_0)$ using \eqref{r0_final2_p}, we obtain
\begin{align}\label{r01_final2}
	\nonumber n(R_0 + R_1) & \leq \sum_{i=1}^n \left[ I(W_0,\mathbf{Y}_1^{i-1},\tilde{\mathbf{Y}}_3^{i+1};Y_{3,i}) - I(\tilde{\mathbf{Y}}_3^{i+1};Y_{1,i}|W_0,\mathbf{Y}_1^{i-1}) + I(W_1;Y_{1,i}|W_0,\mathbf{Y}_1^{i-1})  \right] + n(\gamma_3 + \gamma_4) \\
	& \leq \sum_{i=1}^n \left[ I(U_{3,i};Y_{3,i}) + I(U_{2,i};Y_{1,i}|U_{1,i}) \right] + n(\gamma_3 + \gamma_4).
\end{align}

For Case 2, we similarly have
\begin{align}\label{r01_y2_inter}
	\nonumber & n(R_0 + R_1) = H(W_0) + I(W_1;\mathbf{Y}_2|W_0) + H(W_1|W_0,\mathbf{Y}_2) \\
	\nonumber & \leq H(W_0) + \sum_{i=1}^n I(W_1;Y_{2,i}|W_0,\tilde{\mathbf{Y}}_2^{i+1}) + n\gamma_5 \\
	\nonumber & = H(W_0) + \sum_{i=1}^n \left[ I(W_1,\mathbf{Y}_1^{i-1};Y_{2,i}|W_0,\tilde{\mathbf{Y}}_2^{i+1}) - I(\mathbf{Y}_1^{i-1};Y_{2,i}|W_0,W_1,\tilde{\mathbf{Y}}_2^{i+1}) \right] + n\gamma_5 \\
	 \nonumber & \stackrel{(a)}{=} H(W_0) + \sum_{i=1}^n \left[ I(\mathbf{Y}_1^{i-1};Y_{2,i}|W_0,\tilde{\mathbf{Y}}_2^{i+1}) + I(W_1;Y_{2,i}|W_0,\mathbf{Y}_1^{i-1},\tilde{\mathbf{Y}}_2^{i+1}) - I(\tilde{\mathbf{Y}}_2^{i+1};Y_{1,i}|W_0,W_1,\mathbf{Y}_1^{i-1}) \right] + n\gamma_5 \\
	 \nonumber & \stackrel{(b)}{=} H(W_0) + \sum_{i=1}^n \left[ I(\tilde{\mathbf{Y}}_2^{i+1};Y_{1,i}|W_0,\mathbf{Y}_1^{i-1}) + I(W_1;Y_{2,i}|W_0,\mathbf{Y}_1^{i-1},\tilde{\mathbf{Y}}_2^{i+1}) \right] + n\gamma_5 \\
	  & = H(W_0) + \sum_{i=1}^n \left[ I(\tilde U_{2,i};Y_{1,i}|U_{1,i}) + I(U_{2,i};Y_{2,i}|\tilde U_{2,i},U_{1,i}) \right] + n\gamma_5 = H(W_0) + \sum_{i=1}^n I(U_{2,i};Y_{2,i}|U_{1,i}) + n\gamma_5
	 \end{align}
where (a) has the last term in the sum by [\ref{ck}, Lemma 7] giving
\[
\sum_{i=1}^n I(\mathbf{Y}_1^{i-1};Y_{2,i}|W_0,W_1,\tilde{\mathbf{Y}}_2^{i+1})=\sum_{i=1}^n I(\tilde{\mathbf{Y}}_2^{i+1};Y_{1,i}|W_0,W_1,\mathbf{Y}_1^{i-1});
\]
(b) is by $I(\tilde{\mathbf{Y}}_2^{i+1};Y_{1,i}|W_0,W_1,\mathbf{Y}_1^{i-1}) = I(\tilde{\mathbf{Y}}_3^{i+1};Y_{1,i}|W_0,W_1,\mathbf{Y}_1^{i-1})=I(U_{3,i};Y_{1,i}|U_{2,i},U_{1,i})=0$ from Relation \ref{relation3} and \eqref{mkv_oi_2} and first term under the sum by [\ref{ck}, Lemma 7] from which
\[
\sum_{i=1}^n I(\mathbf{Y}_1^{i-1};Y_{2,i}|W_0,\tilde{\mathbf{Y}}_2^{i+1})=\sum_{i=1}^n I(\tilde{\mathbf{Y}}_2^{i+1};Y_{1,i}|W_0,\mathbf{Y}_1^{i-1}).
\]
Combining with $H(W_0)$ using \eqref{r0_final1_p} and \eqref{r0_final2_p}, we obtain
\begin{align}
n(R_0 + R_1) & \leq \sum_{i=1}^n \left[ I(U_{1,i};Y_{1,i})+I(U_{2,i};Y_{2,i}|U_{1,i})\right] + n(\gamma_1 + \gamma_5), \label{r01_final3}\\
n(R_0 + R_1) & \leq \sum_{i=1}^n \left[ I(U_{3,i};Y_{3,i})+I(U_{2,i};Y_{2,i}|U_{1,i}) \right] + n(\gamma_3 + \gamma_5). \label{r01_final4}
\end{align}

For Case 3 we have
\begin{align}
	\nonumber & n(R_0 + R_1) = H(W_0,W_1) =  I(W_0,W_1;\mathbf{Y}_2) + H(W_0,W_1|\mathbf{Y}_2) \\
	\nonumber & \stackrel{(a)}{\leq} \sum_{i=1}^n I(W_0,W_1;Y_{2,i}|\tilde{\mathbf{Y}}_2^{i+1}) + n\gamma_5 \\
	\nonumber & = \sum_{i=1}^n \left[ I(W_0,W_1,\mathbf{Y}_1^{i-1};Y_{2,i}|\tilde{\mathbf{Y}}_2^{i+1}) - I(\mathbf{Y}_1^{i-1};Y_{2,i}|W_0,W_1,\tilde{\mathbf{Y}}_2^{i+1})\right] + n\gamma_5 \\
	\nonumber & \stackrel{(b)}{=} \sum_{i=1}^n \left[ I(W_0,W_1,\mathbf{Y}_1^{i-1};Y_{2,i}|\tilde{\mathbf{Y}}_2^{i+1}) - I(\tilde{\mathbf{Y}}_2^{i+1};Y_{1,i}|W_0,W_1,\mathbf{Y}_1^{i-1})\right] + n\gamma_5 \\
	& \leq \sum_{i=1}^n I(\tilde U_{2,i},U_{2,i};Y_{2,i}) + n\gamma_5 = \sum_{i=1}^n I(U_{2,i};Y_{2,i}) + n\gamma_5. \label{r01_final5}
\end{align}
where (a) is by Fano's inequality, and (b) has second term in the sum by [\ref{ck}, Lemma 7].

For rates $(R_0 + R_2)$ consider message $W_2$ sent to receiver $Y_1$ and $W_0$ sent to either $Y_1$ or $Y_3$.  To begin, we have
\begin{align}
	\nonumber & n(R_0 + R_2) = H(W_0) + H(W_2|W_0) = H(W_0) + I(W_2;\mathbf{Y}_1|W_0) + H(W_2;\mathbf{Y}_1|W_0) \\
	\nonumber & \stackrel{(a)}{\leq} H(W_0) + \sum_{i=1}^n I(W_2;Y_{1,i}|W_0,\mathbf{Y}_1^{i-1}) + n\gamma_6 \leq H(W_0) + \sum_{i=1}^n I(W_2;W_1,Y_{1,i}|W_0,\mathbf{Y}_1^{i-1}) + n\gamma_6 \\
	 \nonumber & \stackrel{(b)}{=} H(W_0) + \sum_{i=1}^n I(W_2;Y_{1,i}|W_0,W_1,\mathbf{Y}_1^{i-1}) + n\gamma_6 \stackrel{(c)}{\leq} H(W_0) + \sum_{i=1}^n I(X_i;Y_{1,i}|W_0,W_1,\mathbf{Y}_1^{i-1}) + n\gamma_6 \\
	  & = H(W_0) + \sum_{i=1}^n \left[ I(X_i,W_1;Y_{1,i}|W_0,\mathbf{Y}_1^{i-1}) - I(W_1;Y_{1,i}|W_0,\mathbf{Y}_1^{i-1}) \right] + n\gamma_6 \label{r02_inter1} \\
	 \nonumber & \stackrel{(d)}{=} H(W_0) + \sum_{i=1}^n \left[ I(X_i,W_1;Y_{1,i}|W_0,\mathbf{Y}_1^{i-1}) - I(W_1,\tilde{\mathbf{Y}}_3^{i+1};Y_{1,i}|W_0,\mathbf{Y}_1^{i-1}) \right] + n\gamma_6 \\
	 \nonumber & \stackrel{(e)}{=} H(W_0) + \sum_{i=1}^n \left[ I(X_i;Y_{1,i}|W_0,\mathbf{Y}_1^{i-1}) - I(W_1,\tilde{\mathbf{Y}}_3^{i+1};Y_{1,i}|W_0,\mathbf{Y}_1^{i-1}) \right] + n\gamma_6 \\
	 & = H(W_0) + \sum_{i=1}^n I(X_i;Y_{1,i}|U_{2,i},U_{3,i}) + n\gamma_6 \label{r02_inter2}
\end{align}
where (a) is by Fano's inequality; (b) is by the independence of $W_1$ and $W_2$; (c) is by $W_1 \to X_i \to Y_{1,i}$; (d) is by Relation \ref{relation2}; and (e) is because
\begin{align}
\nonumber I(W_1;Y_{1,i}|W_0,\mathbf{Y}_1^{i-1},X_i) & = H(Y_{1,i}|W_0,\mathbf{Y}_1^{i-1},X_i) - H(Y_{1,i}|W_0,W_1,\mathbf{Y}_1^{i-1},X_i) \\
& =H(Y_{1,i}|W_0,\mathbf{Y}_1^{i-1},X_i) - H(Y_{1,i}|W_0,\mathbf{Y}_1^{i-1},X_i)=0
\end{align}
with the second term in the second equality being due to $W_1 \to X_i \to Y_{1,i}$.  Combine the results with $H(W_0)$ in two ways.  Firstly, we do this by combining with \eqref{r0_final1} using \eqref{r02_inter2} to get
\begin{equation}\label{r02_final1}
	n(R_0 + R_2) \leq \sum_{i=1}^n \left[ I(U_{1,i};Y_{1,i}) + I(X_i;Y_{1,i}|U_{2,i},U_{3,i}) \right] + n(\gamma_1 + \gamma_6).
\end{equation}
Next combine with \eqref{r0_final2_p} using \eqref{r02_inter1} to get
\begin{align}
	\nonumber n(R_0 + R_2) & \leq \sum_{i=1}^n \left[ I(W_0,\mathbf{Y}_1^{i-1},\tilde{\mathbf{Y}}_3^{i+1};Y_{3,i}) - I(\tilde{\mathbf{Y}}_3^{i+1};Y_{1,i}|W_0,\mathbf{Y}_1^{i-1}) + I(X_i,W_1;Y_{1,i}|W_0,\mathbf{Y}_1^{i-1}) \right. \\
	\nonumber & \;\;\;\; \left. - I(W_1;Y_{1,i}|W_0,\mathbf{Y}_1^{i-1}) \right] + n(\gamma_3 + \gamma_6) \\
	\nonumber & = \sum_{i=1}^n \left[ I(U_{3,i};Y_{3,i}) - I(U_{3,i};Y_{1,i}|U_{1,i}) + I(X_i;Y_{1,i}|U_{1,i}) - I(U_{2,i};Y_{1,i}|U_{1,i}) \right] + n(\gamma_3 + \gamma_6) \\
	\nonumber & \stackrel{(a)}{\leq} \sum_{i=1}^n \left[ I(U_{3,i};Y_{3,i}) + I(X_i;Y_{1,i}|U_{1,i}) - I(U_{2,i},U_{3,i};Y_{1,i}|U_{1,i})\right] + n(\gamma_3 + \gamma_6) \\
	& = \sum_{i=1}^n \left[ I(U_{3,i};Y_{3,i}) + I(X_i;Y_{1,i}|U_{2,i},U_{3,i}) \right] + n(\gamma_3 + \gamma_6), \label{r02_final2}
\end{align}
where (a) is by
\begin{align}
	\nonumber & I(U_{2,i};Y_{1,i}|U_{1,i})+I(U_{3,i};Y_{1,i}|U_{1,i}) = I(U_{2,i},U_{3,i};Y_{1,i}|U_{1,i}) - I(U_{3,i};Y_{1,i}|U_{1,i},U_{2,i}) + I(U_{3,i};Y_{1,i}|U_{1,i}) \\
	 & \geq I(U_{2,i},U_{3,i};Y_{1,i}|U_{1,i}) - I(U_{3,i};Y_{1,i}|U_{1,i},U_{2,i}) + I(U_{3,i};Y_{1,i}|U_{1,i},U_{2,i}) = I(U_{2,i},U_{3,i};Y_{1,i}|U_{1,i})
\end{align}
with the inequality obtained using \eqref{mkv_oi_1}.

Lastly, for the rates $(R_0 + R_1+R_2)$, consider the following combinations of messages sent to the receivers:
\begin{enumerate}
	\item Case 1: $W_1,W_2$ sent to $Y_1$, $W_0$ sent to $Y_1$ or $Y_3$,
	\item Case 2: $W_1$ sent to $Y_2$, $W_2$ sent to $Y_1$, $W_0$ sent to $Y_1$ or $Y_3$,
	\item Case 3: $W_0,W_1$ sent to $Y_2$, $W_2$ sent to $Y_1$.
\end{enumerate}
For Case 1, we have
\begin{align}
	\nonumber & n(R_0 + R_1 + R_2) = H(W_0) + H(W_1,W_2|W_0) = H(W_0) + I(W_1,W_2;\mathbf{Y}_1|W_0) + H(W_1,W_2|W_0,\mathbf{Y}_1) \\
	\nonumber & \leq H(W_0) + \sum_{i=1}^n \left[ I(W_1;Y_{1,i}|W_0,\mathbf{Y}_1^{i-1}) + I(W_2;Y_{1,i}|W_0,W_1,\mathbf{Y}_1^{i-1}) \right]+ n\gamma_7 \\
	\nonumber & \stackrel{(a)}{\leq} H(W_0) + \sum_{i=1}^n \left[ I(W_1;Y_{1,i}|W_0,\mathbf{Y}_1^{i-1}) + I(X_i;Y_{1,i}|W_0,W_1,\mathbf{Y}_1^{i-1}) \right]+ n\gamma_7 \\
	& = H(W_0) + \sum_{i=1}^n I(X_i;Y_{1,i}|W_0,\mathbf{Y}_1^{i-1}) + n\gamma_7, \label{r012_inter1}
\end{align}
where (a) is by $W_1 \to X_i \to Y_{1,i}$.  Then combining with \eqref{r0_final1_p}, we have
\begin{equation}\label{r012_final1}
	n(R_0 + R_1 + R_2) \leq \sum_{i=1}^n \left[ I(W_0,\mathbf{Y}_1^{i-1};Y_{1,i}) + I(X_i;Y_{1,i}|W_0,\mathbf{Y}_1^{i-1}) \right] + n(\gamma_1+\gamma_7) = \sum_{i=1}^n I(X_i;Y_{1,i}) + n(\gamma_1+\gamma_7).
\end{equation}
Combining \eqref{r012_inter1} with \eqref{r0_final2_p}, we have
\begin{align}
	\nonumber n(R_0 + R_1 + R_2) & \leq \sum_{i=1}^n \left[ I(W_0,\mathbf{Y}_1^{i-1},\tilde{\mathbf{Y}}_3^{i+1};Y_{3,i}) - I(\tilde{\mathbf{Y}}_3^{i+1};Y_{1,i}|W_0,\mathbf{Y}_1^{i-1}) + I(X_i;Y_{1,i}|W_0,\mathbf{Y}_1^{i-1}) \right] + n(\gamma_3+\gamma_7) \\
	& = \sum_{i=1}^n \left[ I(U_{3,i};Y_{3,i}) + I(X_i; Y_{1,i}|U_{3,i}) \right] + n(\gamma_3+\gamma_7). \label{r012_final2}
\end{align}

For Case 2 we have
\begin{align}
	\nonumber & n(R_0 + R_1 + R_2) = H(W_0) + H(W_1|W_0) + H(W_2|W_0,W_1) \\
	\nonumber & = H(W_0) + I(W_1;\mathbf{Y}_2|W_0) + H(W_2|W_0,\mathbf{Y}_2) + I(W_2;\mathbf{Y}_1|W_0,W_1) + H(W_2|W_0,W_1,\mathbf{Y}_1) \\
	\nonumber & \leq  H(W_0) + \sum_{i=1}^n \left[ I(W_1;Y_{2,i}|W_0,\tilde{\mathbf{Y}}_2^{i+1}) + I(W_2;Y_{1,i}|W_0,W_1,\mathbf{Y}_1^{i-1}) \right] + n(\gamma_5 + \gamma_8) \\
	\nonumber & \stackrel{(a)}{\leq} H(W_0) + \sum_{i=1}^n \left[ I(W_1,\mathbf{Y}_1^{i-1};Y_{2,i}|W_0,\tilde{\mathbf{Y}}_2^{i+1}) - I(\mathbf{Y}_1^{i-1};Y_{2,i}|W_0,W_1,\tilde{\mathbf{Y}}_2^{i+1}) \right. \\
	\nonumber & \;\;\;\; \left. + I(X_i;Y_{1,i}|W_0,W_1,\mathbf{Y}_1^{i-1})  \right] + n(\gamma_5 + \gamma_8) \\
	\nonumber & = H(W_0) + \sum_{i=1}^n \left[ I(\mathbf{Y}_1^{i-1};Y_{2,i}|W_0,\tilde{\mathbf{Y}}_2^{i+1}) + I(W_1;Y_{2,i}|W_0,\mathbf{Y}_1^{i-1},\tilde{\mathbf{Y}}_2^{i+1}) - I(\mathbf{Y}_1^{i-1};Y_{2,i}|W_0,W_1,\tilde{\mathbf{Y}}_2^{i+1}) \right. \\
	\nonumber & \;\;\;\; \left. + I(X_i,W_1;Y_{1,i}|W_0,\mathbf{Y}_1^{i-1}) - I(W_1;Y_{1,i}|W_0,\mathbf{Y}_1^{i-1}) \right] + n(\gamma_5 + \gamma_8) \\
	\nonumber & \stackrel{(b)}{=} H(W_0) + \sum_{i=1}^n \left[ I(\tilde{\mathbf{Y}}_2^{i+1};Y_{1,i}|W_0,\mathbf{Y}_1^{i-1}) + I(W_1;Y_{2,i}|W_0,\mathbf{Y}_1^{i-1},\tilde{\mathbf{Y}}_2^{i+1}) - I(\tilde{\mathbf{Y}}_2^{i+1};Y_{1,i}|W_0,W_1,\mathbf{Y}_1^{i-1}) \right. \\
	\nonumber & \;\;\;\; \left. + I(X_i,W_1;Y_{1,i}|W_0,\mathbf{Y}_1^{i-1}) - I(W_1,\tilde{\mathbf{Y}}_3^{i+1};Y_{1,i}|W_0,\mathbf{Y}_1^{i-1}) + I(\tilde{\mathbf{Y}}_3^{i+1};Y_{1,i}|W_0,W_1,\mathbf{Y}_1^{i-1}) \right] + n(\gamma_5 + \gamma_8) \\
	\nonumber & \stackrel{(c)}{=} H(W_0) + \sum_{i=1}^n \left[ I(\tilde{\mathbf{Y}}_2^{i+1};Y_{1,i}|W_0,\mathbf{Y}_1^{i-1}) + I(W_1;Y_{2,i}|W_0,\mathbf{Y}_1^{i-1},\tilde{\mathbf{Y}}_2^{i+1}) \right. \\
	  & \;\;\;\; \left. + I(X_i;Y_{1,i}|W_0,\mathbf{Y}_1^{i-1}) - I(W_1,\tilde{\mathbf{Y}}_3^{i+1};Y_{1,i}|W_0,\mathbf{Y}_1^{i-1}) \right] + n(\gamma_5 + \gamma_8) \label{r012_inter2a} \\
	 \nonumber & = H(W_0) + \sum_{i=1}^n \left[ I(\tilde U_{2,i};Y_{1,i}|U_{1,i}) + I(U_{2,i};Y_{2,i}|\tilde U_{2,i},U_{1,i}) + I(X_i;Y_{1,i}|U_{2,i},U_{3,i}) \right] + n(\gamma_5 + \gamma_8)  \\
	 & = H(W_0) + \sum_{i=1}^n \left[  I(U_{2,i};Y_{2,i}|U_{1,i}) + I(X_i;Y_{1,i}|U_{2,i},U_{3,i}) \right] + n(\gamma_5 + \gamma_8) \label{r012_inter2b}
\end{align}
where (a) is by $W_2 \to X_i \to Y_{1,i}$; (b) is by [\ref{ck}, Lemma 7] which gives
\begin{align}
\nonumber \sum_{i=1}^n I(\mathbf{Y}_1^{i-1};Y_{2,i}|W_0,\tilde{\mathbf{Y}}_2^{i+1}) &=\sum_{i=1}^n I(\tilde{\mathbf{Y}}_2^{i+1};Y_{1,i}|W_0,\mathbf{Y}_1^{i-1}), \\
\nonumber \sum_{i=1}^n I(\mathbf{Y}_1^{i-1};Y_{2,i}|W_0,W_1,\tilde{\mathbf{Y}}_2^{i+1}) &=\sum_{i=1}^n I(\tilde{\mathbf{Y}}_2^{i+1};Y_{1,i}|W_0,W_1,\mathbf{Y}_1^{i-1});
\end{align}
(c) is by $I(\tilde{\mathbf{Y}}_2^{i+1};Y_{1,i}|W_0,W_1,\mathbf{Y}_1^{i-1}) = I(\tilde{\mathbf{Y}}_3^{i+1};Y_{1,i}|W_0,W_1,\mathbf{Y}_1^{i-1})=I(U_{3,i};Y_{1,i}|U_{2,i},U_{1,i})=0$ from Relation \ref{relation3} and \eqref{mkv_oi_2}, and also we have $I(W_1;Y_{1,i}|W_0,\mathbf{Y}_1^{i-1},X_i)$ $=$ $0$ from $W_1 \to X_i \to Y_{1,i}$ and $I(\tilde{\mathbf{Y}}_3^{i+1};Y_{1,i}|W_0,W_1,\mathbf{Y}_1^{i-1})$$=0$ from \eqref{mkv_oi_2}.

Now combine \eqref{r0_final1_p} with \eqref{r012_inter2b} to get
\begin{equation}\label{r012_final3}
	n(R_0 + R_1 + R_2) \leq \sum_{i=1}^n \left[ I(U_{1,i};Y_{1,i}) + I(U_{2,i};Y_{2,i}|U_{1,i}) + I(X_i;Y_{1,i}|U_{2,i},U_{3,i}) \right] + n(\gamma_1 + \gamma_5 + \gamma_8).
\end{equation}
Next combine \eqref{r0_final2_p} with \eqref{r012_inter2a}, so that
\begin{align}
	\nonumber & n(R_0 + R_1 + R_2) \\
	\nonumber & \leq \sum_{i=1}^n \left[ I(W_0,\mathbf{Y}_1^{i-1},\tilde{\mathbf{Y}}_3^{i+1};Y_{3,i}) - I(\tilde{\mathbf{Y}}_3^{i+1};Y_{1,i}|W_0,\mathbf{Y}_1^{i-1}) + I(\tilde{\mathbf{Y}}_2^{i+1};Y_{1,i}|W_0,\mathbf{Y}_1^{i-1}) \right. \\
	\nonumber & \;\;\;\; \left. +I(W_1;Y_{2,i}|W_0,\mathbf{Y}_1^{i-1},\tilde{\mathbf{Y}}_2^{i+1})+ I(X_i;Y_{1,i}|W_0,\mathbf{Y}_1^{i-1}) - I(W_1,\tilde{\mathbf{Y}}_3^{i+1};Y_{1,i}|W_0,\mathbf{Y}_1^{i-1}) \right] + n(\gamma_3+\gamma_5 + \gamma_8) \\
	\nonumber & \stackrel{(a)}{=} \sum_{i=1}^n \left[ I(W_0,\mathbf{Y}_1^{i-1},\tilde{\mathbf{Y}}_3^{i+1};Y_{3,i})+I(W_1;Y_{2,i}|W_0,\mathbf{Y}_1^{i-1},\tilde{\mathbf{Y}}_2^{i+1}) \right. \\
	\nonumber & \;\;\;\; \left. + I(X_i;Y_{1,i}|W_0,\mathbf{Y}_1^{i-1}) - I(W_1,\tilde{\mathbf{Y}}_3^{i+1};Y_{1,i}|W_0,\mathbf{Y}_1^{i-1}) \right] + n(\gamma_3+\gamma_5 + \gamma_8) \\
	& = \sum_{i=1}^n \left[ I(U_{3,i};Y_{3,i}) + I(U_{2,i};Y_{2,i}|U_{1,i}) + I(X_i;Y_{1,i}|U_{2,i},U_{3,i}) \right] + n(\gamma_3+\gamma_5 + \gamma_8), \label{r012_final4}
\end{align}
where (a) is due to $I(\tilde{\mathbf{Y}}_3^{i+1};Y_{1,i}|W_0,\mathbf{Y}_1^{i-1})=I(\tilde{\mathbf{Y}}_2^{i+1};Y_{1,i}|W_0,\mathbf{Y}_1^{i-1})$ by Relation \ref{relation1}.

For Case 3, we have
\begin{align}
	\nonumber & n(R_0 + R_1 + R_2) = H(W_0,W_1) + H(W_2|W_0,W_1) \\
	\nonumber & = I(W_0,W_1;\mathbf{Y}_2) + H(W_0,W_1|\mathbf{Y}_2) + I(W_2;\mathbf{Y}_1|W_0,W_1) + H(W_2|W_0,W_1,\mathbf{Y}_1) \\
	\nonumber & \leq \sum_{i=1}^n \left[ I(W_0,W_1;Y_{2,i}|\tilde{\mathbf{Y}}_2^{i+1}) + I(W_2;Y_{1,i}|W_0,W_1,\mathbf{Y}_1^{i-1}) \right] + n(\gamma_5 + \gamma_7) \\
	 & \stackrel{(a)}{\leq} \sum_{i=1}^n \left[ I(W_0,W_1;Y_{2,i}|\tilde{\mathbf{Y}}_2^{i+1}) + I(X_i;Y_{1,i}|W_0,W_1,\mathbf{Y}_1^{i-1}) \right] + n(\gamma_5 + \gamma_7) \label{r012_final5_p} \\
	 & \stackrel{(b)}{\leq} \sum_{i=1}^n \left[ I(U_{2,i};Y_{2,i}) + I(X_i;Y_{1,i}|U_{2,i},U_{3,i})\right] + n(\gamma_5 + \gamma_7), \label{r012_final5}
\end{align}
where (a) is by $W_2 \to X_i \to Y_{1,i}$; and (b) is by following the steps in \eqref{r01_final5} for the first term in the sum of \eqref{r012_final5_p} and the steps from \eqref{r02_inter1}-\eqref{r02_inter2} for the second term in the sum of  \eqref{r012_final5_p}.

Finally, introduce random variable $G$, which is independent of all other random variables and taking on values $i$, for $i = 1, 2,\dots,n$, with probability $1/n$.  Define $U_k\triangleq(G,U_{k,G})$, $X \triangleq X_G$, $Y_k \triangleq Y_{k,G}$, $k=1,2,3$.  Then, we can obtain the rate region in Theorem \ref{thm2} using \eqref{r1_final1}, \eqref{r1_final2}, \eqref{r2_final}, \eqref{r12e_final}, \eqref{r0_final1}, \eqref{r0_final2}, \eqref{r01_final1}, \eqref{r01_final2}, \eqref{r01_final3}, \eqref{r01_final4}, \eqref{r01_final5}, \eqref{r02_final1}, \eqref{r02_final2}, \eqref{r012_final1}, \eqref{r012_final2}, \eqref{r012_final3}, \eqref{r012_final4} and \eqref{r012_final5}.

\subsection{Proof of the outer bound for the 3-receiver BC with 2 degraded message sets (Type 1)}\label{outer_spec1}
In this section we show the proof for the outer bound of Corollary \ref{col2}. The same code construction as in Section \ref{outer_gen}, and preserve the definitions for the auxiliary random variables.

We begin with the equivocation rate $R_{1e}$.  Following the same procedure to obtain \eqref{r1_cnv} - \eqref{r1_initial}, we have
\begin{align}
	\nonumber nR_{1e} & \leq \sum_{i=1}^n \left[I(W_1;Y_{1,i}|W_0,\mathbf{Y}_1^{i-1})
-I(W_1;Y_{3,i} |W_0,\tilde{\mathbf{Y}}_3^{i+1})\right] +n(\epsilon'_1+\gamma_3+\gamma_4) \\
	\nonumber & = \sum_{i=1}^n \left[ I(W_1,\tilde{\mathbf{Y}}_3^{i+1};Y_{1,i}|W_0,\mathbf{Y}_1^{i-1}) - I(W_1,\mathbf{Y}_1^{i-1};Y_{3,i} |W_0,\tilde{\mathbf{Y}}_3^{i+1}) \right.\\
	\nonumber & \;\;\;\; \left.- I(\tilde{\mathbf{Y}}_3^{i+1};Y_{1,i}|W_0,W_1,\mathbf{Y}_1^{i-1}) + I(\mathbf{Y}_1^{i-1};Y_{3,i}|W_0,W_1,\tilde{\mathbf{Y}}_3^{i+1}) \right]+n(\epsilon'_1+\gamma_3+\gamma_4) \\
	\nonumber & \stackrel{(a)}{=} \sum_{i=1}^n \left[ I(W_1,\tilde{\mathbf{Y}}_3^{i+1};Y_{1,i}|W_0,\mathbf{Y}_1^{i-1}) - I(W_1,\mathbf{Y}_1^{i-1};Y_{3,i} |W_0,\tilde{\mathbf{Y}}_3^{i+1}) \right.\\
	\nonumber & \;\;\;\; \left.- I(\tilde{\mathbf{Y}}_3^{i+1};Y_{1,i}|W_0,W_1,\mathbf{Y}_1^{i-1}) + I(\tilde{\mathbf{Y}}_3^{i+1};Y_{1,i}|W_0,W_1,\mathbf{Y}_1^{i-1}) \right]+n(\epsilon'_1+\gamma_3+\gamma_4) \\
	\nonumber & = \sum_{i=1}^n \left[ I(W_1,\tilde{\mathbf{Y}}_3^{i+1};Y_{1,i}|W_0,\mathbf{Y}_1^{i-1}) - I(W_1,\mathbf{Y}_1^{i-1};Y_{3,i} |W_0,\tilde{\mathbf{Y}}_3^{i+1}) \right] +n(\epsilon'_1+\gamma_3+\gamma_4) \\
	\nonumber & = \sum_{i=1}^n \left[ I(X_i,W_1,\tilde{\mathbf{Y}}_3^{i+1};Y_{1,i}|W_0,\mathbf{Y}_1^{i-1}) - I(X_i,W_1,\mathbf{Y}_1^{i-1};Y_{3,i} |W_0,\tilde{\mathbf{Y}}_3^{i+1}) \right. \\
	\nonumber & \;\;\;\; \left. - \left(I(X_i;Y_{1,i}|W_0,W_1,\mathbf{Y}_1^{i-1},\tilde{\mathbf{Y}}_3^{i+1}) - I(X_i;Y_{3,i}|W_0,W_1,\mathbf{Y}_1^{i-1},\tilde{\mathbf{Y}}_3^{i+1}) \right) \right] +n(\epsilon'_1+\gamma_3+\gamma_4) \\
	\nonumber & \stackrel{(b)}{\leq} \sum_{i=1}^n \left[ I(X_i,W_1,\tilde{\mathbf{Y}}_3^{i+1};Y_{1,i}|W_0,\mathbf{Y}_1^{i-1}) - I(X_i,W_1,\mathbf{Y}_1^{i-1};Y_{3,i} |W_0,\tilde{\mathbf{Y}}_3^{i+1}) \right] +n(\epsilon'_1+\gamma_3+\gamma_4) \\
	\nonumber & \stackrel{(c)}{=} \sum_{i=1}^n \left[ I(X_i;Y_{1,i}|W_0,\mathbf{Y}_1^{i-1}) - I(\mathbf{Y}_1^{i-1};Y_{3,i} |W_0,\tilde{\mathbf{Y}}_3^{i+1}) - I(X_i;Y_{3,i}|W_0,\mathbf{Y}_1^{i-1},\tilde{\mathbf{Y}}_3^{i+1}) \right] +n(\epsilon'_1+\gamma_3+\gamma_4) \\
	\nonumber & = \sum_{i=1}^n \left[ I(X_i;Y_{1,i}|W_0,\mathbf{Y}_1^{i-1}) - I(\mathbf{Y}_1^{i-1};Y_{3,i} |W_0,\tilde{\mathbf{Y}}_3^{i+1}) - I(X_i,\tilde{\mathbf{Y}}_3^{i+1};Y_{3,i}|W_0,\mathbf{Y}_1^{i-1}) \right. \\
	\nonumber & \;\;\;\; \left. + I(\tilde{\mathbf{Y}}_3^{i+1};Y_{3,i}|W_0,\mathbf{Y}_1^{i-1}) \right] +n(\epsilon'_1+\gamma_3+\gamma_4) \\
	\nonumber & \stackrel{(d)}{=} \sum_{i=1}^n \left[ I(X_i;Y_{1,i}|W_0,\mathbf{Y}_1^{i-1}) - I(\tilde{\mathbf{Y}}_3^{i+1};Y_{1,i} |W_0,\mathbf{Y}_1^{i-1}) - I(X_i;Y_{3,i}|W_0,\mathbf{Y}_1^{i-1}) \right. \\
	\nonumber & \;\;\;\; \left. + I(\tilde{\mathbf{Y}}_3^{i+1};Y_{3,i}|W_0,\mathbf{Y}_1^{i-1}) \right] +n(\epsilon'_1+\gamma_3+\gamma_4) \\
	\nonumber & \stackrel{(e)}{=} \sum_{i=1}^n \left[ I(X_i;Y_{1,i}|W_0,\mathbf{Y}_1^{i-1}) - I(\tilde{\mathbf{Y}}_2^{i+1};Y_{1,i} |W_0,\mathbf{Y}_1^{i-1}) - I(X_i;Y_{3,i}|W_0,\mathbf{Y}_1^{i-1}) \right. \\
	\nonumber & \;\;\;\; \left. + I(\tilde{\mathbf{Y}}_2^{i+1};Y_{3,i}|W_0,\mathbf{Y}_1^{i-1}) \right] +n(\epsilon'_1+\gamma_3+\gamma_4) \\
	\nonumber & = \sum_{i=1}^n \left[ I(X_i;Y_{1,i}|U_{1,i}) - I(\tilde U_{2,i};Y_{1,i} |U_{1,i}) - I(X_i;Y_{3,i}|U_{1,i}) + I(\tilde U_{2,i};Y_{3,i}|U_{1,i}) \right] +n(\epsilon'_1+\gamma_3+\gamma_4) \\
	 & \stackrel{(f)}{=} \sum_{i=1}^n \left[ I(X_i;Y_{1,i}|U_{1,i}) - I(X_i;Y_{3,i}|U_{1,i}) \right] +n(\epsilon'_1+\gamma_3+\gamma_4) \label{r1e_final_cor}
\end{align}
where (a) is due to [\ref{ck}, Lemma 7] which gives
\[
\sum_{i=1}^n I(\mathbf{Y}_1^{i-1};Y_{3,i}|W_0,W_1,\tilde{\mathbf{Y}}_3^{i+1})=\sum_{i=1}^n I(\tilde{\mathbf{Y}}_3^{i+1};Y_{1,i}|W_0,W_1,\mathbf{Y}_1^{i-1});
\]
(b) is due to the fact that $Y_1$ is a more capable channel than $Y_3$ so that $I(X_i;Y_{1,i}|W_0,W_1,\mathbf{Y}_1^{i-1},\tilde{\mathbf{Y}}_3^{i+1}) - I(X_i;Y_{3,i}|W_0,W_1,\mathbf{Y}_1^{i-1},\tilde{\mathbf{Y}}_3^{i+1}) \geq 0$ and is true as $(W_0,W_1,\mathbf{Y}_1^{i-1},\tilde{\mathbf{Y}}_3^{i+1}) \to X_i \to (Y_{1,i},Y_{3,i})$ forms a Markov chain so that the more capable channel condition is satisfied; (c) is because $I(W_1,\tilde{\mathbf{Y}}_3^{i+1};Y_{1,i}|W_0,\mathbf{Y}_1^{i-1},X_i)=0$ since given $X_i$, $W_1$ and $\tilde{\mathbf{Y}}_3^{i+1}$ are both independent of $Y_{1,i}$ from a functional dependency graph, and we also have $I(W_1;Y_{3,i}|W_0,\mathbf{Y}_1^{i-1},\tilde{\mathbf{Y}}_3^{i+1},X_i)=0$ since we have $W_1 \to X_i \to Y_{1,i}$; (d) has second term in the sum by [\ref{ck}, Lemma 7] by which we have
\[
\sum_{i=1}^n I(\mathbf{Y}_1^{i-1};Y_{3,i} |W_0,\tilde{\mathbf{Y}}_3^{i+1})=\sum_{i=1}^n I(\tilde{\mathbf{Y}}_3^{i+1};Y_{1,i} |W_0,\mathbf{Y}_1^{i-1}),
\]
and third term by $I(\tilde{\mathbf{Y}}_3^{i+1};Y_{3,i}|W_0,\mathbf{Y}_1^{i-1},X_i)=0$ since given $X_i$, $\tilde{\mathbf{Y}}_3^{i+1}$ is independent of $Y_{3,i}$; (e) is by Relation \ref{relation1}; and (f) is by substituting $\tilde U_{2,i}=U_{1,i}$.

For rates $R_0$ we already have, from \eqref{r0_final1}, \eqref{r0_final2} the rates for $W_0$ sent to $Y_1$ and $Y_3$.  For $W_0$ sent to $Y_2$, we have
\begin{align}
	\nonumber nR_0 & \leq \sum_{i=1}^n I(W_0;Y_{2,i}|\tilde{\mathbf{Y}}_2^{i+1}) + n\gamma_2 \\
	\nonumber & = \sum_{i=1}^n \left[ I(W_0,\mathbf{Y}_1^{i-1};Y_{2,i}|\tilde{\mathbf{Y}}_2^{i+1}) - I(\mathbf{Y}_1^{i-1};Y_{2,i}|W_0,\tilde{\mathbf{Y}}_2^{i+1}) \right] + n\gamma_2 \\
	\nonumber & \stackrel{(a)}{\leq} \sum_{i=1}^n \left[ I(W_0,W_1,\mathbf{Y}_1^{i-1},\tilde{\mathbf{Y}}_2^{i+1};Y_{2,i}) - I(W_1;Y_{1,i}|W_0,\mathbf{Y}_1^{i-1}) \right] + n\gamma_2 \\
	\nonumber & = \sum_{i=1}^n \left[ I(W_0,W_1,\mathbf{Y}_1^{i-1};Y_{2,i}) + I(\tilde{\mathbf{Y}}_2^{i+1};Y_{2,i}|W_0,W_1,\mathbf{Y}_1^{i-1}) - I(W_1;Y_{1,i}|W_0,\mathbf{Y}_1^{i-1}) \right] + n\gamma_2 \\
	\nonumber & \stackrel{(b)}{=} \sum_{i=1}^n \left[ I(W_0,W_1,\mathbf{Y}_1^{i-1};Y_{2,i}) - I(W_1;Y_{1,i}|W_0,\mathbf{Y}_1^{i-1}) \right] + n\gamma_2 \\
	& = \sum_{i=1}^n \left[ I(U_{2,i},U_{1,i};Y_{2,i}) - I(U_{2,i};Y_{1,i}|U_{1,i}) \right] + n\gamma_2,
\end{align}
where (a) is due to
\begin{align}
\nonumber & \sum_{i=1}^n I(\mathbf{Y}_1^{i-1};Y_{2,i}|W_0,\tilde{\mathbf{Y}}_2^{i+1}) = \sum_{i=1}^n I(\tilde{\mathbf{Y}}_2^{i+1};Y_{1,i}|W_0,\mathbf{Y}_1^{i-1}) \\
 & =\sum_{i=1}^n I(\tilde{\mathbf{Y}}_3^{i+1};Y_{1,i}|W_0,\mathbf{Y}_1^{i-1}) = \sum_{i=1}^n I(W_1;Y_{1,i}|W_0,\mathbf{Y}_1^{i-1})
\end{align}
with the first equality due to [\ref{ck}, Lemma 7], the second and third equalities by Relations \ref{relation1} and \ref{relation2}, respectively; (b) is due to $I(\tilde{\mathbf{Y}}_2^{i+1};Y_{2,i}|W_0,W_1,\mathbf{Y}_1^{i-1})=I(\tilde{\mathbf{Y}}_3^{i+1};Y_{2,i}|W_0,W_1,\mathbf{Y}_1^{i-1})=0$ by Relation \ref{relation3} and \eqref{mkv_oi_2}.

So we have
\begin{align}
 nR_0 & \leq \sum_{i=1}^n I(U_{1,i};Y_{1,i}) + n\gamma_1, \label{r0_final1_cor1} \\
 nR_0 & \leq \sum_{i=1}^n \left[ I(U_{2,i};Y_{2,i}) - I(U_{2,i};Y_{1,i}|U_{1,i})\right] + n\gamma_2, \label{r0_final2_cor1}\\
 nR_0 & \leq \sum_{i=1}^n \left[ I(U_{3,i};Y_{3,i}) - I(U_{3,i};Y_{1,i}|U_{1,i})\right] + n\gamma_3. \label{r0_final3_cor1}
\end{align}

For rates $R_0 + R_1$, consider $W_0$ sent to $Y_1,Y_2,Y_3$ and $W_1$ to $Y_1$ only. We have
\begin{align}
	\nonumber n(R_0 + R_1) & \leq H(W_0) + \sum_{i=1}^n I(W_1;Y_{1,i}|W_0,\mathbf{Y}_1^{i-1}) + n\gamma_4 \\
	& \stackrel{(a)}{\leq} H(W_0) + \sum_{i=1}^n I(X_i;Y_{1,i}|W_0,\mathbf{Y}_1^{i-1}) + n\gamma_4, \label{r01_inter_cor}
\end{align}
where (a) is by $W_1 \to X_i \to Y_{1,i}$.  Then combining \eqref{r01_inter_cor} with \eqref{r0_final1_cor1}, \eqref{r0_final2_cor1}, \eqref{r0_final3_cor1}, respectively, we can get
\begin{align}
	n(R_0 + R_1) & \leq \sum_{i=1}^n I(X_i;Y_{1,i}) + n(\gamma_1+\gamma_4) \label{r01_final1_cor} \\
	n(R_0 + R_1) & \leq \sum_{i=1}^n \left[ I(U_{2,i};Y_{2,i})+I(X_i;Y_{1,i}|U_{2,i})\right] + n(\gamma_2+\gamma_4) \label{r01_final2_cor} \\
	n(R_0 + R_1) & \leq \sum_{i=1}^n \left[ I(U_{3,i};Y_{3,i})+I(X_i;Y_{1,i}|U_{3,i}) \right] + n(\gamma_3+\gamma_4). \label{r01_final3_cor}
\end{align}

Now introduce the random variables $G$, $X$, $Y_k$, $k=1,2,3$, and $U_k$, $k=1,2$ as at the end of Section \ref{outer_gen}, and using \eqref{r1e_final_cor}, \eqref{r0_final1_cor1}, \eqref{r0_final2_cor1}, \eqref{r0_final3_cor1}, \eqref{r01_final1_cor}, \eqref{r01_final2_cor} and \eqref{r01_final3_cor}, we can obtain the rate region in Corollary \ref{col2}. So we have shown that the 3 degraded message set outer bound can reduce to the 2 degraded message set (Type 1) outer bound, as we have used the same condition ($Y_1$ more capable than $Y_3$), auxiliary random variable definition and code construction so that \eqref{code_construct_3} is satisfied.

\subsection{Proof for the 3-receiver BC with 2 degraded message sets (Type 2) with both $Y_1$ and $Y_2$ less noisy than $Y_3$}\label{outer_spec2}
In this section we show the converse proof for the bound in Corollary \ref{col3}.  We now use a $(2^{nR_0},2^{nR_1},n)$-code with error probability $P_e^{(n)}$ and code construction so that we have the Markov chain condition $(W_0 , W_1)$ $\to \mathbf{X}\to (\mathbf{Y}_1,\mathbf{Y}_2,\mathbf{Y}_3)$.  Then, the probability distribution on ${\cal W}_0 \times {\cal W}_1 \times \mathcal{X}^n \times \mathcal{Y}_1^n \times \mathcal{Y}_2^n \times \mathcal{Y}_3^n$ is given by
\begin{equation}
p(w_0 )p(w_1 )p(\mathbf{x}|w_0 ,w_1) \prod_{i=1}^n p(y_{1i}, y_{2i},y_{3i} | x_i ).
\end{equation}

We first note that from the definition of more capable and less noisy channels \cite{korner_77}, when $Y_1$ is less noisy than $Y_2$ or $Y_3$, then it also follows that $Y_1$ is more capable than $Y_2$ or $Y_3$.

We now also define the new auxiliary random variable $U_i \triangleq(W_0, \mathbf{Y}_3^{i-1})$ satisfying the condition
\begin{equation}
U_i \to X_i \to (Y_{1,i}, Y_{2,i}, Y_{3,i}), \;\;\;\;\; \forall i.
\end{equation}

To proceed with the proof, we begin with the equivocation rates.  We will consider 2 cases: the first, where $W_1$ is sent to $Y_1$, the second where $W_1$ is sent to $Y_2$.  For $W_1$ sent to $Y_1$, we have, following \eqref{r1_cnv}
\begin{align}
\nonumber & nR_{1e} \leq H(W_1 | \mathbf{Y}_3) + n\epsilon'_1 \;\; \mbox{(by secrecy condition)}\\
& \le I(W_1 ; \mathbf{Y}_1 | W_0) - I(W_1 ; \mathbf{Y}_3 | W_0) + n(\epsilon'_1 + \gamma_3 + \gamma_4). \label{r1e_inter_col3}
\end{align}
Then the first two terms of \eqref{r1e_inter_col3} can be bounded as
\begin{align}
	\nonumber & I(W_1 ; \mathbf{Y}_1 | W_0) - I(W_1 ; \mathbf{Y}_3 | W_0) = \sum_{i=1}^n \left[ I(W_1;Y_{1,i}|W_0,\tilde{\mathbf{Y}}_1^{i+1}) - I(W_1;Y_{3,i}|W_0,\mathbf{Y}_3^{i-1}) \right] \\
	\nonumber & = \sum_{i=1}^n \left[ I(W_1,\mathbf{Y}_3^{i-1};Y_{1,i}|W_0,\tilde{\mathbf{Y}}_1^{i+1}) - I(W_1,\tilde{\mathbf{Y}}_1^{i+1};Y_{3,i}|W_0,\mathbf{Y}_3^{i-1}) - I(\mathbf{Y}_3^{i-1};Y_{1,i}|W_0,W_1,\tilde{\mathbf{Y}}_1^{i+1}) \right.\\
	\nonumber & \;\;\;\; \left.+ I(\tilde{\mathbf{Y}}_1^{i+1};Y_{3,i}|W_0,W_1,\mathbf{Y}_3^{i-1}) \right] \\
	\nonumber & \stackrel{(a)}{=} \sum_{i=1}^n \left[ I(W_1,\mathbf{Y}_3^{i-1};Y_{1,i}|W_0,\tilde{\mathbf{Y}}_1^{i+1}) - I(W_1,\tilde{\mathbf{Y}}_1^{i+1};Y_{3,i}|W_0,\mathbf{Y}_3^{i-1}) - I(\tilde{\mathbf{Y}}_1^{i+1};Y_{3,i}|W_0,W_1,\mathbf{Y}_3^{i-1}) \right.\\
	\nonumber & \;\;\;\; \left.+ I(\tilde{\mathbf{Y}}_1^{i+1};Y_{3,i}|W_0,W_1,\mathbf{Y}_3^{i-1}) \right] \\
	\nonumber & = \sum_{i=1}^n \left[ I(W_1,\mathbf{Y}_3^{i-1};Y_{1,i}|W_0,\tilde{\mathbf{Y}}_1^{i+1}) - I(W_1,\tilde{\mathbf{Y}}_1^{i+1};Y_{3,i}|W_0,\mathbf{Y}_3^{i-1}) \right] \\
	\nonumber & = \sum_{i=1}^n \left[ I(X_i,W_1,\mathbf{Y}_3^{i-1};Y_{1,i}|W_0,\tilde{\mathbf{Y}}_1^{i+1}) - I(X_i,W_1,\tilde{\mathbf{Y}}_1^{i+1};Y_{3,i}|W_0,\mathbf{Y}_3^{i-1}) \right. \\
	\nonumber & \;\;\;\; \left. - \left( I(X_i;Y_{1,i}|W_0,W_1,\tilde{\mathbf{Y}}_1^{i+1},\mathbf{Y}_3^{i-1}) - I(X_i;Y_{3,i}|W_0,W_1,\tilde{\mathbf{Y}}_1^{i+1},\mathbf{Y}_3^{i-1}) \right) \right] \\
	\nonumber & \stackrel{(b)}{\leq} \sum_{i=1}^n \left[ I(X_i,W_1,\mathbf{Y}_3^{i-1};Y_{1,i}|W_0,\tilde{\mathbf{Y}}_1^{i+1}) - I(X_i,W_1,\tilde{\mathbf{Y}}_1^{i+1};Y_{3,i}|W_0,\mathbf{Y}_3^{i-1}) \right] \\
	\nonumber & \stackrel{(c)}{=} \sum_{i=1}^n \left[ I(\mathbf{Y}_3^{i-1};Y_{1,i}|W_0,\tilde{\mathbf{Y}}_1^{i+1}) + I(X_i;Y_{1,i}|W_0,\tilde{\mathbf{Y}}_1^{i+1},\mathbf{Y}_3^{i-1}) - I(X_i;Y_{3,i}|W_0,\mathbf{Y}_3^{i-1})  \right] \\
	\nonumber & \stackrel{(d)}{=} \sum_{i=1}^n \left[ I(\tilde{\mathbf{Y}}_1^{i+1};Y_{3,i}|W_0,\mathbf{Y}_3^{i-1}) + I(X_i,\tilde{\mathbf{Y}}_1^{i+1};Y_{1,i}|W_0,\mathbf{Y}_3^{i-1}) - I(\tilde{\mathbf{Y}}_1^{i+1};Y_{1,i}|W_0,\mathbf{Y}_3^{i-1}) \right. \\
	\nonumber & \;\;\;\; \left. - I(X_i;Y_{3,i}|W_0,\mathbf{Y}_3^{i-1}) \right] \\
	\nonumber & \stackrel{(e)}{=} \sum_{i=1}^n \left[ I(X_i;Y_{1,i}|W_0,\mathbf{Y}_3^{i-1}) - I(X_i;Y_{3,i}|W_0,\mathbf{Y}_3^{i-1}) -\left( I(\tilde{\mathbf{Y}}_1^{i+1};Y_{1,i}|W_0,\mathbf{Y}_3^{i-1}) - I(\tilde{\mathbf{Y}}_1^{i+1};Y_{3,i}|W_0,\mathbf{Y}_3^{i-1}) \right) \right] \\
	 & \stackrel{(f)}{\leq} \sum_{i=1}^n \left[ I(X_i;Y_{1,i}|W_0,\mathbf{Y}_3^{i-1}) - I(X_i;Y_{3,i}|W_0,\mathbf{Y}_3^{i-1}) \right]	\label{r1e_inter2_col3}
\end{align}
where (a) is by [\ref{ck}, Lemma 7] from which
\[
\sum_{i=1}^n I(\mathbf{Y}_3^{i-1};Y_{1,i}|W_0,W_1,\tilde{\mathbf{Y}}_1^{i+1})=\sum_{i=1}^n I(\tilde{\mathbf{Y}}_1^{i+1};Y_{3,i}|W_0,W_1,\mathbf{Y}_3^{i-1});
\]
(b) is by $I(X_i;Y_{1,i}|W_0,W_1,\tilde{\mathbf{Y}}_1^{i+1},\mathbf{Y}_3^{i-1}) - I(X_i;Y_{3,i}|W_0,W_1,\tilde{\mathbf{Y}}_1^{i+1},\mathbf{Y}_3^{i-1}) \geq 0$ as $Y_1$ is more capable than $Y_3$ which is a consequence of the assumption that $Y_1$ is less noisy than $Y_3$, with $(W_0,W_1,\tilde{\mathbf{Y}}_1^{i+1},\mathbf{Y}_3^{i-1}) \to X_i \to (Y_{1,i},Y_{3,i})$ fulfilling the more capable channel condition; (c) is because we have $I(W_1;Y_{1,i}|W_0,\tilde{\mathbf{Y}}_1^{i+1},\mathbf{Y}_3^{i-1},X_i)=0$ since $W_1$ is independent of $Y_{1,i}$ given $X_i$ and $I(W_1,\tilde{\mathbf{Y}}_1^{i+1};Y_{3,i}|W_0,\mathbf{Y}_3^{i-1},X_i)=0$ since $W_1$ and $\tilde{\mathbf{Y}}_1^{i+1}$ are both independent of $Y_{3,i}$ given $X_i$, both of which can be verified using a functional dependency graph; (d) has first term by [\ref{ck}, Lemma 7] from which
\[
\sum_{i=1}^n I(\mathbf{Y}_3^{i-1};Y_{1,i}|W_0,\tilde{\mathbf{Y}}_1^{i+1})=\sum_{i=1}^n I(\tilde{\mathbf{Y}}_1^{i+1};Y_{3,i}|W_0,\mathbf{Y}_3^{i-1});
\]
(e) is by $I(\tilde{\mathbf{Y}}_1^{i+1};Y_{1,i}|W_0,\mathbf{Y}_3^{i-1},X_i)=0$ since $\tilde{\mathbf{Y}}_1^{i+1}$ is independent of $Y_{1,i}$ given $X_i$ from a functional dependency graph; and (f) is due to $ I(\tilde{\mathbf{Y}}_1^{i+1};Y_{1,i}|W_0,\mathbf{Y}_3^{i-1}) - I(\tilde{\mathbf{Y}}_1^{i+1};Y_{3,i}|W_0,\mathbf{Y}_3^{i-1})  \geq 0$ from the fact that $Y_1$ is less noisy than $Y_3$.  Thus we have
\begin{align}
nR_{1e} \le \sum_{i=1}^n \left[ I(X_i;Y_{1,i}|U_{i}) - I(X_i;Y_{3,i}|U_{i}) \right] + n(\epsilon'_1 + \gamma_3 + \gamma_4). \label{r1e_final1_col3}
\end{align}

For rate $R_{1e}$ arising from $W_1$ sent to $Y_2$, we follow the same procedure as in \eqref{r1e_inter_col3} to \eqref{r1e_final1_col3}, except that all terms involving $Y_1$ are replaced with the corresponding terms involving $Y_2$, and carry out the expansion $I(W_1;\mathbf{Y}_2|W_0) = \sum_{i=1}^n I(W_1;Y_{2,i}|W_0,\tilde{\mathbf{Y}}_1^{i+1})$ instead, and the condition that $Y_2$ is less noisy than $Y_3$ is used.  Then we can get
\begin{align}
nR_{1e} \le \sum_{i=1}^n \left[ I(X_i;Y_{2,i}|U_{i}) - I(X_i;Y_{3,i}|U_{i}) \right] + n(\epsilon'_1 + \gamma_2 + \gamma_4). \label{r1e_final2_col3}
\end{align}

The rate $R_0$ may be easily found as
\begin{align}
\nonumber nR_0 & = H(W_0) \leq \sum_{i=1}^n I(W_0;Y_{3,i} |\mathbf{Y}_3^{i-1}) + n\gamma_3\\
& \leq \sum_{i=1}^n I(W_0,\mathbf{Y}_3^{i-1};Y_{3,i} ) + n\gamma_3 = \sum_{i=1}^n I(U_i;Y_{3,i} ) + n\gamma_3. \label{r0_final_col3}
\end{align}

For rates $R_0+R_1$, first consider $W_1$ sent to receiver $Y_1$.  We have
\begin{align}
	\nonumber & n(R_0+R_1) = H(W_0) + H(W_1|W_0) \leq H(W_0) + I(W_1;\mathbf{Y_1}|W_0) + n\gamma_4 \\
	\nonumber & =  H(W_0) + \sum_{i=1}^n I(W_1;Y_{1,i}|W_0,\tilde{\mathbf{Y}}_1^{i+1}) + n\gamma_4 \stackrel{(a)}{\leq} H(W_0) + \sum_{i=1}^n I(X_i;Y_{1,i}|W_0,\tilde{\mathbf{Y}}_1^{i+1}) + n\gamma_4 \\
	\nonumber & \leq  H(W_0) + \sum_{i=1}^n I(X_i,\mathbf{Y}_3^{i-1};Y_{1,i}|W_0,\tilde{\mathbf{Y}}_1^{i+1}) + n\gamma_4 \\
	\nonumber & = H(W_0) + \sum_{i=1}^n \left[ I(\mathbf{Y}_3^{i-1};Y_{1,i}|W_0,\tilde{\mathbf{Y}}_1^{i+1}) + I(X_i;Y_{1,i}|W_0,\tilde{\mathbf{Y}}_1^{i+1},\mathbf{Y}_3^{i-1})  \right] + n\gamma_4 \\
	\nonumber & \stackrel{(b)}{=} H(W_0) + \sum_{i=1}^n \left[ I(\tilde{\mathbf{Y}}_1^{i+1};Y_{3,i}|W_0,\mathbf{Y}_3^{i-1}) + I(X_i,\tilde{\mathbf{Y}}_1^{i+1};Y_{1,i}|W_0,\mathbf{Y}_3^{i-1}) - I(\tilde{\mathbf{Y}}_1^{i+1};Y_{1,i}|W_0,\mathbf{Y}_3^{i-1}) \right] + n\gamma_4 \\
	& \stackrel{(c)}{\leq} H(W_0) + \sum_{i=1}^n I(X_i;Y_{1,i}|W_0,\mathbf{Y}_3^{i-1}) + n\gamma_4 \label{r01_inter1_col3}
\end{align}
where (a) is by $W_1 \to X_i \to Y_{1,i}$; (b) is by [\ref{ck}, Lemma 7] from which
\[
\sum_{i=1}^n I(\mathbf{Y}_3^{i-1};Y_{1,i}|W_0,\tilde{\mathbf{Y}}_1^{i+1})=\sum_{i=1}^n I(\tilde{\mathbf{Y}}_1^{i+1};Y_{3,i}|W_0,\mathbf{Y}_3^{i-1});
\]
and (c) is because $I(\tilde{\mathbf{Y}}_1^{i+1};Y_{1,i}|W_0,\mathbf{Y}_3^{i-1})-I(\tilde{\mathbf{Y}}_1^{i+1};Y_{3,i}|W_0,\mathbf{Y}_3^{i-1}) \geq 0$ as $Y_1$ is less noisy than $Y_3$ and $I(\tilde{\mathbf{Y}}_1^{i+1};Y_{1,i}|W_0,\mathbf{Y}_3^{i-1},X_i)=0$ as $\tilde{\mathbf{Y}}_1^{i+1}$ is independent of $Y_{1,i}$ given $X_i$ from a functional dependency graph.

Next, for $W_1$ sent to $Y_2$, again follow the same procedure as to obtain \eqref{r01_inter1_col3}, except that all terms involving $Y_1$ are replaced with the corresponding terms involving $Y_2$, and carry out the expansion $I(W_1;\mathbf{Y}_2|W_0) = \sum_{i=1}^n I(W_1;Y_{2,i}|W_0,\tilde{\mathbf{Y}}_1^{i+1})$, and the condition that $Y_2$ is less noisy than $Y_3$ is used.  As such, we have
\begin{align}
	n(R_0 + R_1) \leq H(W_0) + \sum_{i=1}^n I(X_i;Y_{2,i}|W_0,\mathbf{Y}_3^{i-1}) + n\gamma_5. \label{r01_inter2_col3}
\end{align}

For rates $(R_0 + R_1)$, considering $(W_0,W_1)$ sent to receiver 1, we combine \eqref{r0_final_col3} with \eqref{r01_inter1_col3} to obtain
\begin{align}
\nonumber  n(R_0 + R_1) & \leq \sum_{i=1}^n \left[ I(W_0,\mathbf{Y}_3^{i-1};Y_{3,i} ) + I(X_i;Y_{1,i}|W_0,\mathbf{Y}_3^{i-1})\right] + n(\gamma_3+\gamma_4) \\
\nonumber  & = \sum_{i=1}^n \left[ I(W_0,\mathbf{Y}_3^{i-1};Y_{1,i} )-\left(I(W_0,\mathbf{Y}_3^{i-1};Y_{1,i} )-I(W_0,\mathbf{Y}_3^{i-1};Y_{3,i} )\right)+ I(X_i;Y_{1,i}|W_0,\mathbf{Y}_3^{i-1})\right] \\
\nonumber  & \;\;\;\; + n(\gamma_3 + \gamma_4 )\\
& \stackrel{(a)}{\leq} \sum_{i=1}^n I(X_i;Y_{1,i}) + n(\gamma_3 + \gamma_4 ) \label{r01_final1_col3}
\end{align}
where (a) is by the condition $I(W_0,\mathbf{Y}_3^{i-1};Y_{1,i} )-I(W_0,\mathbf{Y}_3^{i-1};Y_{3,i} ) \geq 0$ from $Y_1$ being less noisy than $Y_3$.  Now considering $(W_0,W_1)$ sent to receiver 2, combine \eqref{r0_final_col3} with \eqref{r01_inter2_col3} in the same way to obtain
\begin{align}
n(R_0 + R_1) & \leq \sum_{i=1}^n I(X_i;Y_{2,i}) + n(\gamma_3 + \gamma_5 ) \label{r01_final2_col3}
\end{align}
where now we have $I(W_0,\mathbf{Y}_3^{i-1};Y_{2,i} )-I(W_0,\mathbf{Y}_3^{i-1};Y_{3,i} ) \geq 0$ from $Y_2$ being less noisy than $Y_3$.

Finally, introduce the random variables $G$, $X$, $Y_k$, $k=1,2,3$, as at the end of Section \ref{outer_gen}, and the random variable $U \triangleq (G, U_G)$.  Using \eqref{r0_final_col3}, \eqref{r1e_final1_col3}, \eqref{r1e_final2_col3}, \eqref{r01_final1_col3} and \eqref{r01_final2_col3}, we obtain the rate region in Corollary \ref{col3}.  Thus we have shown that the outer bound for this 3-receiver, 2 degraded message set (Type 2) channel is a specialization of the more general 3-receiver, 3 degraded message set channel.  We note that the outer bound to the rate equivocation region in Corollary \ref{col3} also coincides with a special case of the achievable bound of Chia and El Gamal \cite{chia_09} stated in Theorem 1 of \cite{chia_09}, for the same message destinations and secrecy conditions.

\section{Conclusion}\label{concl}
Bounds to the rate-equivocation region for the general 3-receiver BC with degraded message sets, in which receiver 3 is a wiretapper receiving the common message, are presented.  This model is a more general model than the 2-receiver BCs with confidential messages with an external wiretapper, and 3-receiver degraded BCs with confidential messages.  We obtain, with secrecy, new inner and outer bounds to the rate-equivocation region for the 3-receiver BC with 3 degraded message sets.  We also obtain, without secrecy, new outer bounds to the rate region for the general 3-receiver BC with 3 degraded message sets.   Lastly, we obtain new inner and outer bounds for rate-equivocation region for the 3-receiver BC with 2 degraded message sets (Type 1).

In the proof of achievability for the inner bound, we used Wyner's code partitioning combined with double-binning for secrecy. We have shown that the proposed coding scheme can provide security for the 3-receiver BC with 3 degraded message sets or 2 degraded message sets (Type 1), although the 2 degraded message set case (Type 1) will suffer a loss in the secrecy rate.  The proof for the outer bound is shown for the 3-receiver BC with 3 degraded message sets and 2 degraded message sets (Type 1) under the condition that receiver 1 is more capable than receiver 3 the wiretapper; and for the 3-receiver BC with 2 degraded message sets (Type 2) for receivers 1 and 2 less noisy than the wiretapper.  The outer bound for the 3 degraded message set case is shown to specialize to the 2 degraded message set (Type 1). Under the condition that both receivers 1 and 2 are less noisy than the wiretapper, the inner and outer bounds for the 3 degraded message case coincide and specialize to the rate-equivocation region of the 3-receiver BC with 2 degraded message sets (Type 2), and to a special case of a 3-receiver BC with 2 degraded message sets (Type 2) which uses a different coding scheme.

\appendix

Here, we show that we can insert an auxiliary random variable $\tilde U_2$, representing information about $W_0$, between $U_1$ and $U_2$, for the 3-receiver BC with 3 degraded message sets. We show that the conditions for correct code generation and low probability of error for decoding are equivalent to those without insertion of $\tilde U_2$ by setting $\tilde U_2 = U_1$. Thus, by their equivalence, we shall subsequently use the code generation process with the insertion of $\tilde U_2$ to facilitate the derivation of the outer bound.

\begin{figure}[!t]
\centering
\includegraphics[scale=0.65]{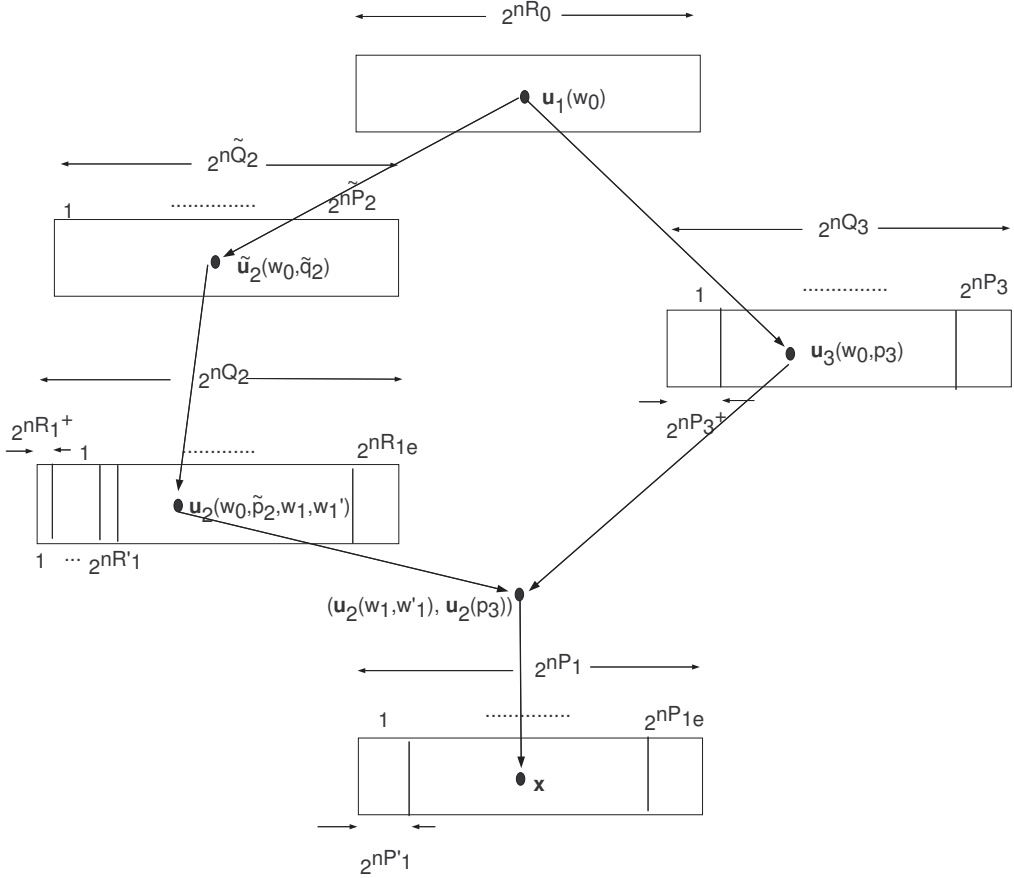}
\caption{Coding for 3-receiver BC with degraded message sets and confidential messages: insertion of auxiliary random variable $\tilde U_2$.}\label{F:Fig2}
\end{figure}

We first note that such an insertion of $\tilde U_2$ gives rise to the Markov chains
\begin{subequations}\label{mkv_appx}
\begin{align}
    & U_1\to \tilde U_2 \to U_2 \to (U_3,X) \to (Y_1,Y_2,Y_3), \\
    & U_1\to U_3\to (\tilde U_2,U_2,X) \to (Y_1,Y_2,Y_3),  \\
    & U_1 \to (\tilde U_2,U_2,U_3) \to X \to (Y_1,Y_2,Y_3).
\end{align}
\end{subequations}

Codebook generation is done as follows: first, generate $2^{nR_0}$ sequences $\mathbf{U}_1(w_0)$.  Then, for each $\mathbf{U}_1 (w_0)$, generate $2^{n \tilde Q_2}$ sequences $\tilde{\mathbf{U}}_2 (w_0, \tilde q_2)$ and partition them into $2^{n \tilde P_2}$ equal-sized bins, and also $2^{nQ_3}$ sequences $\mathbf{U}_3 (w_0, q_3)$. For each $\tilde{\mathbf{U}}_2 (w_0, \tilde p_2)$, generate $2^{nQ_2}$ sequences $\mathbf{U}_2 (w_0,\tilde p_2, q_2)$ and partition them into $2^{n \tilde R_1}$ bins. Also partition the $\mathbf{U}_3(w_0,q_3)$ into $2^{n \tilde P_3}$ equally-sized bins.

Each product bin $(w_1,w_1',p_3)$ contains the joint typical pair $(\mathbf{U}_2(w_0,\tilde p_2,w_1,w_1',w_1^\dag),\mathbf{U}_3(w_0,p_3, p_3^\dag))$ with high probability under the conditions \cite{el_gamal_81}
\begin{equation}\label{bin_rates_appx}
\left\{\begin{aligned}
\tilde P_2 &\le \tilde Q_2, \\
R_{1e} + R_1' &\leq Q_2,\\
P_3 &\leq Q_3,\\
\tilde P_2 + P_3 &\le \tilde Q_2 + Q_3 - I(\tilde U_2;U_3|U_1),\\
\tilde P_2 + R_{1e} + R_1'+P_3&\le \tilde Q_2 + Q_2+Q_3-I(U_2;U_3|U_1).
\end{aligned}\right.
\end{equation}
For the joint typical pair $(\mathbf{U}_2 (w_0,\tilde p_2,w_1 , w_1'), \mathbf{U}_3 (w_0,p_3))$ corresponding to the product bin $(w_1,w_1',p_3)$, generate $2^{n P_1}$ sequences of codewords $\mathbf{X}(w_0,\tilde p_2,w_1,w_1',p_3,p_1,p_1')$.  The decoding follows from what described in Section \ref{ach}. Assume that $(w_0,\tilde p_2,w_1,p_3,p_1) = (1,1,1,1,1)$ is sent. For receiver $Y_1$, joint typical decoding of $\{ \mathbf{u}_1, \tilde{\mathbf{u}}_2, \mathbf{u}_2, \mathbf{u}_3, \mathbf{y}_1\}$ is carried out. We list the error events and the conditions that ensure low error probability when decoding, while noting that the decoding of $\tilde p_2$ and $w_1$ is independent:
\begin{enumerate}
    \item $\Pr \{ \mathtt{E}_1 : (w_0 \neq 1) \} \leq \epsilon$ when
    \begin{align}\label{appx1}
        R_0 + \tilde P_2 + R_{1e} + R_1' + P_{1e} + P_1' + P_3 < I(X;Y_1).
    \end{align}
    \item $\Pr \{ \mathtt{E}_2 : (w_0 = 1, \tilde p_2 \neq 1) \} \leq \epsilon$ when
    \begin{align}
        \tilde P_2 + R_{1e} + R_1' + P_{1e} + P_1' + P_3 < I(X;Y_1|U_1).
    \end{align}
    \item $\Pr \{ \mathtt{E}_3 : (w_0 = 1, \tilde p_2 \neq 1, w_1 \neq 1) \} \leq \epsilon$ when
    \begin{align}
        \tilde P_2 + R_{1e} + R_1' + P_{1e} + P_1' + P_3 < I(X;Y_1|U_1).
    \end{align}
    \item $\Pr \{ \mathtt{E}_4 : (w_0 = 1, \tilde p_2 \neq 1, w_1 \neq 1,p_3 = 1) \} \leq \epsilon$ when
    \begin{align}
        \tilde P_2 + R_{1e} + R_1' + P_{1e} + P_1' < I(X;Y_1|U_3,U_1)=I(X;Y_1|U_3).
    \end{align}
    \item $\Pr \{ \mathtt{E}_5 : (w_0 = 1, \tilde p_2 \neq 1, w_1 = 1,p_3 \neq 1) \} \leq \epsilon$ when
    \begin{align}
        \tilde P_2 + P_{1e} + P_1' + P_3 < I(X;Y_1|U_2,U_1)=I(X;Y_1|U_2).
    \end{align}
    \item $\Pr \{ \mathtt{E}_6 : (w_0 = 1, \tilde p_2 \neq 1, w_1 = 1,p_3= 1,p_1 \neq 1) \} \leq \epsilon$ when
    \begin{align}
        \tilde P_2 + P_{1e} + P_1' < I(X;Y_1|U_2,U_3,U_1)=I(X;Y_1|U_2,U_3).
    \end{align}
    \item $\Pr \{ \mathtt{E}_7 : (w_0 = 1, \tilde p_2 = 1, w_1 \neq 1) \} \leq \epsilon$ when
    \begin{align}
        \tilde R_{1e} + R_1' + P_{1e} + P_1' + P_3 < I(X;Y_1|\tilde U_2,U_1)=I(X;Y_1|\tilde U_2).
    \end{align}
    \item $\Pr \{ \mathtt{E}_8 : (w_0 = 1, \tilde p_2 =1, w_1 \neq 1,p_3 = 1) \} \leq \epsilon$ when
    \begin{align}
        \tilde R_{1e} + R_1' + P_{1e} + P_1' < I(X;Y_1|U_3,\tilde U_2,U_1)=I(X;Y_1|U_3,\tilde U_2).
    \end{align}
    \item $\Pr \{ \mathtt{E}_9 : (w_0 = 1, \tilde p_2 = 1, w_1 = 1,p_3 \neq 1) \} \leq \epsilon$ when
    \begin{align}
        \tilde P_{1e} + P_1' + P_3 < I(X;Y_1|\tilde U_2,U_2)=I(X;Y_1|U_2).
    \end{align}
    \item $\Pr \{ \mathtt{E}_10 : (w_0 = 1, \tilde p_2 = 1, w_1 = 1,p_3= 1,p_1 \neq 1) \} \leq \epsilon$ when
    \begin{align}
        \tilde P_{1e} + P_1' < I(X;Y_1|\tilde U_2,U_2,U_3,U_1)=I(X;Y_1|U_2,U_3).
    \end{align}
\end{enumerate}
Receiver $Y_2$ finds $(w_0,\tilde q_2)$ by indirectly decoding $U_2$, and $w_1$ by decoding $U_2$ conditioned on $(\tilde U_2, U_1)$. As a result, we have the conditions
\begin{align}
R_0 + \tilde Q_2 + Q_2 &< I(U_2;Y_2), \\
Q_2& < I(U_2;Y_2|\tilde U_2,U_1) = I(U_2;Y_2|\tilde U_2).
\end{align}
Receiver $Y_3$ finds $w_0$ by indirectly decoding $U_3$, which has low probability of error under the condition
\begin{equation}\label{appx2}
R_0 + Q_3 < I(U_3;Y_3).
\end{equation}
Compare the above conditions with the conditions for the 3-receiver BC with 3 degraded message sets without insertion of $\tilde U_2$ found in \eqref{bin_rates}, \eqref{pe1_2}, \eqref{pe1_3}, \eqref{pe1_4}, \eqref{pe1_5}, \eqref{pe1_6}, \eqref{pe2_2}, \eqref{pe2_3} and \eqref{pe3}.  By setting $\tilde U_2 = U_1$, the conditions \eqref{bin_rates_appx}, \eqref{appx1}--\eqref{appx2} are maximized. Furthermore, by setting $\tilde P_2 = \tilde Q_2=0$, the conditions \eqref{bin_rates_appx}, \eqref{appx1}--\eqref{appx2} are equivalent to those in \eqref{bin_rates}--\eqref{pe3}. Thus, we may insert $\tilde U_2$ representing information about $W_0$ between $U_1$ and $U_2$ giving the Markov chain conditions \eqref{mkv_appx}, and the conditions on decoding and code generation thus obtained are equivalent to the original conditions with $\tilde U_2 = U_1$. As such, we can derive the outer bound in 2 steps. In the first step, we insert $\tilde U_2$ and use Markov chain conditions \eqref{mkv_appx} to obtain an outer bound $\mathcal{R}'_O$ which is equivalent to the one with original conditions \eqref{mkv} by setting $\tilde U_2 = U_1$. Then, set $\tilde U_2 = U_1$ in $\mathcal{R}'_O$ to obtain $\mathcal{R}_O$.

\end{document}